\documentclass[a4paper,11pt]{article}
\usepackage[margin=1in]{geometry}  

\usepackage{listings}
\usepackage{amsmath,mathtools}
\usepackage{amsthm}
\usepackage{tikz}
\usepackage{caption,subcaption}
\usepackage{array}
\usepackage{mdwmath}
\usepackage{multirow}
\usepackage{mdwtab}
\usepackage{eqparbox}
\usepackage{multicol}
\usepackage{amsfonts}
\usepackage{multirow,bigstrut,threeparttable}
\usepackage{amsthm}
\usepackage{array}
\usepackage{bbm}
\usepackage{subcaption}
\usepackage{epstopdf}
\usepackage{mdwmath}
\usepackage{mdwtab}
\usepackage{eqparbox}
\usepackage{tikz}
\usepackage{latexsym}
\usepackage{amssymb}
\usepackage{bm}
\usepackage{amssymb}
\usepackage{graphicx}
\usepackage{mathrsfs}
\usepackage{epsfig}
\usepackage{psfrag}
\usepackage{setspace}
\usepackage[
CJKbookmarks=true,
bookmarksnumbered=true,
bookmarksopen=true,
colorlinks=true,
citecolor=red,
linkcolor=blue,
anchorcolor=red,
urlcolor=blue
]{hyperref}
\usepackage{cleveref}
\usepackage[ruled]{algorithm2e}
\usepackage{algpseudocode}
\usepackage{stfloats}
\usepackage[
autocite    = superscript,
backend     = bibtex, 
sortcites   = true,
style       = numeric 
]{biblatex} 
\usepackage{nameref}
\usepackage[normalem]{ulem}

\usepackage{soul}

\input xy
\xyoption{all}

\theoremstyle{plain}
\newtheorem{theorem}{Theorem}[section]
\newtheorem{proposition}{Proposition}[section]
\newtheorem{lemma}{Lemma}[section]

\newtheorem{corollary}{Corollary}[section]

\crefname{theorem}{Theorem}{Theorems}
\crefname{proposition}{Proposition}{Propositions}
\crefname{corollary}{Corollary}{Corollaries}

\theoremstyle{definition}
\newtheorem{definition}{Definition}
\newtheorem{remark}{Remark}

\crefname{definition}{Definition}{Definitions}
\crefname{remark}{Remark}{Remarks}

\def\EE{\mathbb{E}}

\def\PP{\mathbb{P}}

\def\RR{\mathbb{R}}

\def\calA{\mathcal{A}}

\def\calD{\mathcal{D}}

\def\calF{\mathcal{F}}

\def\calI{\mathcal{I}}
\def\calJ{\mathcal{J}}

\def\calN{\mathcal{N}}

\def\calS{\mathcal{S}}
\def\calT{\mathcal{T}}

\def\1{\mathbbm{1}}
\def\var{\mathsf{Var}}

\newcommand\independent{\protect\mathpalette{\protect\independenT}{\perp}}
\def\independenT#1#2{\mathrel{\rlap{$#1#2$}\mkern2mu{#1#2}}}

\newcommand{\argmin}{\mathop{\mathrm{argmin}}}

\usepackage{booktabs}
\usepackage{adjustbox}
\usepackage{multirow}
\usepackage{listings}

\theoremstyle{plain}

\def \var {\mathsf{Var}}

\usepackage{xspace}

\def\independenT#1#2{\mathrel{\rlap{$#1#2$}\mkern2mu{#1#2}}}

\definecolor{myblue}{rgb}{.8, .8, 1}
\definecolor{mathblue}{rgb}{0.2472, 0.24, 0.6} 
\definecolor{mathred}{rgb}{0.6, 0.24, 0.442893}
\definecolor{mathyellow}{rgb}{0.6, 0.547014, 0.24}

\usepackage{enumitem}

\newcommand{\No}{{n}}
\newcommand{\NoNull}{{n_0}}
\newcommand{\NoNonNull}{{n_1}}
\newcommand{\NoNc}{m}
\newcommand{\probNull}{{\pi}}

\newcommand{\pval}[1]{{p_{#1}}}
\newcommand{\pvalOrder}[1]{{p_{(#1)}}}
\newcommand{\pvalAll}{{\boldsymbol{p}}}
\newcommand{\qval}[1]{{p_{#1}}}

\newcommand{\qvalAll}{{\boldsymbol{p}}}
\newcommand{\testStatistics}[1]{{T_{#1}}}

\newcommand{\testStatisticsAll}{{\boldsymbol{T}}}

\newcommand{\cdfTestStatistics}[1]{{F_{#1}}}
\newcommand{\cdfTestStatisticsNull}{{F_{0}}}
\newcommand{\cdfTestStatisticsNonNull}{{F_{1}}}

\newcommand{\pdfTestStatistics}[1]{{f_{#1}}}
\newcommand{\pdfTestStatisticsNull}{{f_{0}}}
\newcommand{\pdfTestStatisticsNonNull}{{f_{1}}}

\newcommand{\ncTestStatistics}[1]{{C_{#1}}}

\newcommand{\nickname}{{\text{RANC}}}

\newcommand{\hypothesis}[1]{{H_{#1}}}

\newcommand{\hypothesisIndex}[1]{{\calI_{#1}}}
\newcommand{\nullHypothesis}[1]{{H_{0,#1}}}

\newcommand{\nullHypothesisIndex}{{\calI_{0}}}

\newcommand{\responseTreatment}[1]{{Y_{#1}^{\text{t}}}}
\newcommand{\responseControl}[1]{{Y_{#1}^{\text{c}}}}

\newcommand{\truePositive}{{S}}
\newcommand{\falsePositive}{{V}}
\newcommand{\ncFalsePositive}{{V_{\text{nc}}}}
\newcommand{\trueNegative}{{U}}

\newcommand{\falseNegative}{{T}}
\newcommand{\totalPositive}{{R}}

\newcommand{\FWER}{\text{FWER}}

\newcommand{\FDR}{\text{FDR}}
\newcommand{\FDRLevel}{q}
\newcommand{\FDP}{\text{FDP}}
\newcommand{\localFDR}{\text{local-FDR}}
\newcommand{\LocalFDR}{\text{Local-FDR}}

\newcommand{\PRDS}{{\text{PRDS}}}
\newcommand{\MTPTwo}{{$\text{MTP}_2$}}
\newcommand{\BH}{{\text{BH}}}

\newcommand{\permutationFunction}[1]{{G(#1)}}

\newcommand{\stoppingTime}{{\tau_q}}
\newcommand{\stoppingTimeLambda}{{\tau_q}}

\newcommand{\estimatedLocalFDR}[1]{{\hat{q}(#1)}}
\newcommand{\threshold}{{t}}
\newcommand{\trueThreshold}{{\tau^*_\lambda}}

\newcommand{\EmpiricalThreshold}{{\hat{\tau}_{\lambda, \No, \NoNc}}}

\newcommand{\BayesRisk}{{R}}
\newcommand{\BayesRiskLambda}{{R_{\lambda}}}
\newcommand{\EmpiricalBayesRisk}{{{\ell}_{\No, \NoNc}}}
\newcommand{\EmpiricalBayesRiskLambda}{{{\ell}_{\lambda, \No, \NoNc}}}

\newcommand{\EmpiricalCdfTestStatistics}{{{F}_n}}
\newcommand{\EmpiricalCdfTestStatisticsNull}{{{F}_{0,\NoNc}}}

\newcommand{\pdfBased}{{\text{PDF-based}}}

\newcommand{\cdfBased}{{\text{CDF-based}}}

\newcommand{\weightNull}{{w_0}}
\newcommand{\weightNonNull}{{w_1}}
\newcommand{\level}{{q}}

\title{Simultaneous Hypothesis Testing Using Internal Negative
  Controls with An Application to Proteomics}
\author{Zijun Gao\thanks{Statistical Laboratory, Department of Pure Mathematics and Mathematical Statistics, University of Cambridge, UK. Email: \{zg305\}@cam.ac.uk.} ~ and
  ~Qingyuan Zhao\thanks{Statistical Laboratory, Department of Pure Mathematics and Mathematical Statistics, University of Cambridge, UK. Email: \{qyzhao\}@statslab.cam.ac.uk.}
}

\begin{document}

\maketitle

\begin{abstract}

  Negative control is a common technique in scientific investigations
  and broadly refers to the situation where a null effect (``negative
  result'') is expected. Motivated by a real proteomic dataset, we will
  present three promising and
  closely connected methods of using negative controls to assist
  simultaneous hypothesis testing. The first method uses negative
  controls to construct a permutation p-value for every hypothesis under
  investigation, and we give several sufficient conditions for such
  p-values to be valid and positive regression dependent on the set
  (PRDS) of true nulls. The second method uses negative controls to
  construct an estimate of the false discovery rate (FDR), and we give a
  sufficient condition under which the step-up procedure based on this
  estimate controls the FDR. The third method, derived from an existing
  ad hoc algorithm for proteomic analysis, uses negative controls to
  construct a nonparametric estimator of the local false discovery
  rate. We conclude with some practical suggestions and connections to
  some closely related methods that are propsed recently.
\end{abstract}


\clearpage


\section{Introduction}\label{sec:introduction}

With the rapid development of high-throughput sequencing technologies,
a common task in modern statistical applications is to test a large
number of hypotheses simultaneously. A wealth of multiple testing
procedures have been proposed in the literature; most of them operate
by combining p-values for the individual hypotheses. Prominent
examples include Bonferroni's correction and Simes' test
\parencite{simes1986improved} for family-wise error rate (FWER) control,
the Benjamini-Hochberg (\BH) procedure for false discovery rate (FDR)
control \parencite{benjamini1995controlling}, the closed testing
principle \parencite{marcus1976closed}, and empirical Bayes methods
for controlling the local false discovery rate \localFDR\
\parencite{efron2001empirical}.

In many practical situations, however, the validity of these p-values
may be jeopardized by various problems. For example, many high-throughput
platforms for biological experiments are subject to batch effects
\parencite{leek10_tackl_wides_critic_impac_batch}. Other reasons for
invalid p-values include model misspecification and small sample
sizes. As a consequence, multiple testing procedures that combine
these p-values may fail to control the relevant statistical errors.

In this paper, we employ internal negative controls, of which the null
hypotheses are known to be true, to perform valid simultaneous
hypothesis testing. Motivated by a real proteomic analysis, we present
three closely related methods to use the negative controls. The first
method uses, for the test statistic of each
hypothesis under investigation, its Rank Among the Negative Control
(\nickname) as a permutation p-value; alternatively, this can be
understood as using negative controls to form a (nonparametric)
empirical null distribution. The second method uses the negative
controls to give an estimate of the false discovery rate of a set of
rejected hypotheses. The third method, derived from an existing
ad hoc algorithm for proteomic analysis, uses negative controls to
construct a nonparametric estimator of \localFDR.



\subsection{Motivating example}\label{sec:motivating.data}

Our considerations are motivated by a real proteomic analysis shared
by our collaborating neuroscientists. \textcite{shuster2022situ}
used proteomic profiling to identify
candidate cell membrane proteins that affect dendrite morphogenesis of
Purkinje cells.
To focus on the statistical problem, in this example we will consider
developing cells ($15$ days postnatal) under the treatment
condition (labelled as
HRP+H$_2$O$_2$ in the original paper) and the control condition (HRP
only). For each condition, the
Purkinje cells of one mouse were extracted, cultivated under the assigned
condition, and prepared for mass spectrometry that measures the abundance
of each protein.

In total, $4,753$ proteins were detected and their subcellular
localizations were annotated in the UniProt database.
\textcite{shuster2022situ} were interested in
determining proteins annotated with plasma membrane ($740$ in total)
that show a higher level of expression under the treatment
condition. As there were no biological repeats, the authors ranked the
membrane proteins by the difference in their expression under the
treatment and control conditions, and then used internal negative control
proteins to determine a cutoff
value.\footnote{\textcite{shuster2022situ} referred to internal negative
  control proteins as ``false positives'' and proteins under
  investigation as ``true positives''.} Here, the internal negative
control
proteins are those annotated with nuclear,
mitochondrial, or cytoplasmic but not plasma membrane ($2067$ in
total).\footnote{The
  numbers here are slightly different from
  those reported in \cite{shuster2022situ} due to an update of the
  UniProt knowledgebase.}

\Cref{fig:protein} gives a side-by-side
comparison of the test statistics (treatment-minus-control
expressions) of the proteins under investigation and of the internal negative
controls. The bulk of
the proteins under investigation is approximately normally distributed
and resemble the internal negative control
proteins, but a number of proteins show a much larger difference
compared to the internal negative controls. Our scientific collaborators
informed us that such a pattern is commonly observed in similar
experiments.

One immediate challenge with this dataset is that there is no
biological repeat for each condition. This precludes us from deriving
a null distribution for any single protein based on measurements of
just that protein. One possibile solution is to use an \emph{empirical
  null} distribution. \Cref{fig:hist-pval-theoretical} shows the
histogram of the p-values obtained from a two-sample $t$-test with the
standard error estimated by pooling the
proteins. It is apparent from this plot that the standard error is likely to be overestimated, resulting in unexpected concentration of
p-values around $0.5$. In Section \ref{sec:real.data} we describe more
details of this and other ways to estimate the null distribution,
including the method suggested by \textcite{efron2004large}.

\begin{figure}[tbp]
  \centering
    \includegraphics[width = 0.8\textwidth]{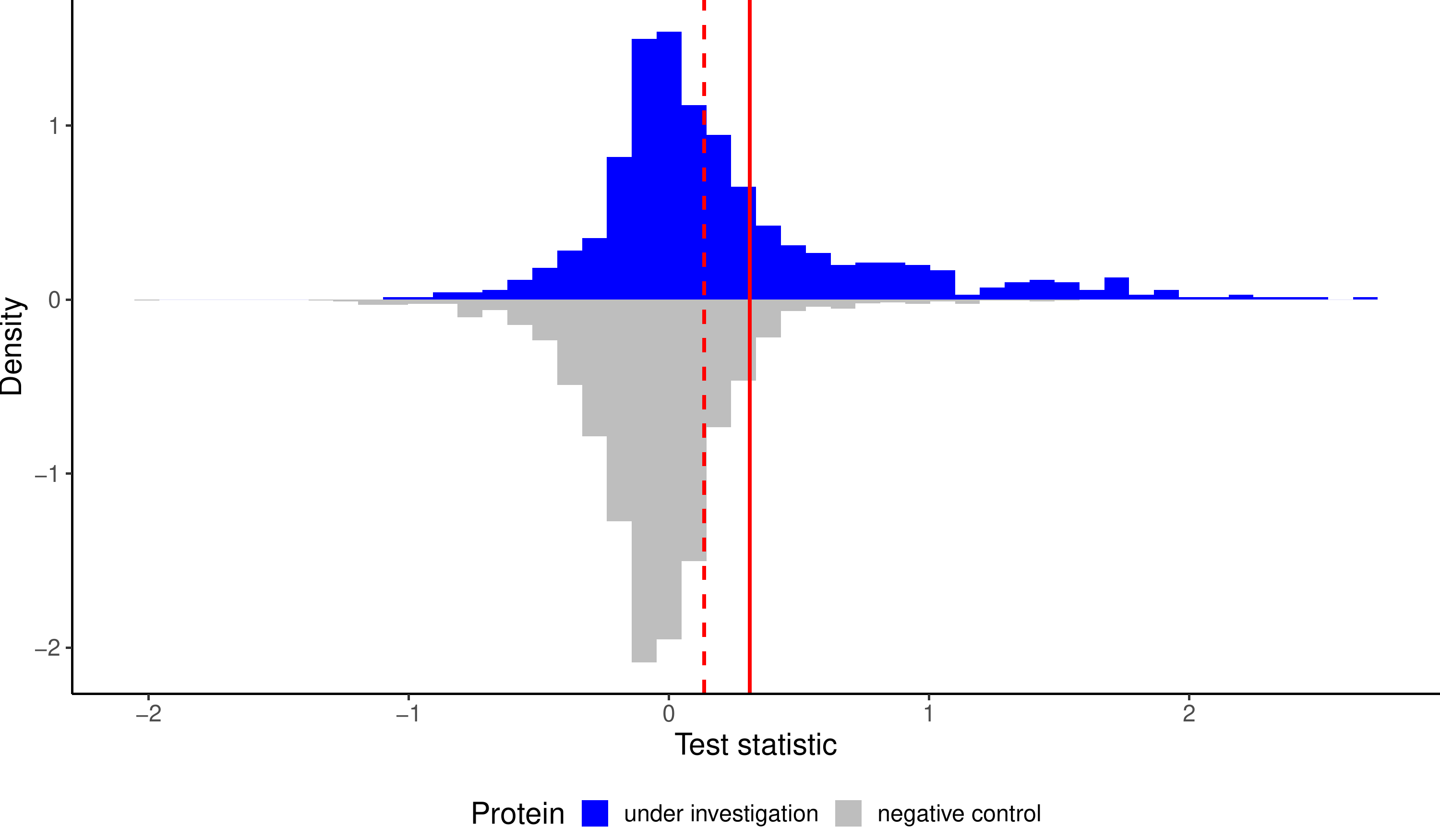}

  \caption{Histogram of test statistics (abundance differences between the
      treatment and control conditions of proteins
      under investigation (blue) and internal negative control
      proteins (grey)). Superimposed are rejection thresholds
      (red) with \FDR~level $q = 0.2$ \localFDR~level $q = \pi$ ($\pi$
      is the proportion of null hypotheses). See \Cref{fig:method} for
    how the thresholds are determined.}
    \label{fig:protein}
\end{figure}

\subsection{Overview of the proposed methods}\label{sec:overview}

Our first method may be viewed as a nonparametric extension of the
methods in the last paragraph. Specifically, we propose to
estimate the null distribution using the empirical distribution of the
negative control test statistics. \Cref{fig:hist-pval-ranc} shows the
histogram of what we call the Rank Among Negative Controls (\nickname)
p-values obtained from this empirical null distribution. Compared to
\Cref{fig:hist-pval-theoretical}, a striking feature of
\Cref{fig:hist-pval-ranc} is that the negative control p-values are
almost uniformly distributed over $[0,1]$. This is expected from this
choice of the empirical null; in fact, these p-values are precisely
equal to $\{1/(\NoNc+1), \dotsc, \NoNc/(\NoNc+1), 1\}$ (assuming no
ties). Among the membrane proteins under investigation, the \nickname\
p-values follow a desirable pattern: their distribution has a spike
near $0$ and is nearly uniform elsewhere. One may apply familiar
multiple testing methods such as the \BH\ procedure to \nickname\
p-values; this will be justified in \Cref{sec:pval}.

\begin{figure}[p]
  \centering
    \begin{subfigure}[t]{0.45\textwidth}
    \centering
    \includegraphics[width =
    \textwidth]{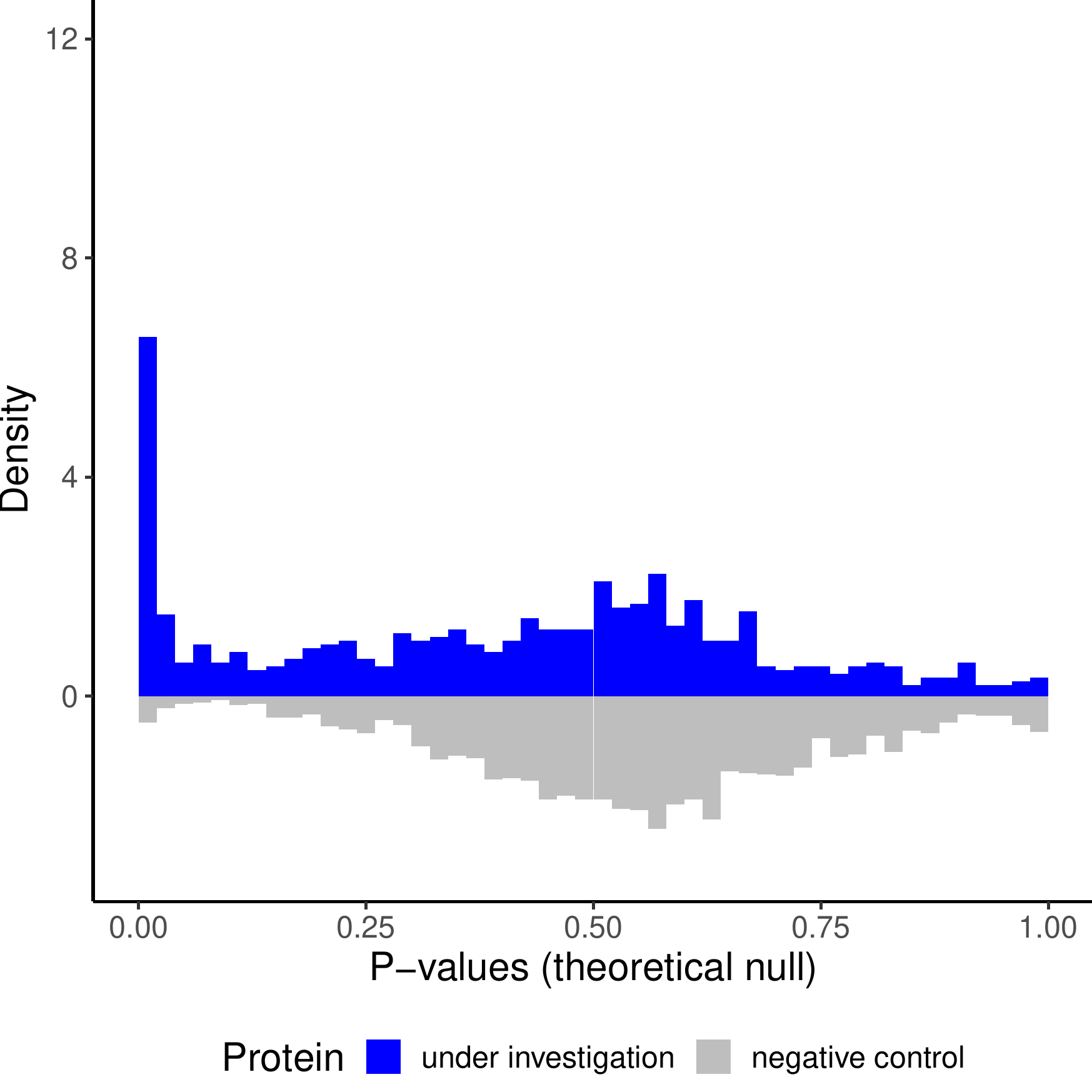}
    \caption{Histogram of p-values calculated from a two-sample
      $t$-test.}
    \label{fig:hist-pval-theoretical}
  \end{subfigure} \quad
  \begin{subfigure}[t]{0.45\textwidth}
    \centering
    \includegraphics[width =
    \textwidth]{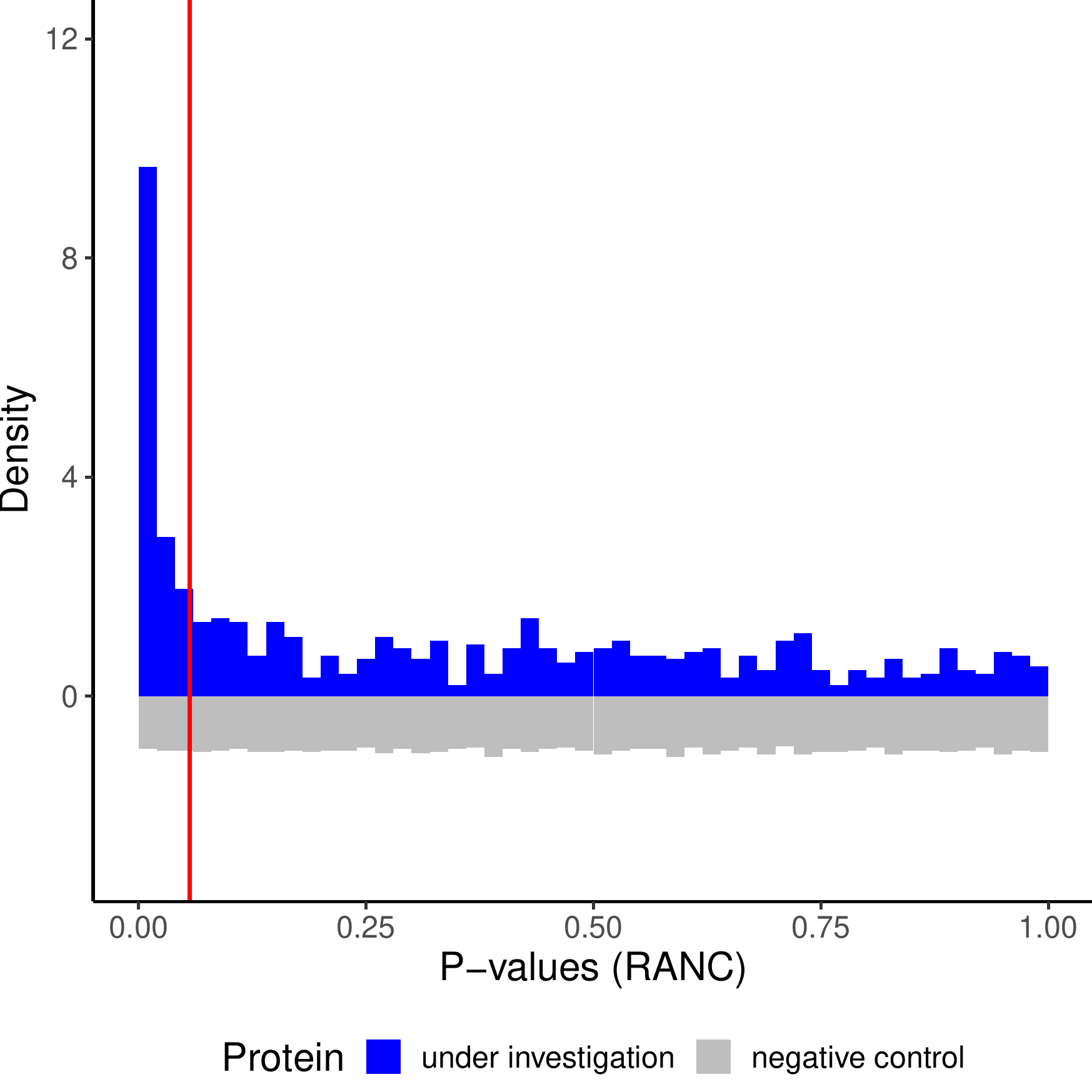}
    \caption{Histogram of p-values calculated using rank among
      negative controls (RANC, method 1). The FDR rejection threshold
      can be determined by applying the Benjamini-Hochberg procedure.}
    \label{fig:hist-pval-ranc}
  \end{subfigure}

  \begin{subfigure}[t]{0.45\textwidth}
    \centering
    \includegraphics[width = \textwidth]{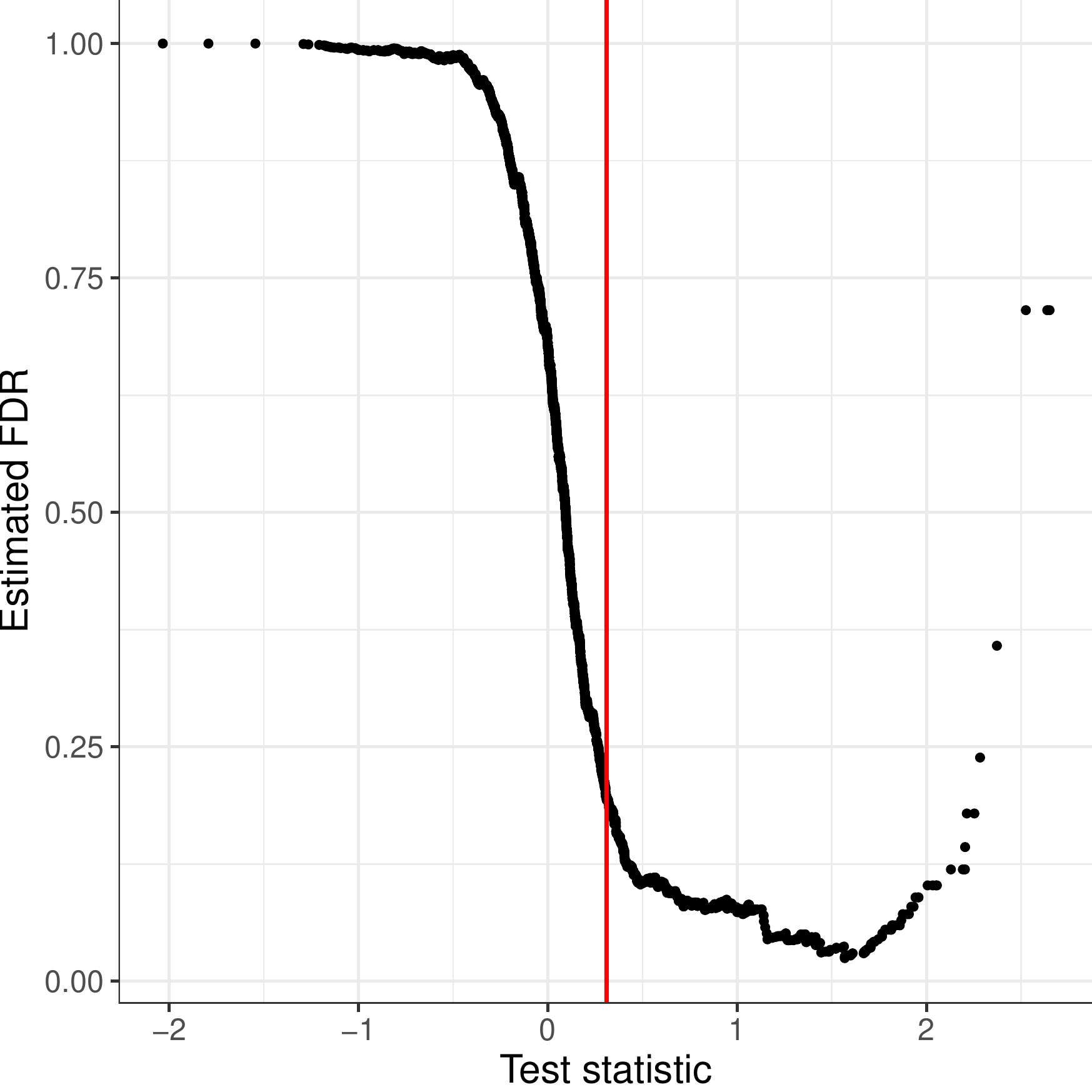}
    \caption{Rejection threshold determined using estimated FDR
      (method 2).}
    \label{fig:estimated-fdr}
  \end{subfigure}
  \quad
  \begin{subfigure}[t]{0.45\textwidth}
    \centering
    \includegraphics[width = \textwidth]{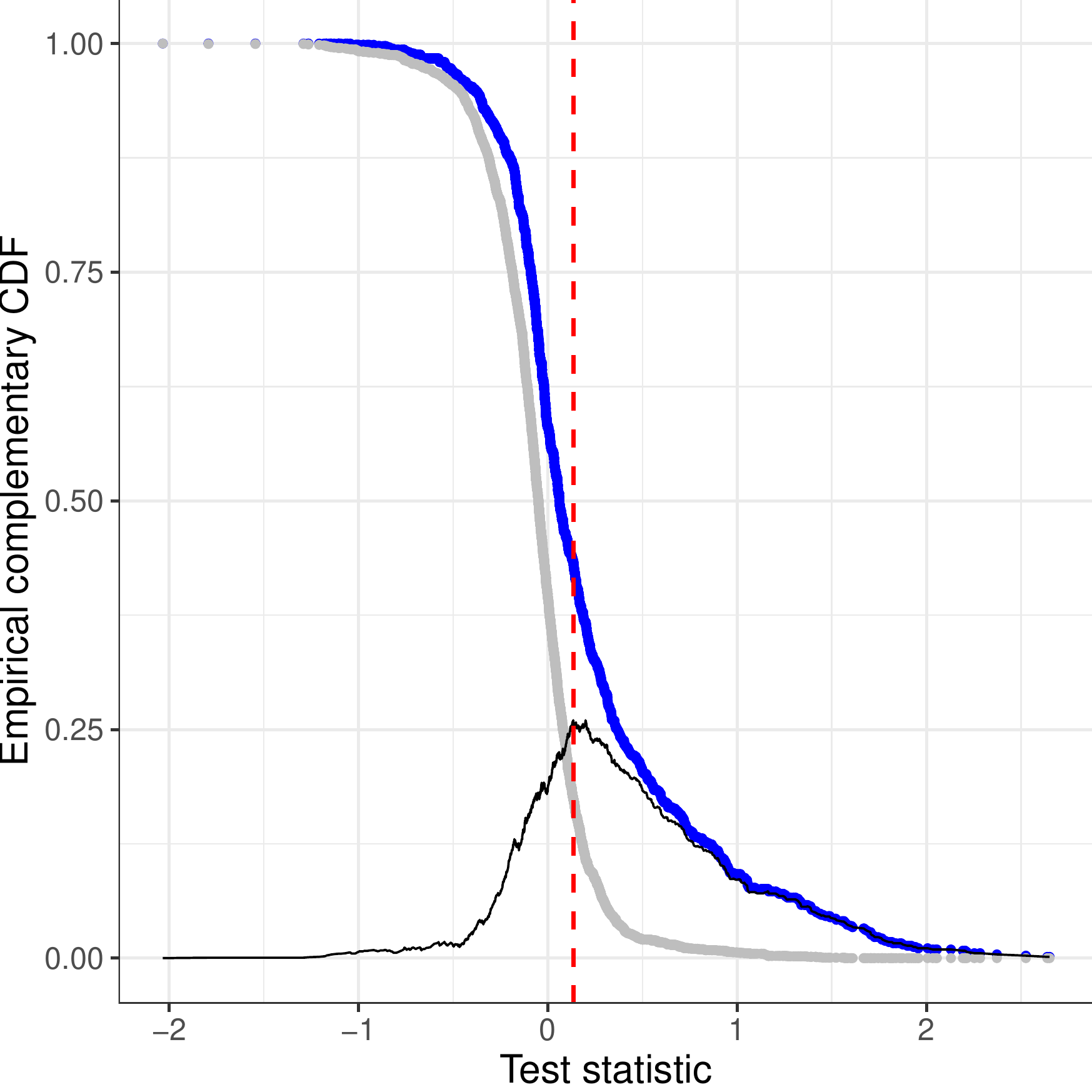}
    \caption{Rejection threshold determined by maximizing the
      difference of empirical complementary CDFs (black curve,
      method 3).}
    \label{fig:lfdr}
  \end{subfigure}

  \caption{An illustration of the methods proposed in this
    article. The red, vertical lines correspond to rejection
    thresholds with FDR level $q=0.2$ (solid) and local
    FDR level $q = \pi$ (dashed, $\pi$ is the proportion of null
    hypotheses).}
  \label{fig:method}

\end{figure}

Our second method is motivated by the empirical process perspective of
\FDR\ control
\parencite{genovese2004stochastic,storey2004strong}. Specifically, we
propose to estimate the \FDR\ above (or below) any rejection threshold
by the ratio of the proportions of test statistics and
negative controls that would be rejected by that
threshold. \Cref{fig:estimated-fdr} illustrates this proposal by
showing the estimated \FDR\ curve for the motivating proteomic
analysis, which can be used to select the rejection threshold for any
targeted \FDR\ level $q$ ($q = 0.2$ is illustrated in the figure).

Our third method is motivated by a rule-of-thumb cut-off analysis in the
original study \parencite{shuster2022situ} proposed by
\textcite{hung14_proteom_mappin_human_mitoc_inter}. This method
computes the empirical (in this case, complementary) cumulative
distribution functions of the test statistics and negative controls
and simply chooses the maximizer of their difference as the rejection
threshold; see \Cref{fig:lfdr} for an illustration. Although this
procedure seems rather ad hoc, the objective that it minimizes can indeed
be viewed as a nonparametric estimate of a weighted mis-classification
risk for multiple testing. Using this heuristic, we show that the
procedure in \textcite{hung14_proteom_mappin_human_mitoc_inter} can be
extended to control the \localFDR~at the rejection threshold.

Shortly after releasing the first preprint of this paper, we
discovered that some recent works, stemming from different motivating
applications, have proposed some almost identical ideas
above. Specifically, our first method (\nickname\ p-value) is the same
as the conformal p-values to test for outliers proposed by
\textcite{bates21_testin_outlier_with_confor_p_values}. Our
second method (empirical estimation of FDR) is essentially the same as
the so-called ``semi-supervised'' multiple testing proposed by
\textcite{mary22_semi_super_multip_testin}, and is closely related to
the analysis of knockoff filters in
\textcite{weinstein17_power_predic_analy_knock_with_lasso_statis};
they are all essentially derived from the martingale argument in
\textcite{storey2004strong}. Our
third method (estimator of \localFDR) is closely related to the method in
\textcite{soloff22_edge_discov} which assumes a known null
distribution. Thus, our paper in effect grounds these recent
ideas in an exploding and diverging literature on a simple and
concrete biological application described next. Moreover, as we are
motivated by biological applications in which negative controls could
easily be misselected, the theoretical
conditions developed below are generally weaker than the i.i.d.\ or
exchangeability conditions used before. We discuss the
connections and differences with these recent works and the broader
literature in more detail below.

\subsection{Literature review}


When most of the null hypotheses are true (i.e.\ $\NoNull/\No$ is close
to $1$) and the theoretical model is correct, the bulk of the test
statistics should be close to the theoretical null
distribution. However, this is often not the case in practice. One
example is the population stratification in genome-wide association
studies, where systematic ancestry difference may distort the
null distribution
\parencite{price06_princ_compon_analy_correc_strat,hellwege17_popul_strat_genet_assoc_studies}. Another
example is batch effect or unwanted variation due to the sequencing
platform
\parencite{leek10_tackl_wides_critic_impac_batch}. \textcite{efron2004large}
argues that if such problems arise, it may be sensible to estimate the null
distribution using the bulk of the empirical distribution of the test
statistics. Other approaches based on more sophisticated models of the
data (such as latent factor models) attempt to estimate the null
distribution empirically by assuming the false hypotheses are sparse
\parencite{leek07_captur_heter_gene_expres_studies,wang17_confoun_adjus_multip_hypot_testin} or by using
negative controls
\parencite{gagnon2012using,wang17_confoun_adjus_multip_hypot_testin}.

Broadly speaking, negative control refers to the situation where a
null effect (``negative result'') is expected. For a given scientific
experiment, there are two types of negative controls: data on internal negative control
units in the same experiment and data from an external control
experiment (e.g.\ the placebo group in a clinical trial). The motivating
example in \Cref{sec:motivating.data} contains both types of negative controls.
Proteins annotated with nuclear, mitochondrial, or cytoplasmic but not
plasma membrane are used as internal negative controls, because the
different
chemical treatments (HRP+H$_2$O$_2$ and HRP only) are not expected to
change their expression levels. The HRP (horseradish peroxidase) only condition is used as an
external control in \textcite{shuster2022situ} because both HRP and H$_2$O$_2$ are needed to tag cell
membrane proteins so that they can be detected by mass spectrometry.


In this paper we focus on internal negative controls, which
essentially represent prior information
about the non-existence of certain causal connections. Such prior
information may come from scientific contexts, experimental
techniques, and previous research studies. In high-throughput
sequencing applications, it is common that the scientific
understanding of the treatment suggests that certain measured units
should not be affected by the treatment. Examples include the
non-membrane proteins in our motivating example \parencite{li2020cell,
  shuster2022situ} and housekeeping genes that are required for the
maintenance of basic cellular function
and thus have a stable level of expression
\parencite{gagnon2012using}. Researchers may also artificially create
internal negative controls by adding units that
should not be affected by the treatment. Examples
include exogenous cells or molecules (often called spike-ins) that do
not interact with the treatment in biological experiments
\parencite{lippa2010exploring} and carefully designed questions in surveys
\parencite{lipsitch2010negative}.

Although internal negative controls have long been used in scientific
investigations, their utility in statistical methodology has only been
explored since recently. In epidemiology, \textcite{lipsitch2010negative}
employed internal negative controls to detect and remove
confounding. This was formally studied by
\textcite{miao18_ident_causal_effec_with_proxy} and is often referred
to as ``proximal causal inference'' in the literature
\parencite{tchetgen20_introd_to_proxim_causal_learn}.
In microarray studies, \textcite{gagnon2012using} used internal negative
control genes to remove unwanted variation, and their
method was analyzed and extended by \textcite{wang2017confounder}.

There exists a small and scattered literature in computational biology and
biostatistics that attempts to use internal negative controls to
control the number of false discoveries. Some authors proposed to fit
a parametric model to the internal negative controls in order to
estimate the null distribution of the test statistics
\cite{nix2008empirical, listgarten2013powerful,
  slattery2011cofactor}. Naturally, such methods are sensitive to the
parametric specification. \textcite{parks2018using}
proposed to estimate the local \FDR~using a kernel density
estimator based on the internal negative controls but did not provide
theoretical justifications of their method. Other authors suggested
heuristic approaches that use the internal negative controls
to estimate the \FDR~\parencite{zhang2008model,
  song2007model}.

Conformal inference seeks distribution-free uncertainty quantification
of the predictions from black-box machine learning models. This
framework, originally developed by Vladimir Vovk and collaborators
\parencite{vovk2005algorithmic}, seeks to make predictive inference
based on how close new observations ``conform'' with the training
data. Conformal inference has received rapidly increasing attention
recently; see \textcite{angelopoulos2021gentle}
for a recent review and some historic notes and
\textcite{zhang2022randomization} for an interpretation from the
viewpoint of classical randomization/permutation tests. Our work is
closely related to conformal inference. In particular, the \nickname\
p-value proposed here may be viewed as a permutation p-value and is
indeed identical to the conformal p-value for outlier detection in
\textcite{bates21_testin_outlier_with_confor_p_values}. See
\textcite{marandon22_machin_learn_meets_false_discov_rate} and
\textcite{liang2022integrative} for some extensions. Another closely
related perspective is to view multiple testing with
negative controls as a semi-supervised problem
\parencite{mary22_semi_super_multip_testin}, a term borrowed from the
machine learning literature. Motivated by problems in astrostatistics,
\textcite{mary22_semi_super_multip_testin} proposed to use
the same p-value using the empirical null viewpoint. They analyzed the
\BH\ procedure applied to such p-values using a martingale argument
and provided some further optimality results. To our knowledge, the
connection between conformal inference, negative controls, and
empirical null in multiple testing has not been well recognized yet.


Last but not least, we brieflly review the literature
concerning \localFDR. \LocalFDR\ was first proposed in
\textcite{efron2001empirical} as a Bayesian alternative to
\FDR. \textcite{sun2007oracle} showed that controlling \localFDR\ is
closely related to minimizing a misclassification loss for the
multiple testing procedure and developed a procedure that controls the
marginal \FDR\ based on a given estimate of the \LocalFDR\
curve. Typically, \LocalFDR\ is estimated by fitting parametric or
semiparametric density models to the data; see, for example,
\textcite{efron2012large}. The recent paper by
\textcite{soloff22_edge_discov} proposed to estimate \localFDR\
nonparametrically using the Grenander's estimator for monotone density
functions and is most closely related to
our third method (and hence the ad hoc procedure in
\textcite{hung14_proteom_mappin_human_mitoc_inter}). The main
difference is that \textcite{soloff22_edge_discov} assumes the null
density function is known, while we use negative controls to estimate
it empirically.



\subsection{Outline and notation}
\label{sec:outline-notation}


In \Cref{sec:setup-terminology}, we describe the mathematical setup of
this paper and review some terminologies in multiple testing.
In \Cref{sec:pval}, we consider the perspective of using negative
controls to form an empirical null distribution and propose the
\nickname~p-values. We give sufficient conditions under which the
\nickname~p-values are valid and satisfy a \PRDS~property, and discuss how the
\nickname~p-values can be combined using various multiple testing
procedures. In \Cref{sec:empirical.process}, we consider the
perspective of using negative controls to form an empirical estimator
of the \FDR. We show the step-up procedure with such as estimtor
can control the \FDR\ even if some negative controls are selected incorrectly.
In \Cref{sec:localFDR}, we develop a method that estimates the
rejection threshold for a given level of \localFDR\ and study its
asymptotic properties. In \Cref{sec:simulation}, we investigate the
performance of \nickname~p-values using numerical simulations.
In \Cref{sec:real.data}, we come back to the motivating proteomic
dataset and investigate different choices of the null distribution. In
\Cref{sec:discussion}, we conclude with some further discussion.

We introduce some mathematical conventions used below.
For $x \in \RR$, we use $\lfloor x \rfloor$ to denote the maximal integer smaller or equal to $x$.
For a set $\calA$, we denote its cardinality by $|\calA|$.
We abbreviate cumulative distribution function as CDF,
probability density function as PDF,
almost everywhere as a.s.,
independently and identically distributed as i.i.d..
We denote the uniform distribution on the interval $[0,1]$ by $U[0,1]$, and the normal distribution with mean $\mu$ and variance $\sigma^2$ by $\calN(\mu, \sigma^2)$.
We use $\Phi(\cdot)$ to denote the CDF of the standard normal distribution.
For two random variables $X$, $Y$, $X \stackrel{d}{=} Y$ means $X$ and $Y$ follow the same distribution.
If $\PP(X \le t) \le \PP(Y \le t)$ for all $t \in \RR$, we say $X$
stochastically dominates $Y$ and denote the relationship by $X \gtrsim
Y$ or $Y \lesssim X$. Given a collection of random variables
$X_i,~i=1,\dotsc,n$, the $k$-th order statistic is defined as its $k$-th
smallest value and is denoted as $X_{(k)}$.

\section{Setup and terminology}
\label{sec:setup-terminology}

Suppose there
are $\No + \NoNc$ null hypotheses. The first $\No$
hypotheses $\hypothesisIndex{} = \{1,\dotsc,\No\}$ are under
investigation. Let $\nullHypothesisIndex$ denote the set of true
null hypotheses and $\NoNull = |\nullHypothesisIndex|$; neither
$\nullHypothesisIndex$ nor $\NoNull$ is known. The last $\NoNc$ hypotheses $\hypothesisIndex{\text{nc}} =
\{\No+1,\dotsc,\No+m\}$ are known to be true and we shall refer to
them as the (internal) negative controls. In our motivating example in
\Cref{sec:motivating.data}, $\No = 740$ and $\NoNc = 2067$. Each
hypothesis $\hypothesis{i}$ is
associated with a test statistic $\testStatistics{i}$, and we assume
that a small test statistic provides evidence against that
hypothesis. For example, $\testStatistics{i}$ can be a p-value for
$\hypothesis{i}$ calculated under some possibly misspecified
statistical model. We are interested in identifying as many non-null
hypotheses in $\hypothesisIndex{} \setminus \nullHypothesisIndex$ as
possible while maintaining control of some multiple testing error.

Next, we briefly review some error rates that are commonly used for
simultaneous hypothesis testing. Given a set of hypotheses
$\{\hypothesis{i}:i \in \calI\}$, let $\totalPositive$ be the total
number of rejections, $\falsePositive$ be the number of incorrect
rejections, and $\truePositive$ be the number of
the correct rejections; see also \Cref{tab:outcome}.
\FWER~is the probability of
making at least one false discovery, \FDP~is the proportion of false
discoveries among all discoveries, and \FDR~is the expectation of \FDP:
\begin{align*}
    \FWER := \PP(\falsePositive \ge 1),\quad\FDP :=
  \frac{\falsePositive}{\totalPositive \vee 1},\quad\text{and}\quad\FDR
    := \EE[\FDP]
    = \EE\left[\frac{\falsePositive}{\totalPositive \vee 1}\right],
\end{align*}
where $R \vee 1$ is the maximum of $R$ and $1$.
We say a multiple testing procedure controls the \FWER~at level
$\alpha$ if $\FWER \leq \alpha$; similarly, a procedure controls the
\FDR~at $q$ if $\FDR \leq q$.
Note that unlike the other two quantities, \FDP~is random and can
only be controlled in some probabilistic sense. For example, we say
a procedure controls the tail probability of \FDP~at
$q$ at level $\alpha$ if $\PP(\FDP > q) \leq \alpha$.
Finally, it has been shown all procedures that control \FWER~or
(the tail probability of) \FDP~can be improved by the closed testing
principle that
combines tests of intersection null hypotheses
\parencite{marcus1976closed,goeman2021only};
an
intersection or global null hypothesis $\cap_{i \in \calI}
\hypothesis{i}$ is said to be true if and only if all individual
hypotheses $\hypothesis{i}$, $i \in \calI$ are true.

When discussing \FWER, \FDR, and \FDP, we will assume that
the CDF of the test statistic $\testStatistics{i}$, denoted by
$\cdfTestStatistics{i}$, is a continuous function for all $i
=1,\dotsc,\No+\NoNc$, and there are no ties (with probability one). We
discuss tiebreakers in \Cref{sec:discussion}.

\begin{table}[tbp]
  \centering
  \caption{{Outcomes in multiple testing.}}
\label{tab:outcome}
\begin{tabular}{c|cc|c}
\toprule
  & Not rejected & Rejected & Total  \\
  \midrule
$\nullHypothesis{i}$ true & $\trueNegative$ & $\falsePositive$  & $\NoNull$ \\
$\nullHypothesis{i}$ false & $\falseNegative$ & $\truePositive$ & $\NoNonNull$ \\ \hline
  Total & $\No - \totalPositive$ & $\totalPositive$ & $\No$  \\
  \bottomrule
\end{tabular}
\end{table}

We will also consider the \localFDR, a Bayesian alternative to
traditional multiple testing criteria. Suppose
$(\hypothesis{i})_{i \in \calI} \overset{\text{i.i.d.}}{\sim}
\text{Bernoulli}(1 - \pi)$, and the
test statistic $\testStatistics{i}$ follows the distribution
\begin{equation*}
  \label{eq:two-mixture}
  \testStatistics{i} \mid \hypothesis{i} = 0 \sim
  \cdfTestStatisticsNull \quad \text{and} \quad \testStatistics{i}
  \mid \hypothesis{i} = 1 \sim \cdfTestStatisticsNonNull.
\end{equation*}
The marginal CDF of $\testStatistics{i}$ is thus given by the mixture
$\cdfTestStatistics{} = \probNull \cdfTestStatisticsNull +
(1-\probNull) \cdfTestStatisticsNonNull$.
When discussing \localFDR, we will assume that the density functions
of $\cdfTestStatisticsNull$,
$\cdfTestStatisticsNonNull$, and $\cdfTestStatistics{}$ exist and
denote them by $\pdfTestStatisticsNull$, $\pdfTestStatisticsNonNull$,
and $\pdfTestStatistics{}$, respectively. The \localFDR~at $t$ is
simply the posterior probability $\Pr(H_i = 0 \mid T_i = t)$ of $H_i$
being a true null given its test statistic is
$T_i=t$ \parencite{efron2001empirical}, i.e.\
\begin{align} \label{eq:local-fdr}
  \localFDR(t)
    = \frac{\probNull \pdfTestStatisticsNull(t)}{\probNull
  \pdfTestStatisticsNull(t) + (1-\probNull)
  \pdfTestStatisticsNonNull(t)},
\end{align}
We say a procedure controls \localFDR~at level $q$ if
$\localFDR(\testStatistics{i}) \le q$ for all rejected hypotheses
$\hypothesis{i}$.

\section{The empirical null perspective}\label{sec:pval}

\subsection{Rank among negative controls}
\label{sec:ranc-p-values}

Our first method uses negative control statistics to form a
nonparametric estimator of the null distribution. Specifically, we
define the \nickname~p-value for $\hypothesis{i},~i \in
\hypothesisIndex{}$ as $\pval{i} = \hat{F}(\testStatistics{i})$, where
\begin{align}\label{defi:nc.ecdf.2}
  \hat{F}(t)
  = \frac{1 + \sum_{j \in \hypothesisIndex{\text{nc}}}
  \1_{\{\testStatistics{j} \le
  t\}}}{1 + \NoNc}
\end{align}
is the empirical cumulative distribution function (CDF) of
$(-\infty,\testStatistics{\No+1},\dotsc,\testStatistics{\No+\NoNc})$.
Here, we include a $-\infty$ in the definition of
$\hat{F}$ to ensure that the \nickname~p-value does not equal zero with
a positive probability.
Another way to put this is that $\pval{i}$ is simply the
normalized rank of $\testStatistics{i}$ among
$(\testStatistics{j})_{j \in \{i\} \cup \hypothesisIndex{\text{nc}}}$:
\begin{align}\label{eq:pval}
  \pval{i} := \frac{1 + \sum_{j \in \hypothesisIndex{\text{nc}}}
  \1_{\{\testStatistics{j} \le \testStatistics{i}\}}}{1 + \NoNc} =
  \frac{1 + (\text{number of negative control statistics
  }\leq T_i)}{1 + (\text{number of negative control statistics})}.
\end{align}
This is why we call $\pval{i}$ the Rank Among
Negative Control (RANC) p-value.

If the negative control statistics resemble the test statistics under
the null, we expect that $\hat{F}$ to be close to the null distribution
of the test statistics and $\pval{i} = \hat{F}(\testStatistics{i})$ to
be approximately uniformly distributed when $\NoNc \to \infty$. If we
further assume that
\begin{equation}
  \label{eq:basic-exchangeable}
  \text{the collection of random variables}~(\testStatistics{i})_{i\in
    \hypothesisIndex{0} \cup \hypothesisIndex{\text{nc}}} ~\text{is
    exchangeable},
\end{equation}
then $\pval{i}$ is exactly the p-value of the permutation test of
exchangeability using just $T_i$. Thus, the proposed \nickname~p-value
is valid under \eqref{eq:basic-exchangeable} in the sense
that $\PP(\pval{i} \leq \alpha) \leq \alpha~\text{for all}~0 < \alpha
< 1$.

Our goal in the rest of this section is to give a more precise,
non-asymptotic analysis of the \nickname~p-values. In
particular, we will give sufficient conditions under which the
\nickname~p-values are individually valid and satisfy a PRDS
property. To this end, we first give a formal definition of
exchangeability and PRDS on a subset of random variables.

\begin{definition}
  We say a sequence of random variables $(X_i)_{i \in \calI}$ is
  exchangeable on a subset $(X_i)_{i \in \calJ}$ for some $\calJ
  \subseteq \calI$, if for any
  permutation $g:\calI \to \calI$ such that $g(i) = i$ for
  all $i \not \in \calJ$, the distribution of $(X_{g(i)})_{i \in
    \calI}$ is the same as $(X_i)_{i \in \calI}$. When this holds for
  $\calJ = \calI$, we simply say the sequence $(X_i)_{i \in \calI}$ is exchangeable.
\end{definition}

\begin{remark}
  Note that this is equivalent to assuming that $(X_i)_{i \in \calJ}$
  is exchangeable conditionally on $(X_i)_{i \in \calI \setminus \calJ}$. We
  introduce this new terminology of exchangeability on a subset to
  contrast with the definition of PRDS below.
\end{remark}

To define PRDS, we say a set
$\calD \subseteq \RR^n$ is increasing if $\calD$ contains all $y
\in \RR^n$ that satisfies $y_i \ge x_i,~1 \le i \le n$ for some $x
\in \calD$.

\begin{definition}\label{defi:PRDS}
    We say a sequence of random variables $(X_i)_{i \in \calI}$ exhibits
    Positive Regression Dependence on a Subset (PRDS) $(X_j)_{j \in
      \calJ}$ for some $\calJ
    \subseteq \calI$, if for any increasing set $\calD
    \subseteq \RR^{|\calI|}$ and any $j \in \calJ$, the conditional
    probability $\PP\big((X_i)_{i \in \calI} \in \calD \mid X_j = x\big)$ is
    increasing in $x$. When this holds for $\calJ = \calI$, we simply
    say the sequence $(X_i)_{i \in \calI}$ is PRD.
\end{definition}

It follows from the definition that PRDS is preserved by co-monotone
transformtions: given some monotonically increasing (or decreasing)
functions $G_i$ for $i \in
\calI$, the assumption that $(X_i)_{i \in \calI}$ is PRDS implies that
$(G_i(X_i))_{i \in \calI}$ is also PRDS.

\subsection{Validity}\label{sec:validity}


\begin{proposition}\label{prop:pvalue.validity}
Fix a true null hypothesis $\hypothesis{i}$ for some $i \in
\nullHypothesisIndex$ and suppose the following assumptions are
satisfied:
\begin{enumerate}[label = (\alph*), ref = (\alph*)]
    \item \label{vali:assu:null.nc.conservative.general}
      $\cdfTestStatistics{i}(t) \le \cdfTestStatistics{j}(t)$ for all $j
      \in \hypothesisIndex{\text{nc}}$ and $t \in \RR$;
    \item \label{vali:assu:one.exchangeable.general}
      $(\cdfTestStatistics{j}(\testStatistics{j}))_{j \in \{i\} \cup
        \hypothesisIndex{\text{nc}}}$ is exchangeable.
\end{enumerate}
Then the \nickname~p-value $\pval{i}$ is valid in the sense
that $\PP(\pval{i} \leq \alpha) \leq \alpha$ for all $0 < \alpha < 1$.
\end{proposition}




A proof of \Cref{prop:pvalue.validity} can be found in the
Appendix. It uses a monotone coupling argument and the fact that the
rank of exchangeable random variables is uniformly distributed.

\begin{remark}
The conditions in \Cref{prop:pvalue.validity} are stated in terms of
the probability integral transforms of the test statistics and are weaker
than the exchangeability in \eqref{eq:basic-exchangeable} in two
ways. First, as we are only
concerned with the validity of $\pval{i}$ for some fixed $i$, it is
only necessary to
assume that $\testStatistics{i}$ is exchangeable with the negative
control statistics. Second, the null statistic
$\testStatistics{i}$ is allowed to be stochastically larger than the
internal negative control statistic $\testStatistics{j}$ for all $j\in
\hypothesisIndex{\text{nc}}$. This relaxation is useful for
testing one-sided hypotheses; see \Cref{rem:one-sided} below.
The exchangeability
condition~\ref{vali:assu:one.exchangeable.general} is
satisfied when the transformed test statistics are i.i.d.\ or follow a
mixture of i.i.d.\ distributions.
\end{remark}


\subsection{\PRDS}\label{sec:PRDS}

Because the \nickname~p-values $\pval{i} =
\hat{F}(\testStatistics{i}),~i=1,\dotsc,\No$ are calculated using
the same empirical null distribution
$\hat{F}(\cdot)$, they are generally not
independent even if the test satistics
$\testStatistics{i},~i=1,\dotsc,\No$ are independent. However, it can
be shown that the \nickname~p-values may satisfy a desirable
\PRDS~property that is sufficient for the validity of many multiple
hypothesis testing procedures
\parencite{sarkar1997simes,benjamini2001control}.


\begin{theorem}\label{prop:PRDS}
Suppose one of the two sets of conditions holds:
\begin{enumerate}[label=\roman*., ref=(\roman*)]
  \item \label{PRDS:assu:independent}
    \begin{enumerate}[label = (\alph*),ref = (i.\alph*)]
    \item \label{PRDS:assu:null.nc.identical} $\testStatistics{i}
      \stackrel{d}{=} \testStatistics{j}$ for any $i \in
      \nullHypothesisIndex$ and $j \in \hypothesisIndex{\text{nc}}$;
    \item \label{PRDS:assu:all.nc.set.independent}
      $\left( \testStatistics{i}\right)_{i \in \hypothesisIndex{}}
      \independent \left( \testStatistics{j} \right)_{j \in \hypothesisIndex{\text{nc}}}$;
    \item \label{PRDS:assu:nc} $\left( \testStatistics{j}  \right)_{j
        \in \hypothesisIndex{\text{nc}}}$ is mutually independent;
    \item \label{PRDS:assu:PRDS} $\left( \testStatistics{i} \right)_{i
        \in \hypothesisIndex{}}$ is \PRDS~on $\left(
        \testStatistics{i}  \right)_{i \in
        \nullHypothesisIndex}$;
    \end{enumerate}
    \item \label{PRDS:assu:exchangeable} $\left( \testStatistics{i}
      \right)_{i \in \hypothesisIndex{} \cup
        \hypothesisIndex{\text{nc}}}$ is exchangeable on $\left( \testStatistics{i}
      \right)_{i \in \hypothesisIndex{0} \cup
        \hypothesisIndex{\text{nc}}}$.
\end{enumerate}
Then the \nickname~p-values are valid and $(\pval{i})_{i \in
  \hypothesisIndex{}}$ is \PRDS~on $(\pval{i})_{i \in
  \nullHypothesisIndex}$.
\end{theorem}


The validity directly follows from \Cref{prop:pvalue.validity}. Our
proof of the \PRDS~property is more involved and is based on
the following heuristic: if we swap any $T_i, i \in
\hypothesisIndex{0}$ with the next smallest negative control
statistic, the probability that $\bm p \in \mathcal{D}$ for any
increasing set $\mathcal{D}$ can only increase. See the Appendix for
more detail.

\begin{remark} \label{rem:bates-condition}
  \textcite[thm.\ 2]{bates21_testin_outlier_with_confor_p_values}
  stated and proved the PRDS property in \Cref{prop:PRDS} under the
  assumption that $\left(
    \testStatistics{i}  \right)_{i \in \hypothesisIndex{} \cup
    \hypothesisIndex{\text{nc}}}$ is mutually
  independent,\footnote{\textcite[thm.\
    2]{bates21_testin_outlier_with_confor_p_values} does make any
    assumption on non-null statistics. We believe this is most likely
    a typo; see \Cref{sec:bates-typo} in the Appendix.} which
  implies the partial/conditional exchangeability condition
  \ref{PRDS:assu:exchangeable}. The set of
  conditions in \ref{PRDS:assu:independent}, especially the PRDS condition
  \ref{PRDS:assu:PRDS} on the original test statistics, appears to be
  novel and may be quite useful when the test statistics are
  positively dependent.
\end{remark}

\begin{remark}
    To our knowledge, the conclusion of \Cref{prop:PRDS} does not
  directly follow from any existing results about positively dependent
  distributions. First, the PRDS property of $(\pval{i})_{i \in
  \hypothesisIndex{}}$ in the conclusion of \Cref{prop:PRDS} does not
immediately follow from condition~\ref{PRDS:assu:PRDS}---the same PRDS
property for $(\testStatistics{i})_{i \in \hypothesisIndex{}}$---by applying a
co-monotone transformation. This is because the transformation,
defined by the internal negative control statistics, is random. Second, it is tempting to treat
$(\testStatistics{i})_{i \in \hypothesisIndex{\text{nc}}}$ as latent
and apply sufficient conditions for \PRDS~in latent variable
models. However, existing results either assumes a single latent variable
\parencite{benjamini2001control}, considers only binary random
variables \parencite{holland1986conditional}, or requires the
\MTPTwo~property in \textcite{karlin1980classes} that is not implied
by the conditions in \Cref{prop:PRDS}.
\end{remark}

\begin{remark} \label{rem:ia-cannot-relax}
Condition~\ref{PRDS:assu:null.nc.identical} in
\Cref{prop:PRDS} cannot be relaxed to the
stochastic dominance
condition~\ref{vali:assu:null.nc.conservative.general} in
\Cref{prop:pvalue.validity}; see \Cref{sec:example.PRDS} for a
counter-example.
\end{remark}

\subsection{Multiple testing with
  \nickname~p-values}\label{sec:multiple.testing}

Given the conclusions in \Cref{prop:pvalue.validity,prop:PRDS}, we
briefly discuss the multiple testing procedures that can be
applied to \nickname~p-values.

\subsubsection{Testing an intersection null}\label{sec:global.null}


We first consider Bonferroni's test and Simes' test of an intersection
null $\hypothesis{\calI} = \cap_{i=1}^n \hypothesis{i}$ at level
$\alpha$. Suppose each hypothesis $\hypothesis{i}$ is associated with
a valid p-value $\pval{i}$. Let $\pval{(1)} \leq \dotsb \leq
\pval{(n)}$ be the ordered p-values. Bonferroni's test rejects
$\hypothesis{\calI}$ if $\pval{(1)} \leq \alpha / n$ and controls the type I
error at $\alpha$ as long as the individual p-values are valid. Thus,
when applied to the \nickname~p-values, Bonferroni's test is valid if the
conditions in \Cref{prop:pvalue.validity} are satisfied for
all $i$.

Simes' test rejects
$\hypothesis{\calI}$ if $\pvalOrder{i} \le i  \alpha/\No$ for some $1 \le i \le n$
\cite{simes1986improved}. Obviously, Simes' test rejects
$\hypothesis{\calI}$ whenever Bonferroni's test rejects
$\hypothesis{\calI}$. Simes' test has been shown to be valid if the
p-values are PRD under $H$ \cite{sarkar1997simes}. This property and
\Cref{prop:PRDS} lead to the following result.

\begin{proposition}\label{prop:simes}
Suppose $\cdfTestStatistics{i}(t) \le \cdfTestStatistics{j}(t)$ for
all $i \in \hypothesisIndex{0}$, $j \in \hypothesisIndex{\text{nc}}$,
and $t \in \RR$,
and one of the following sets of conditions holds,
\begin{enumerate}[label = \roman*.,ref = (\roman*)]
    \item \label{simes:assu:PRDS}
      \ref{PRDS:assu:all.nc.set.independent},
      \ref{PRDS:assu:nc}, and \ref{PRDS:assu:PRDS} in
      \Cref{prop:PRDS};
    \item \label{simes:assu:exchangeable} $\left( \cdfTestStatistics{i}(\testStatistics{i})
      \right)_{i \in \hypothesisIndex{0} \cup
        \hypothesisIndex{\text{nc}}}$ is exchangeable.
\end{enumerate}
Then Simes' test applied to the \nickname~p-values controls the type I error
for testing the intersection null $\hypothesis{\calI}$.
\end{proposition}

Apart from a relaxation of condition
\ref{PRDS:assu:null.nc.identical} in \Cref{prop:PRDS} to stochastic
dominance, the conditions in \Cref{prop:simes} are the same as those
in \Cref{prop:PRDS} (note that $\hypothesisIndex{} =
\hypothesisIndex{0}$ if $\hypothesis{}$ is true). We cannot
use \Cref{prop:PRDS} directly to prove
\Cref{prop:simes} (see \Cref{rem:ia-cannot-relax} above). However,
\Cref{prop:simes} suggests that the (infeasible) \nickname~p-values
$(\tilde{\pval{i}})_{i \in \hypothesisIndex{}}$, obtained
from test statistics $(\cdfTestStatistics{i}(\testStatistics{i}))_{i \in
    \hypothesisIndex{}}$ and internal negative control statistics
  $(\cdfTestStatistics{i}(\testStatistics{i}))_{i \in
    \hypothesisIndex{\text{nc}}}$, are PRD under $H$. So Simes' test
  applied to $(\tilde{\pval{i}})_{i \in
    \hypothesisIndex{}}$ controls the type I error. Under the
  assumptions in \Cref{prop:simes}, we have $\pval{i} \geq
  \tilde{\pval{i}}$ for all $i \in \hypothesisIndex{}$. Thus, if Simes'
  test applied to $(\pval{i})_{i \in \hypothesisIndex{}}$ rejects $H$,
  it must also reject $H$ when applied to $(\tilde{\pval{i}})_{i \in
    \hypothesisIndex{}}$. The conclusion in \Cref{prop:simes}
  immediately follows.

There are many other global tests besides Bonferroni's correction and Simes'
test; a prominent example is Fisher's combination test that requires
independent p-values \cite{fisher1925statistical}. When applied to
\nickname~p-values, however, such methods may not always control the
type I error; in the case of Fisher's test, see
\Cref{sec:example.fisher} for a counter-example and
\textcite[sec.\ 2.1]{bates21_testin_outlier_with_confor_p_values} for
a theoretical characterization of this negative result.
Nevertheless, by assuming the exchangeability in
\eqref{eq:basic-exchangeable}, the distribution of any test statistic
under the global null
can be obtained by using permutations. \nickname~p-values come up
naturally in such permutation tests: \Cref{prop:permutation} in the
Appendix shows that all permutation tests that are invariant under
monotone transformations and certain permutations can be written as a
function of the \nickname~p-values.

\begin{remark}
It may be interesting to compare Simes' test applied to $(\pval{i})_{i
  \in \hypothesisIndex{}}$ with the permutation test applied to the
statistic $\min_{i \in \hypothesisIndex{}} \pval{(i)}/i$. As both
tests use the same test statistic and the permutation test is exact,
Simes' test can be viewed as a conservative approximation to the
permutation test with a simple rejection threshold. Numerical
simulations shows that this approximation
becomes more accurate when $\NoNc$ is much larger than
$\No$ (\Cref{fig:Simes.test}). Heuristically, this is because Simes' test is exact when the
p-values are independent, and the dependence of the \nickname~p-values
decreases as the number of negative controls increases.
\end{remark}

\begin{figure}[tbp]
    \centering
    \begin{subfigure}[b]{0.4\textwidth}
    \centering
\includegraphics[width  = \textwidth]{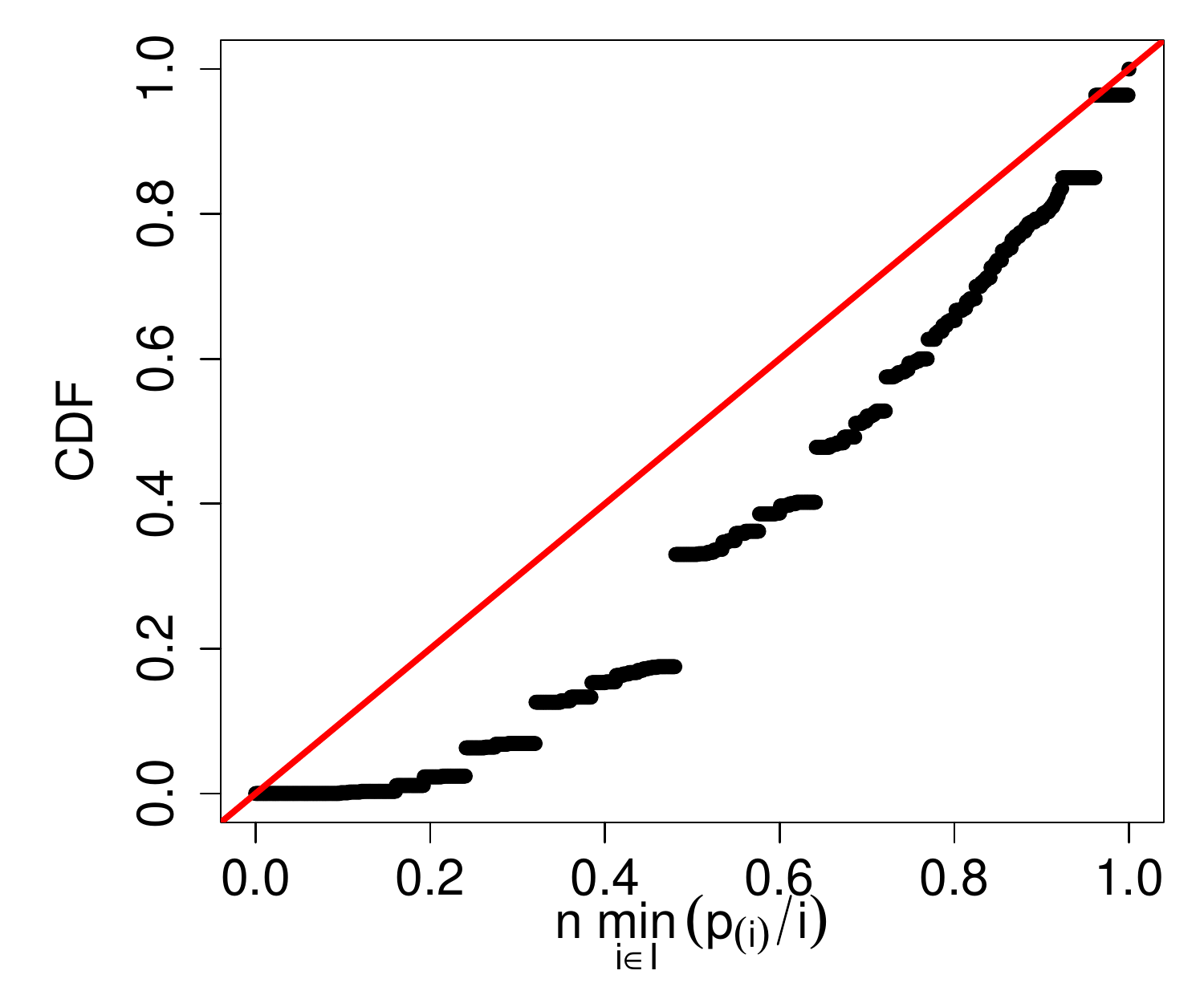}
    \subcaption{$\No = 25$, $\NoNc = 25$}
    \end{subfigure} \quad
    \begin{subfigure}[b]{0.4\textwidth}
    \centering
\includegraphics[width  = \textwidth]{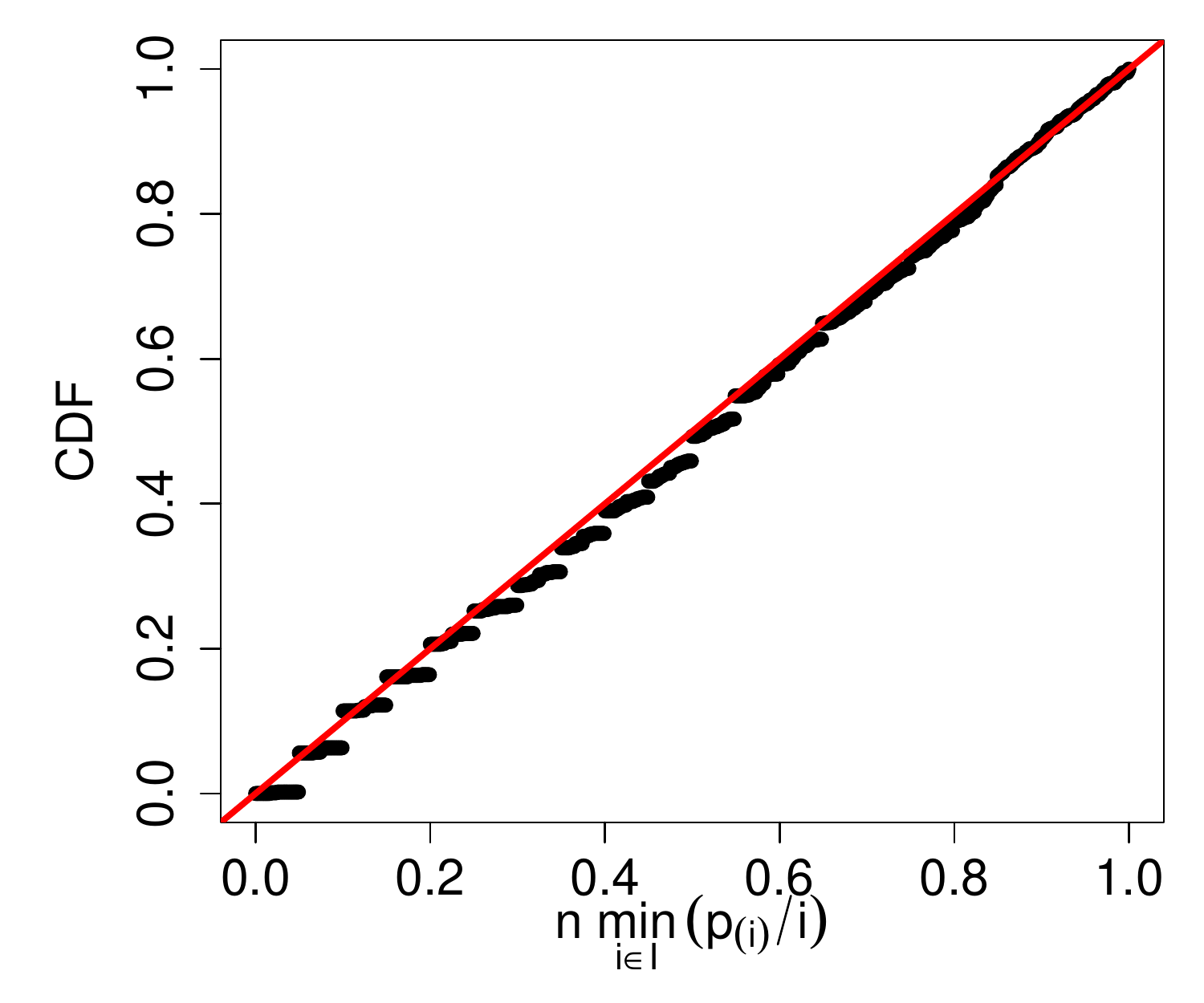}
    \subcaption{$\No = 25$, $\NoNc = 500$}
    \end{subfigure}
    \caption{{Cumulative distribution functions of the permutation
          distribution of the Simes statistic $n \min_{i
  \in \hypothesisIndex{}} (\pval{(i)}/i)$ with different numbers of
internal negative controls ($\NoNc = 25$ and $\NoNc = 500$), estimated by $1000$ random permutations. Simes' test corresponds to assuming a uniform distribution.}}
    \label{fig:Simes.test}
\end{figure}

\subsubsection{Testing individual hypotheses}\label{sec:individual.hypothesis}

We now review methods for testing the individual hypotheses
$\hypothesis{i}$, $i=1,\dotsc,\No$. Bonferroni's correction rejects
$\hypothesis{i}$ if $\pval{i} \leq \alpha / n$. Holm's step-down
procedure is obtained from closing Bonferroni's correction; it orders
the hypotheses by ranking p-values from small to large and 
keeps on rejecting hypotheses as long as $p_{(i)} \leq \alpha/(\No - i
+ 1)$.
\parencite{holm1979simple}. Both procedures only require that the p-values
are individually valid and thus, when applied to the
\nickname~p-values, control the \FWER~if the conditions in
\Cref{prop:pvalue.validity} are satisfied for all $i \in
\hypothesisIndex{0}$. This also applies to graph-based procedures such
as fixed sequence testing and the fallback procedure
\parencite{wiens03_fixed_sequen_bonfer_proced_testin_multip_endpoin,bretz2009graphical}.
Hochberg's procedure
\parencite{hochberg88_sharp_bonfer_proced_multip_tests_signif} and
Hommels' procedure
\parencite{hommel88_stagew_rejec_multip_test_proced} are based on
closing the Simes' test. When applied to \nickname~p-values, they
control the \FWER~under the conditions in \Cref{prop:simes}.

For \FDP~control,
\textcite{lehmann2005generalizations} proposed a step-down procedure
that keeps on rejecting hypotheses as long as $p_{(i)} \leq (\lfloor
q i \rfloor + 1)\alpha/(\No + \lfloor q i \rfloor + 1 - i)$. They
showed that this procedure satisfies $\PP(\FDP > q) \leq \alpha$ if
the so-called Simes' inequality is satisfied \parencite[thm.\
3.2]{lehmann2005generalizations}. Therefore, their procedure applied
to the \nickname~p-values is valid under the conditions in
\Cref{prop:simes}.

As mentioned above, \FWER~and \FDP~control are intimately related
to testing intersection nulls. Given valid tests of the intersection
nulls, \textcite{marcus1976closed} proposed a closed testing
procedure that rejects $H_i$ if all intersection nulls $\{H_{\calJ} =
\cap_{j \in \calJ} H_j: \calJ \ni i\}$ that logically imply $H_i$
are rejected. They showed that this procedure controls the
\FWER. \textcite{goeman2011multiple} extended this procedure to
simultaneously estimate the \FDP~among any subset of hypotheses
\parencite[see also][]{genovese2004stochastic,genovese06_exceed_contr_false_discov_propor}.
\textcite{goeman2021only} showed that all methods that
control \FWER~or some tail probability of \FDP~can be written as or
uniformly improved by a closed testing procedure.
Because closed testing only requires validity of the tests for the
intersection nulls, in principle it can be applied with Simes' tests
using the \nickname~p-values when the conditions in \Cref{prop:simes} are
satisfied.

For \FDR~control, the \BH~step-up procedure is
most widely used. Let $i$ be
the largest index
such that $\pvalOrder{i} \le i q/\No$; the \BH~procedure rejects
$\hypothesis{j}$ for all $j$ such that $\pval{j} \le \pvalOrder{i}$
\parencite{benjamini1995controlling}.
The \BH~procedure is proven to control the \FDR~at the nominal level
when the p-values are independent \cite{benjamini1995controlling} or
 \PRDS~on the set of true nulls \cite{ benjamini2001control}. This
 implies the following result.
\begin{corollary}\label{coro:FDR}
Under any sets of conditions in \Cref{prop:PRDS}, the
\BH~procedure applied to the \nickname~p-values controls the \FDR.
\end{corollary}


\begin{remark} \label{rem:emn}
  A concrete example of an exchangeable but non-i.i.d.\ sequence is the
  equicorrelated multivariate normal (EMN) model:
  \begin{equation}
    \label{eq:emn}
      \testStatistics{i} = \mu_i
+ \sqrt{\rho}Z + \sqrt{1-\rho}X_i,\quad {i \in \hypothesisIndex{} \cup
  \hypothesisIndex{\text{nc}}},
\end{equation}
where $Z$ and $X_i$ are
i.i.d.~standard normal variables and $0 \le \rho < 1$.
Jointly, the test statistics follow a multivariate normal distribution
with all pairwise correlations equal to $\rho$. Suppose $\mu_i = 0$
for all $i \in \nullHypothesisIndex \cup \hypothesisIndex{\text{nc}}$, then
$(\testStatistics{i})_{i \in \nullHypothesisIndex \cup
  \hypothesisIndex{\text{nc}}}$ is exchangeable but not
independent. It is straightforward to verify that the exchangeability
conditions in
\Cref{prop:pvalue.validity}\ref{vali:assu:one.exchangeable.general},
\Cref{prop:PRDS}\ref{PRDS:assu:exchangeable}, and
\Cref{prop:simes}\ref{simes:assu:exchangeable} are all satisfied.
A simulated example from the EMN model is given in \Cref{sec:simulation}.
\end{remark}

\section{The empirical process
  perspective}\label{sec:empirical.process}

An unsatisfactory aspect of \Cref{coro:FDR} is that it requires
the marginal distribution of the test statistic to be the same for all
true null hypotheses and negative controls. In this section, we
develop an alternative justification of this \BH~procedure that
relaxes this assumption. Our argument is based on using
negative controls to form a nonparametric estimator of the false
discovery rate, and is closely related to the empirical process
perspective in \textcite{storey2004strong,genovese2004stochastic}. A
practical advantage of this approach is that one can increase the
power of the \BH~procedure by estimating the proportion of nulls
\parencite{storey2002direct}.

We first set up the notation to state our main result. Following
\textcite{storey2004strong,genovese2004stochastic}, we
view \FDP~as an empirical process indexed by the rejection
threshold. As the \nickname~p-values are invariant under a monotone
transformation of the original test statistics, without loss of
generality, we assume $\testStatistics{i} \in [0,1]$ for all $i$. The
empirical processes for false rejections ($\falsePositive$ in
\Cref{tab:outcome}), all rejections ($\totalPositive$ in
\Cref{tab:outcome}), and the \FDP~are defined as
\begin{align*}
    \falsePositive(t) &:= \sum_{i \in \nullHypothesisIndex}
                        \1_{\{\testStatistics{i} \le t\}}, \quad
    \totalPositive(t) :=
    \sum_{i \in \hypothesisIndex{}} \1_{\{ \testStatistics{i} \le
t\}}, \quad
    \text{and} \quad
    \FDP(t) := \frac{V(t)}{R(t) \vee 1},\quad 0 \leq t \leq 1.
\end{align*}
For fixed $t$, let the expectation of $\FDP(t)$ be $\FDR(t) =
\EE\left[\FDP(t)\right]$. A multiple testing method such as the
\BH~procedure selects a data-dependent rejection threshold $\tau$, and
we are interested in controlling $\FDR = \EE[\FDP(\tau)]$ at level $q$. This
may be achieved by directly estimating $\FDR(t)$ and stopping the
procedure when the estimated \FDR~is above $q$. This typically
involves estimating the
number of false positives $\falsePositive(t)$. For example, in the
usual setting that the test statistics are p-values and follow $U[0,1]$ under the
null, we may estimate
$\falsePositive(t)$ conservatively by $nt$, the expectation of
$\falsePositive(t)$ when $\nullHypothesisIndex = \hypothesisIndex{}$.

Compared to previous work
\parencite{storey2004strong,genovese2004stochastic}, a key
differerence in our problem is that the null distribution of the test
statistics is unknown and must be estimated from the negative
controls. To this end, let the empirical process for the negative
control rejections and its normalization be,
respectively,
\[
  V_{\text{nc}}(t) := \sum_{j \in \hypothesisIndex{\text{nc}}}
  \1_{\{\testStatistics{j} \le t \}},\quad \bar{V}_{\text{nc}}(t) :=
  \frac{n(V_{\text{nc}}(t)+2)}{\NoNc+1}, \quad
  0 \leq t \leq 1.
\]
We propose to estimate $\FDR(t)$ by
\begin{align}\label{eq:FDR.estimate}
    \widehat{\FDR}_{\lambda}(t)
    := \frac{\hat{\pi}(\lambda) \cdot
  \bar{V}_{\text{nc}}(t)}{R(t) \vee 1},\quad \text{for some}~0 <
  \lambda \leq 1,
\end{align}
where $\hat{\pi}(\lambda)$ is the following estimator of
the proportion of true nulls $|\nullHypothesisIndex|/|\hypothesisIndex{}|$:
\begin{align}\label{defi:eq:prop.null}
  \hat{\pi}(\lambda) =
  \begin{dcases}
    1, & \text{if}~\lambda = 1, \\
  \frac{\No + 1 - \totalPositive(\lambda)}{\No} \cdot \frac{\NoNc +
    1}{\NoNc -  \ncFalsePositive(\lambda)},& \text{if}~0 < \lambda <
                                             1. \\
  \end{dcases}
\end{align}
Equation \eqref{defi:eq:prop.null} is modified from \textcite[eq.\
(6)]{storey2004strong}.
It can be shown that $\widehat{\FDR}_{\lambda}(t)$ is a conservative
(i.e.\ downward biased) estimator of $\FDR(t)$ when $\lambda = 1$; see
\Cref{prop:FDR.estimate} in the Appendix.
 Finally, let the rejection threshold be
\begin{align*}
    \stoppingTime = \sup\left\{0 \le t \le \lambda:
  \widehat{\FDR}_{\lambda}(t) \le q \right\},
\end{align*}
so a hypothesis $\hypothesis{i},~i\in\hypothesisIndex{}$ is rejected
if $\testStatistics{i} \leq \stoppingTime$.

The next proposition relates this rejection threshold with the
\BH~procedure applied to the \nickname~p-values.
\begin{proposition}\label{prop:FDR.equivalence}
A hypothesis $\hypothesis{i},~i\in\hypothesisIndex{}$ is rejected by
the above procedure when $\lambda = 1$ if and only if it is rejected
by the \BH~procedure with the following modified \nickname~p-values:
\[
   \tilde{\pval{i}} = \frac{2 + \sum_{j \in \hypothesisIndex{\text{nc}}}
  \1_{\{\testStatistics{j} \le \testStatistics{i}\}}}{1 + \NoNc}
\wedge 1.
\]
\end{proposition}

\Cref{prop:FDR.equivalence} follows from the simple observation that
\begin{align*}
    \widehat{\FDR}(\testStatistics{i}) \le q \quad \text{if and only
  if} \quad
    \tilde{\pval{i}} =
  \frac{2+V_{\text{nc}}(\testStatistics{i})}{1+\NoNc} \wedge 1 \le
  \frac{R(\testStatistics{i})}{\No} q.
\end{align*}
Compared to the original \nickname~p-value $p_i$ defined in
\eqref{eq:pval}, an extra $1$ is added to the numerator of
$\tilde{\pval{i}}$. This subtle modification is needed because, unlike
the problem with a known null distribution studied by
\textcite{storey2004strong} and others, due to the
discreteness of $V_{\text{nc}}(t)$ (and hence
$\bar{V}_{\text{nc}}(t)$), it
is generally not true that $\widehat{\FDR}_{\lambda}(\stoppingTime)
\leq q$. Although this modification only makes a
minuscule difference in most practical problems, it is needed in the
super-martingale proof of the next Theorem.

The final piece we need to state the main theorem of this section is a
stronger notion of stochastic dominance \parencite{zhao2019multiple}.
\begin{definition}[Uniform stochastic dominance]\label{defi:uniform.dominance}
    For two random variables $X$, $Y$ supported on $[0,1]$, we say $X$
    is uniformly stochastically larger than $Y$ if $\PP(X \le t) > 0$,
    $\PP(Y \le t) > 0$, and $\PP(X \le s \mid X \le t) \le \PP(Y \le s
    \mid Y \le t)$ for all $0 < s \leq t \leq 1$.
\end{definition}
Heuristically, uniform stochastic dominance just means stochastic
dominance after conditioning on the variable is less than $t$ for all
$t$. It is satisfed if the distributions of $X$ and $Y$ are in a family
with monotone likelihood ratio; more examples and results can be found
in \textcite{whitt1980uniform,zhao2019multiple}.

\begin{theorem}\label{prop:FDR}
Suppose the following conditions are true:
\begin{enumerate}[label = (\alph*),ref = (\alph*)]
\item \label{FDR:assu:null.nc.uniformly.conservative}
  $\testStatistics{i}$ is uniformly stochastically larger than
  $\testStatistics{j}$ for all $i \in \nullHypothesisIndex$ and $j \in \hypothesisIndex{\text{nc}}$;
 \item \label{FDR:assu:independent}
 $(\testStatistics{i})_{i \in \hypothesisIndex{} \cup
   \hypothesisIndex{\text{nc}}}$ is mutually independent.
\end{enumerate}
Then for any fixed $0 < \lambda \leq 1$, the step-up procedure with
rejection threshold $\stoppingTime$
controls the \FDR~at level $q$.
\end{theorem}

We provide a sketch proof of this result by modifying the martingale
argument in \cite{storey2004strong}; more details can be found in
\Cref{sec:results-empirical-process}.
We consider the time-reversals of $\falsePositive(t)$,
$\truePositive(t)$, and $\ncFalsePositive(t)$ starting from $t=1$
and define the backward filtrations as $\calF_t =
\sigma\big(\falsePositive(s), \truePositive(s), \ncFalsePositive(s): t
\le s \le 1\big)$ for $0 \leq t \leq 1$.
Our proof rests on showing the following process is a backward
super-martingale:
\begin{align}\label{eq:super.martingale}
    M(t) = \frac{\falsePositive(t)}{(1+\ncFalsePositive(t))/(1+\NoNc)}.
\end{align}
This extends the martingale $V(t)/t$ in \cite{storey2004strong} when
the null CDF is known to be $F_0(t) = t$.
More precisely, it is proved in the Appendix that
\begin{align} \label{eq:super-martingale}
    \EE\left[M(s) \mid \calF_t \right] \leq M(t)
    \cdot \left(1 - \left(1-p_t^s
  \right)^{\ncFalsePositive(t) + 1} \right)
    \le M(t)~\text{for some}~0 \leq p_{t}^s \leq 1 ~\text{and all}~0
  \leq s \leq t \leq 1.
\end{align}
\Cref{prop:FDR} then follows from applying the optional
stopping theorem.

\begin{remark} \label{rem:one-sided}
  Since \Cref{prop:FDR} does not require exchangeability of the test
  statistics, it can be applied to one-sided tests. Specifically,
  suppose $T_i$ is
  the likelihood-ratio statistic for testing $\hypothesis{0}: \theta_i
  \le \theta_0$ vs.\ $\hypothesis{1}: \theta_i > \theta_0$ in a
  one-dimensional exponential family with natural (or mean) parameter
  $\theta$ (and thus has a monotone likelihood ratio). We have
  $\theta_i \leq \theta_0$ for $i \in \hypothesisIndex{0}$ by
  definition. If $\theta_i \geq \theta_0$ for all $i \in
  \hypothesisIndex{\text{nc}}$, then
  condition \ref{FDR:assu:null.nc.uniformly.conservative} in
  \Cref{prop:FDR} is satisfied. This suggests another useful aspect of
  the stochastic dominance
  condition~\ref{FDR:assu:null.nc.uniformly.conservative}: the
  definition of ``negative controls'' can be relaxed and they do not need to
  be true null hypotheses. In other words, the \nickname~p-values are
  robust to incorrect selection of negative controls in the sense that
  the FDR may still be controlled at the nominal level. However, when
  too many negative controls are not true nulls, the multiple testing
  procedure may have very little power.

\end{remark}

\begin{remark}
  A similar martingale argument was developed by
  \textcite{mary22_semi_super_multip_testin} to prove that the
  \BH~procedure ($\lambda = 1$) applied to the \nickname~p-values
  controls the \FDR. A main distinction is that they
  require partial exchangeability of the test statistics (condition
  \ref{PRDS:assu:exchangeable} in \Cref{prop:PRDS}), which is weaker
  than the independence condition \ref{FDR:assu:independent} in
  \Cref{prop:FDR} but does not allow
  the case of uniformly stochastic
  dominance in condition
  \ref{FDR:assu:null.nc.uniformly.conservative}.\footnote{Although by
    using de Finetti's theorem, the
    independence condition \ref{FDR:assu:independent} in
    \Cref{prop:FDR} can be easily relaxed to (conditional)
    $\infty$-extendability \parencite{diaconis1980finite}.}
  The uniformly stochastic dominance
  condition~\ref{FDR:assu:null.nc.uniformly.conservative} arises
  naturally in our proof of the first inequality. It remains unclear to us whether this
  can be allowed in the proof in
  \textcite{mary22_semi_super_multip_testin}, as their martingale is
  not indexed by the rejection threshold. When the marginal
  distributions of the null test statistics and internal negative controls are the same,
  the gap $M(t) - \EE\left[M(s) \mid \calF_t \right] $ is small for large
  $\ncFalsePositive(t)$, so $M(t)$ is almost a martingale. This can be
  used to prove a lower bound on the \FDR; see
  \textcite[thm.\ 3.1]{mary22_semi_super_multip_testin}.
\end{remark}

\section{\LocalFDR~control}\label{sec:localFDR}

We now turn to our third method motivated by the ad hoc procedure in
\textcite{hung14_proteom_mappin_human_mitoc_inter}. We will consider
the two-mixture setup for multiple testing in
\Cref{sec:setup-terminology}. More specifically, we will assume for
the rest of this section that $(\hypothesis{i},
\testStatistics{i}),~i=1,\dotsc,\No$ are i.i.d., $\hypothesis{i} \sim
\text{Bernoulli}(1 - \pi)$, and $\testStatistics{i} \mid
\hypothesis{i} \sim \cdfTestStatistics{\hypothesis{i}}$, so the
marginal CDF of $T_i,~i \in \hypothesisIndex{}$ is given by $F(t) =
\pi F_0(t) + (1 - \pi) F_1(t)$. To simplify
the discussion, we will assume the null proportion $\probNull$ is
known. In practice, $\probNull$ is generally unknown and a wealth of
estimators of $\probNull$ have been proposed in the literature
\cite{schweder1982plots, hochberg1990more, hengartner1995finite,
  swanepoel1999limiting, benjamini2000adaptive, storey2002direct,
  jin2007estimating}; see \cite{genovese2004stochastic} for a review
of their asymptotic properties.

\subsection{\pdfBased\ methods}
\label{sec:pdfbased-methods}

Currently, most
practical applications estimate the $\localFDR$ by plugging in
estimators of $\pdfTestStatistics{}$ (and $\pdfTestStatisticsNull$ if it
is unknown) into the definition of \localFDR~in \eqref{eq:local-fdr};
we will call such methods \pdfBased, to contrast with our proposal below
that is \cdfBased. A major limitation of the \pdfBased~approach is
that the estimated densities may be very inaccurate at the tail where
the rejection threshold is likely to be located. This is worsened by
the fact that the marginal density $f(t)$ appears in the denominator
of the definition of $\localFDR(t)$.
Another problem is that \pdfBased~estimators of the $\localFDR$ are
generally not invariant to, monotone transformation of the test
statistics, but the definition of $\localFDR$ is. Thus, very different
rejection sets may be obtained if the investigator chooses to use
different transformations of the test statistics.

We illustrate the performance of \pdfBased~methods with a simple
simulation example.  We generate $\No =
400$ test statistics with $\probNull = 0.5$, $\cdfTestStatisticsNull =
t_{10}$ (t-distribution with $10$ degrees of freedom), and
$\cdfTestStatisticsNonNull = \text{Exp}(1)$, and an
independent set of $\NoNc = 1000$ internal negative controls from
$\cdfTestStatisticsNull$. We assume $\probNull$ is known and set
$\level = 0.3$. \Cref{fig:pdfbased-1} shows that when a simple kernel
density estimator is used to estimate both $\cdfTestStatistics{}$ and
$\cdfTestStatisticsNull$, the estimated \localFDR\ is highly variable
at the left tail and the step-down rejection threshold (dashed
verticle line) is too conservative to be useful. Transforming the test
statistics to z-scores improves the performance of this
\pdfBased~method, but the rejection threshold is still too small
(\Cref{fig:pdfbased-1}).

\subsection{\cdfBased\ methods}
\label{sec:cdfbased-methods}

Next, we relate the \cdfBased\ cut-off analysis in
\textcite{hung14_proteom_mappin_human_mitoc_inter} to
\localFDR~control. To this end, consider the following optimization
problem for some given $\lambda = q / \probNull$ and $0 < q < 1$:
\begin{align}\label{eq:Bayes.risk.2}
 \EmpiricalThreshold := \argmin_{\threshold}
 \EmpiricalCdfTestStatisticsNull(\threshold) - \lambda
 \EmpiricalCdfTestStatistics(\threshold),
\end{align}
where the objective function is a weighted difference between the empirical CDFs
of the negative controls and test statistics:
\[
  \EmpiricalCdfTestStatisticsNull(t) = \frac{1}{\NoNc} \sum_{j \in
    \hypothesisIndex{\text{nc}}}  \1_{\{\testStatistics{j} \le t\}},~
  \EmpiricalCdfTestStatistics(t) = \frac{1}{\No} \sum_{i \in
    \hypothesisIndex{}}  \1_{\{\testStatistics{i} \le t\}}.
\]
Thus, the procedure in
\textcite{hung14_proteom_mappin_human_mitoc_inter} (see
\Cref{fig:lfdr}) corresponds to using $\lambda=1$ or equivalently $q =
\pi$. In a moment, it will be clear that $q$ can be understood as the
\localFDR~level targeted by using the rejection threshold
$\EmpiricalThreshold$.

By the Glivenko-Cantelli theorem, $\EmpiricalCdfTestStatisticsNull$
and $\EmpiricalCdfTestStatistics$ converge uniformly
to $\cdfTestStatisticsNull{}$ and $\cdfTestStatistics{}$,
respectively. Thus, we expect $\EmpiricalThreshold$ to converge to the
minimizer
\begin{equation}
  \label{eq:optimal-threshold}
  \tau_{\lambda}^{*} := \argmin_\threshold
  \cdfTestStatisticsNull(\threshold) - \lambda
  \cdfTestStatistics{}(\threshold).
\end{equation}
By setting the derivative of Eq.~\eqref{eq:optimal-threshold} to zero, we
obtain $\pdfTestStatisticsNull(\tau_{\lambda}^{*}) = (q / \probNull)
\pdfTestStatistics{}(\tau_{\lambda}^{*})$, or equivalently,
\[
  \localFDR(\tau_{\lambda}^{*}) = \frac{\pi
    \pdfTestStatisticsNull(\tau_{\lambda}^{*})}{\pdfTestStatistics{}(\tau_{\lambda}^{*})}
  = q,
\]
where $\pdfTestStatisticsNull{}$ and $\pdfTestStatistics{}$ are the
density functions corresponding to $\cdfTestStatisticsNull{}$
and $\cdfTestStatistics{}$.
As mentioned above, the procedure in
\textcite{hung14_proteom_mappin_human_mitoc_inter} corresponds to
using $\lambda = 1$, or equivalent $q = \pi$. Intuitively, this is a
sensible choice because $\pi$ is simply the probability of making a
false discovery by rejecting a random hypothesis.

\begin{figure}[tbp]
        \centering
        \begin{minipage}{0.3\textwidth}
                \centering
                \includegraphics[clip, trim = 0cm 0cm 0cm 0cm, width = \textwidth]{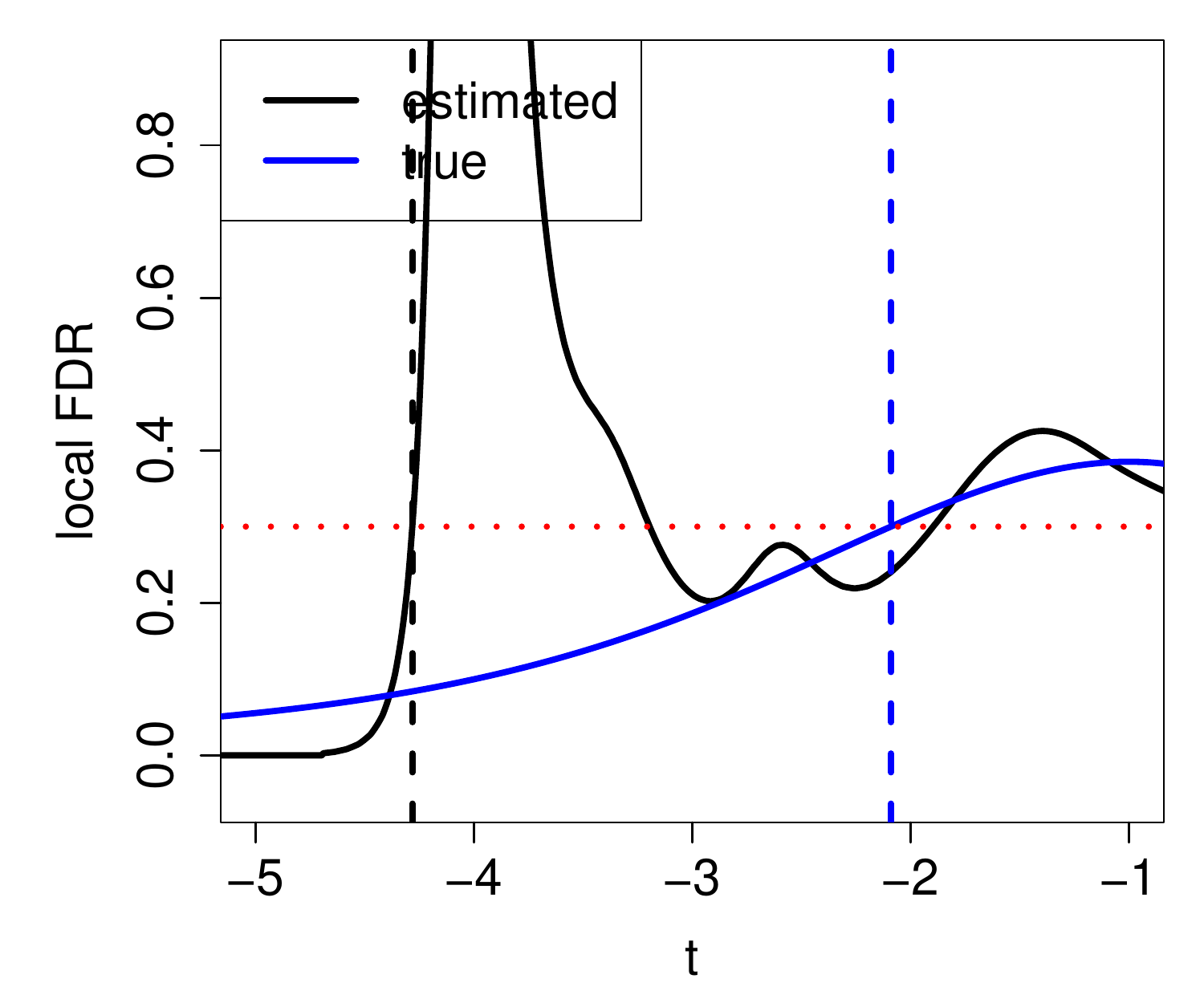}
                \subcaption{\pdfBased}
                \label{fig:pdfbased-1}
        \end{minipage}
    \begin{minipage}{0.3\textwidth}
                \centering
                \includegraphics[clip, trim = 0cm 0cm 0cm 0cm, width = \textwidth]{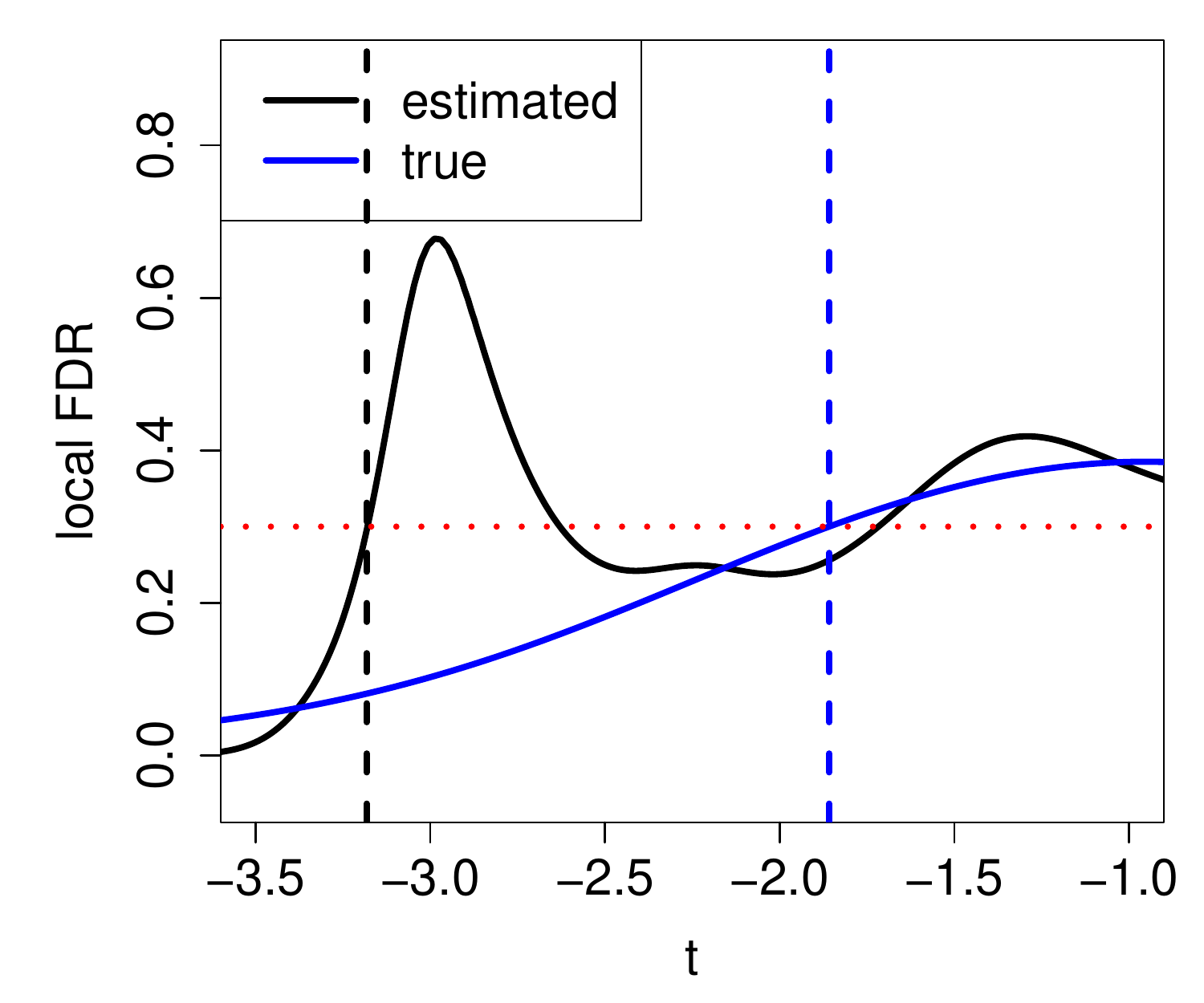}
                \subcaption{\pdfBased, transformed $\testStatistics{i}$}
                \label{fig:pdfbased-2}
        \end{minipage}
        \begin{minipage}{0.3\textwidth}
                \centering
                \includegraphics[clip, trim = 0cm 0cm 0cm 0cm, width = \textwidth]{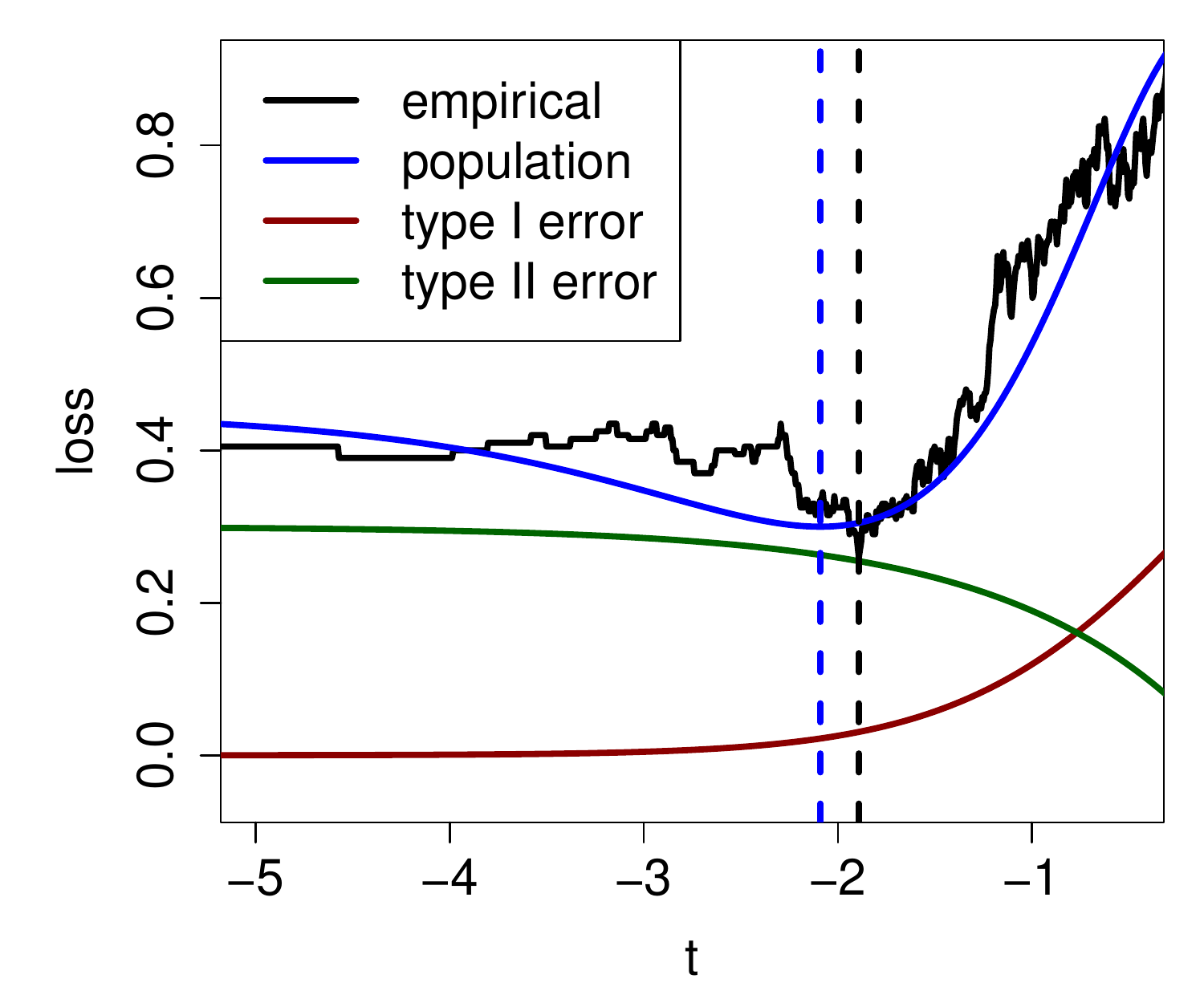}
                \subcaption{\cdfBased}
                \label{fig:cdfBased}
        \end{minipage}
        \caption{Comparison of \cdfBased~and \pdfBased~methods where
          $\No = 400$, $\NoNc = 1000$, $\pi = 0.5$, $F_0 = t_{10}$,
          $F_1 = \text{Exp}(1)$, and $q = 0.3$.
        Panel (a)
        shows the result of a \pdfBased~method, where the null density
        and the marginal density are estimated by a kernel estimator (R
        function \texttt{density} with default options). In panel (b),
        we apply the \pdfBased~method to the transformed test
        statistics
        $\Phi^{-1}(\cdfTestStatisticsNull(\testStatistics{i}))$, where
        the null distribution of the transformed test statistics is
        $\calN(0,1)$. In panel (c), we showcase the proposed
        \cdfBased~method \eqref{eq:Bayes.risk.2}. We plot the weighted
        type I error
        $(1-\level)\cdfTestStatisticsNull(\threshold)$ (dark red), the
        weighted type II error $\level
        (1-\probNull)/\probNull(1-\cdfTestStatisticsNonNull(\threshold))$
        (dark green), and superimpose the population risk (blue)---the
        sum of the two weighted errors up to a constant shift, and the
        empirical objective (black). (The population and
        empirical risk curves have been scaled for visualization
        purposes.)
    }
        \label{fig:heavyTailed}
\end{figure}


The heuristic argument in the previous paragraph can be made precise
with the help of the empirical process theory. By assuming that the
density functions $f_0$ and $f$ are
differentiable at $\tau_{\lambda}^{*}$ and $(f_0/f)'(\tau_{\lambda}^{*})
> 0$, it can be shown that the estimated threshold
$\EmpiricalThreshold$ converges to $\tau_{\lambda}^{*}$ at
rate $(\No \wedge \NoNc)^{-1/3}$ when $n, m \to \infty$. Morevoer, by
assuming monotone likelihood ratio (i.e.\
$(\pdfTestStatisticsNull{}/\pdfTestStatistics{})'>0$) and
additional regularity conditions, it can be shown that such
convergence is uniform over $\lambda$. See \Cref{sec:results-localFDR} for some technical results.

The optimization problem \eqref{eq:optimal-threshold} can be rewritten
as a decision-theoretic problem: let $\hat{H}_i(t) = \1_{\{T_i \leq
  \threshold\}}$ be the one-sided decision rule, then
\begin{equation}
  \label{eq:optimal-threshold-2}
  \tau_{\lambda}^{*} = \argmin_\threshold
  \EE_{\cdfTestStatisticsNull}\left[\hat{H}_i(t)\right]
   - \lambda \EE_{\cdfTestStatistics{}}\left[\hat{H}_i(t)\right]
   = \argmin_\threshold \mathbb{E}\left[
     \weightNull \1_{\{\hypothesis{i} < \hat{\hypothesis{i}}(t)\}}
     + \weightNonNull \1_{\{\hypothesis{i} > \hat{\hypothesis{i}}(t)\}}
  \right],
\end{equation}
where $\weightNull = 1-q$ is the loss of a false
positive and $\weightNonNull = q$ is the loss of a
false negative. 
This connection was pointed out by
\textcite{sun2007oracle}.

In contrast to \pdfBased~methods, the \cdfBased~method proposed here
is free of tuning parameters and invariant to monotone
transformations of the test statistics. \Cref{fig:cdfBased}
illustrates the practical performance of this procedure using the
same simulation example above, by showing the
population and empirical risk in \eqref{eq:optimal-threshold-2} as
curves of the rejection threshold $t$. It is evident that the
\cdfBased~rejection threshold $\EmpiricalThreshold$ is close to
$\tau_{\lambda}^{*}$ and its risk is close to the oracle.


We offer some additional remarks on the \cdfBased~method above.

\begin{remark}
  By using the order statistics $\pval{(1)} < \dotsb < \pval{(\No)}$ of
  the \nickname~p-values, the optimization
  problem~\eqref{eq:Bayes.risk.2} can also be rewritten as as
  \begin{align*}
     \EmpiricalThreshold = T_{(i^{*})},~\text{where}~i^{*} = \argmin_i
      \frac{\NoNc+1}{\NoNc} \pvalOrder{i} -
      \frac{\lambda i}{\No}.
  \end{align*}
  This is almost identical to the method in
  \textcite{soloff22_edge_discov}. More precisely,
  \textcite{soloff22_edge_discov} assumes an independent sequence of
  p-values that are uniformly distributed under the null (so $F_0(t) =
  t$ is known) is given and does not include the factor $(m+1)/m$
  in the last equation. \textcite{soloff22_edge_discov} not only
  proved the same $\No^{-1/3}$ convergence rate for their estimator of
  the rejection threshold but also showed that the expected maximum
  \localFDR\ of the rejected hypotheses is controlled if the
  likelihood ratio is monotone. Our numerical simulation suggests that
  the method proposed here might also be able to control the maximum
  \localFDR, but we are unable use the technique developed by
  \textcite{soloff22_edge_discov} to prove this because the
  \nickname~p-values are not independent.
\end{remark}

\begin{remark}
The method above can be extended to estimate the \localFDR~curve,
which shall be denoted by $q(t)$. In fact,
$\hat{\tau}(\level):=\hat{\tau}_{\lambda = q/\probNull, \No, \NoNc}$
is an increasing and piece-wise constant function. Thus, the
\localFDR~curve can be estimated by inverting $\hat{\tau}(\level)$,
\begin{align}\label{eq:curve:inverse}
    \estimatedLocalFDR{\threshold} = \inf_\level \{\level: \hat{\tau}(q) \ge \threshold \}, \quad \text{for any}~\threshold.
\end{align}
The resulting \localFDR~curve $\estimatedLocalFDR{\threshold}$ is increasing, piece-wise constant, left-continuous, and the jump points are contained in $(\testStatistics{i})_{i \in \hypothesisIndex{}}$.
The estimated curve
$\hat{q}(t)$ is essentially a light modification of Grenander's
estimator \cite{grenander1956theory} of a monotone density
function. It is well known that Grenander's estimator converges at
the rate $n^{-1/3}$; see \cite{durot2018limit} for a recent
review.
\end{remark}

\begin{remark}
Because the objective function
\eqref{eq:optimal-threshold} is typicaly locally convex at its optimum,
it is expected that the regret
  $\{\cdfTestStatisticsNull(\EmpiricalThreshold) - \lambda
\cdfTestStatistics{}(\EmpiricalThreshold)\} -
\{\cdfTestStatisticsNull(\tau_{\lambda}^{*}) - \lambda
\cdfTestStatistics{}(\tau_{\lambda}^{*})\}$ usually converges to zero
twice as fast as $\EmpiricalThreshold$ converges to
$\tau_{\lambda}^{*}$. However, even when the risk is almost flat around
$\tau_{\lambda}^{*}$, a simple argument using the
Dvoretzky–Kiefer–Wolfowitz (DKW)
  inequality shows that the regret
converges at least at the
rate $\No^{-1/2}$; see \Cref{prop:convergence.rate.weak} in the
Appendix. Thus, the simple \cdfBased~method almost always has a small
regret, even if the rejection threshold $\tau_{\lambda}^{*}$ is not
estimated very accurately.
\end{remark}


\section{Simulations}\label{sec:simulation}

When using multiple testing methods based on negative controls, a
concern in practice is that they may be not very powerful. We
investigate this using numerical simulations.


We generate a set of
baseline p-values $(\testStatistics{i})_{i \in \hypothesisIndex{} \cup
  \hypothesisIndex{\text{nc}}}$ that might be individually invalid,
and compare three variations of the \BH~procedure:
\begin{enumerate}
\item the standard \BH~procedure (\BH) that assumes the validity of the
baseline p-values and directly aggregates $(\testStatistics{i})_{i \in
  \hypothesisIndex{}}$;
\item the \BH~procedure with \nickname~p-values (\BH~\nickname) that first
computes the \nickname~p-values $(\pval{i})_{i \in
  \hypothesisIndex{}}$ using $(\testStatistics{i})_{i \in
  \hypothesisIndex{} \cup \hypothesisIndex{\text{nc}}}$ and then
applies the \BH~procedure to $(\pval{i})_{i \in \hypothesisIndex{}}$;
\item the oracle \BH~procedure (\BH~oracle) applies the \BH~procedure to the
p-values corrected by the true null CDF.
\end{enumerate}

We experiment with $6 = 3 \times 2$ joint distributions of the
baseline p-values with different marginal distributions and dependency
structures.
\begin{itemize}
    \item \textit{Marginal distribution}. We consider three marginal
      distributions of null baseline p-values: $T \overset{d}{=}
      \Phi(Z + \mu)$ where $Z \sim \calN(0,1)$ and $\mu \in
      \{-0.5,0,0.5\}$, corresponding respectively to anti-conservative
      (anti-csvr.), exact (exact), and conservative (csvr.) p-values.
    The marginal distribution of the non-null baseline p-values is set to $\Phi(Z - 3)$.
    \item \textit{Dependency}. We consider two types of dependencies:
      the independent setting (ind.) where $(\testStatistics{i})_{i
        \in \hypothesisIndex{} \cup \hypothesisIndex{\text{nc}}}$ are
      mutually independent, and the exchangeable setting (exch.) where
      $(\testStatistics{i})_{i \in \hypothesisIndex{} \cup
        \hypothesisIndex{\text{nc}}}$ are generated from an EMN model
      in \Cref{rem:emn} with correlation parameter $\rho = 0.5$.
\end{itemize}
We generate $\NoNull = 100$ null test statistics, $\NoNonNull = 10$ non-null test statistics, and $\NoNc = 200$ internal negative controls.
In each trial, we record the \FDP~and the true positive rate (the
number of true discoveries divided by the total number of non-nulls)
of each method. All configurations are repeated $10^4$ times. We set
the target \FDR~level to $\FDRLevel = 0.2$.

\begin{table}[tbp]
  \centering
\caption{{\FDR~and power analysis. We compare the validity and the power of three variants of the \BH~procedure (\FDR~level $\FDRLevel = 0.2$) across $6$ joint distributions of the baseline p-values. The standard deviation of the \FDP~and the true positive rate is recorded in the bracket. All settings are repeated $10^4$ times.}}
\label{tab:simulation}
\begin{tabular}{cc|cc|cc|cc}
\toprule
\multicolumn{2}{c|}{\multirow{2}{*}{}}                                    & \multicolumn{2}{c|}{BH}              & \multicolumn{2}{c|}{BH \nickname}              & \multicolumn{2}{c}{BH oracle}                \\ \cline{3-8}
\multicolumn{2}{c|}{}
                                                                          & \multicolumn{1}{c|}{FDR}      & power & \multicolumn{1}{c|}{FDR} & power & \multicolumn{1}{c|}{FDR} & power\\
  \midrule
\multicolumn{1}{c|}{\multirow{6}{*}{ind.}}  & \multirow{2}{*}{csvr.}      & \multicolumn{1}{c|}{0.047}   & 0.8      & \multicolumn{1}{c|}{0.17}    & 0.9           & \multicolumn{1}{c|}{0.18}      & 0.93            \\
\multicolumn{1}{c|}{}                       &                             & \multicolumn{1}{c|}{(0.072)} & (0.15)   & \multicolumn{1}{c|}{(0.13)}  & (0.13)        & \multicolumn{1}{c|}{(0.12)}    & (0.087)         \\ \cline{2-8}
\multicolumn{1}{c|}{}                       & \multirow{2}{*}{exact}      & \multicolumn{1}{c|}{0.18}    & 0.82     & \multicolumn{1}{c|}{0.16}    & 0.76          & \multicolumn{1}{c|}{0.18}      & 0.82            \\
\multicolumn{1}{c|}{}                       &                             & \multicolumn{1}{c|}{(0.13)}  & (0.14)   & \multicolumn{1}{c|}{(0.14)}  & (0.21)        & \multicolumn{1}{c|}{(0.13)}    & (0.14)          \\ \cline{2-8}
\multicolumn{1}{c|}{}                       & \multirow{2}{*}{anti-csvr.} & \multicolumn{1}{c|}{0.49}    & 0.87     & \multicolumn{1}{c|}{0.16}    & 0.53          & \multicolumn{1}{c|}{0.18}      & 0.63            \\
\multicolumn{1}{c|}{}                       &                             & \multicolumn{1}{c|}{(0.13)}  & (0.12)   & \multicolumn{1}{c|}{(0.16)}  & (0.28)        & \multicolumn{1}{c|}{(0.15)}    & (0.19)          \\ \hline
\multicolumn{1}{c|}{\multirow{6}{*}{exch.}} & \multirow{2}{*}{csvr.}      & \multicolumn{1}{c|}{0.044}   & 0.76     & \multicolumn{1}{c|}{0.17}    & 1             & \multicolumn{1}{c|}{0.13}      & 0.9             \\
\multicolumn{1}{c|}{}                       &                             & \multicolumn{1}{c|}{(0.14)}  & (0.31)   & \multicolumn{1}{c|}{(0.13)}  & (0.018)       & \multicolumn{1}{c|}{(0.25)}    & (0.2)           \\ \cline{2-8}
\multicolumn{1}{c|}{}                       & \multirow{2}{*}{exact}      & \multicolumn{1}{c|}{0.13}    & 0.77     & \multicolumn{1}{c|}{0.17}    & 0.98          & \multicolumn{1}{c|}{0.13}      & 0.77            \\
\multicolumn{1}{c|}{}                       &                             & \multicolumn{1}{c|}{(0.25)}  & (0.31)   & \multicolumn{1}{c|}{(0.13)}  & (0.047)       & \multicolumn{1}{c|}{(0.25)}    & (0.31)          \\ \cline{2-8}
\multicolumn{1}{c|}{}                       & \multirow{2}{*}{anti-csvr.} & \multicolumn{1}{c|}{0.31}    & 0.77     & \multicolumn{1}{c|}{0.17}    & 0.91          & \multicolumn{1}{c|}{0.13}      & 0.57            \\
\multicolumn{1}{c|}{}                       &
                                                                                                                 & \multicolumn{1}{c|}{(0.35)}  & (0.31)   & \multicolumn{1}{c|}{(0.13)}  & (0.12)        & \multicolumn{1}{c|}{(0.25)}    & (0.38)          \\
  \bottomrule
\end{tabular}
\end{table}

\Cref{tab:simulation} reports the result of this simulation study.
The \BH~\nickname~controls \FDR~in all $6$ settings, while the
standard \BH~fails when the individual p-values are invalid (in two
anti-csvr.\ settings).
In terms of statistical power, the \BH~\nickname~is comparable to the
\BH~oracle when the p-values are independent, and, perhaps
surprisingly, is more powerful when the p-values are positively
dependent (in three exch.\ settings). In fact, \BH~\nickname~is more
powerful when the baseline p-values are positively dependent than when
they are independent.
We believe this surprising gain of power is due to the fact that the
\nickname~p-values are invariant to monotone transformations and are
hence invariant to the shared latent factor $Z$ in the EMN model
\eqref{eq:emn}. In other words, the \nickname~p-values effectively
have a larger signal-to-noise-ratio in the dependent case.



An important question in practice is how many negative controls are
needed to use \nickname~p-values and have decent power. In order for
\BH~\nickname~to reject all non-nulls, a necessary
condition is that $p_{(\NoNonNull)} \le \FDRLevel \NoNonNull/\No$.
Since the minimal \nickname~p-value is at least $1/(1+\NoNc)$, the
above constraint implies that we need $\NoNc \ge C \No/(\FDRLevel
\NoNonNull)$ for some multiple $C > 1$ that may depend on the signal
strength. By varying the number of internal negative controls in the
the simulation setup above, we find that $C = 2$ is a good rule of
thumb for the power of \BH~\nickname~to be close to that of
\BH~oracle; see \Cref{fig:nc.sample.size} in the Appendix. When the
non-nulls have a smaller effect size, a larger $C$ may be
required. See also the simulation study in
\textcite{mary22_semi_super_multip_testin} and a related rule-of-thumb
developed there.



\section{Real data analysis}\label{sec:real.data}

We now return to the motivating proteomic data analysis described in
\Cref{sec:motivating.data}. As mentioned already, there are $\No =
740$ proteins under investigation and $\NoNc = 2,067$ internal
negative control proteins.
We denote the protein abundance under the treatment and the control
condition by $\responseTreatment{i}$ and $\responseControl{i}$,
respectively.
For each protein $i \in \hypothesisIndex{} = \{1,\dotsc,\No\}$, we would
like to test the one-sided hypothesis:
\begin{align*}
    \hypothesis{i,0}: \EE[\responseTreatment{i}] \le  \EE[\responseControl{i}] \quad \text{v.s.} \quad \hypothesis{i,1}: \EE[\responseTreatment{i}] >  \EE[\responseControl{i}].
\end{align*}
For internal negative control protein $j \in
\hypothesisIndex{\text{nc}} = \{741, \dotsc, 2807\}$, its expression
is anticipated to be the same over the two conditions, i.e.\
$\EE[\responseTreatment{j}] =  \EE[\responseControl{j}]$.

\subsection{Falsification of negative controls}
\label{sec:fals-negat-contr}

The validity of the internal negative controls can be falsified by
comparing the empirical distribution of $\responseTreatment{j} -
\responseControl{j}$ over different subgroups of negative controls.
In this example, \Cref{fig:nc.subgroup} shows that the test statistics
for proteins annotated with different non-membrane subcellular
locations are distributed similarly, thereby supporting the usage of
them as internal negative controls.

\begin{figure}[tbp]
    \centering
    \begin{minipage}{5cm}
    \centering
\includegraphics[clip, trim = 0cm 0cm 0cm 0cm, width  = 5cm]{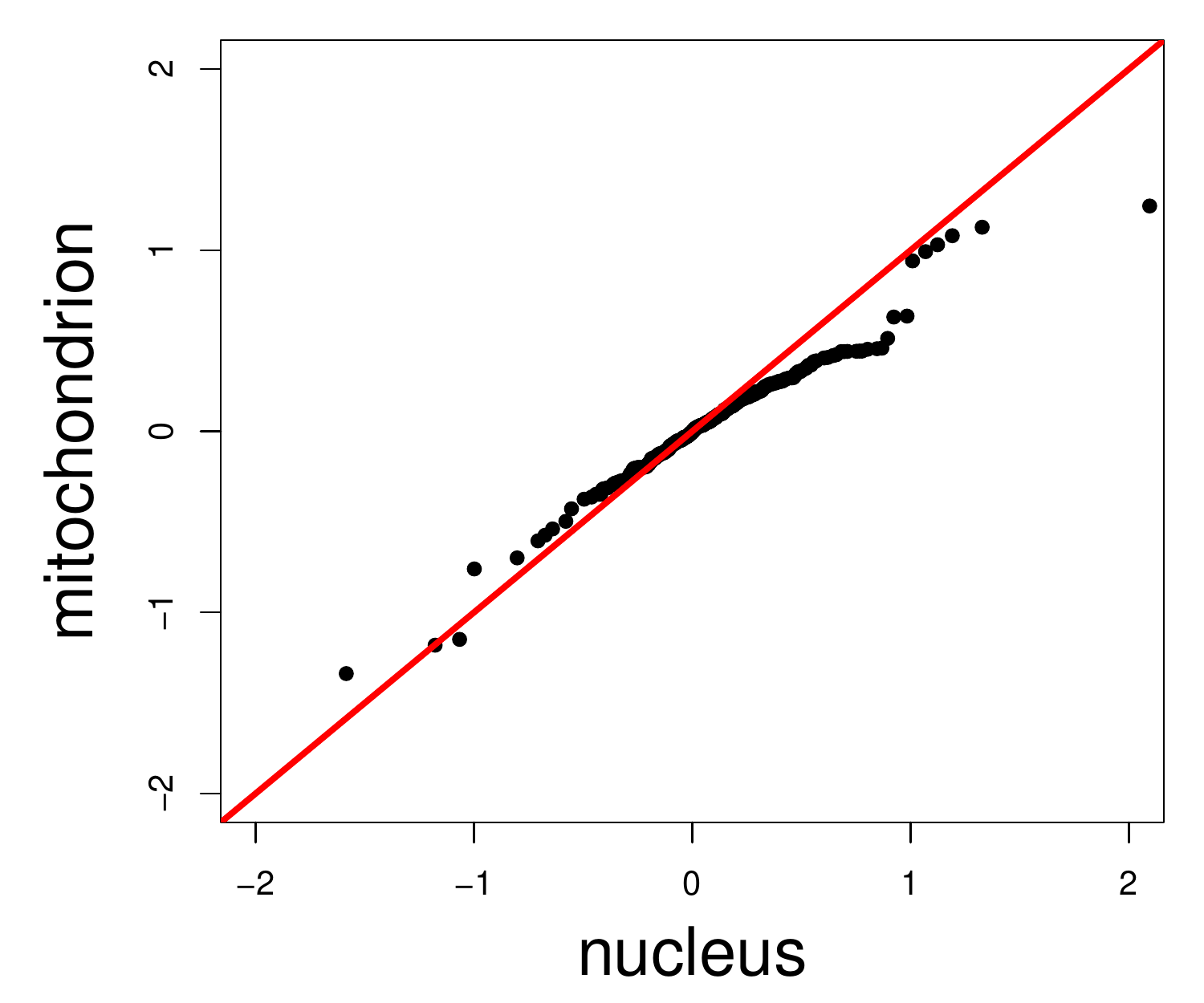}
    \end{minipage}
    \begin{minipage}{5cm}
    \centering
\includegraphics[clip, trim = 0cm 0cm 0cm 0cm, width  = 5cm]{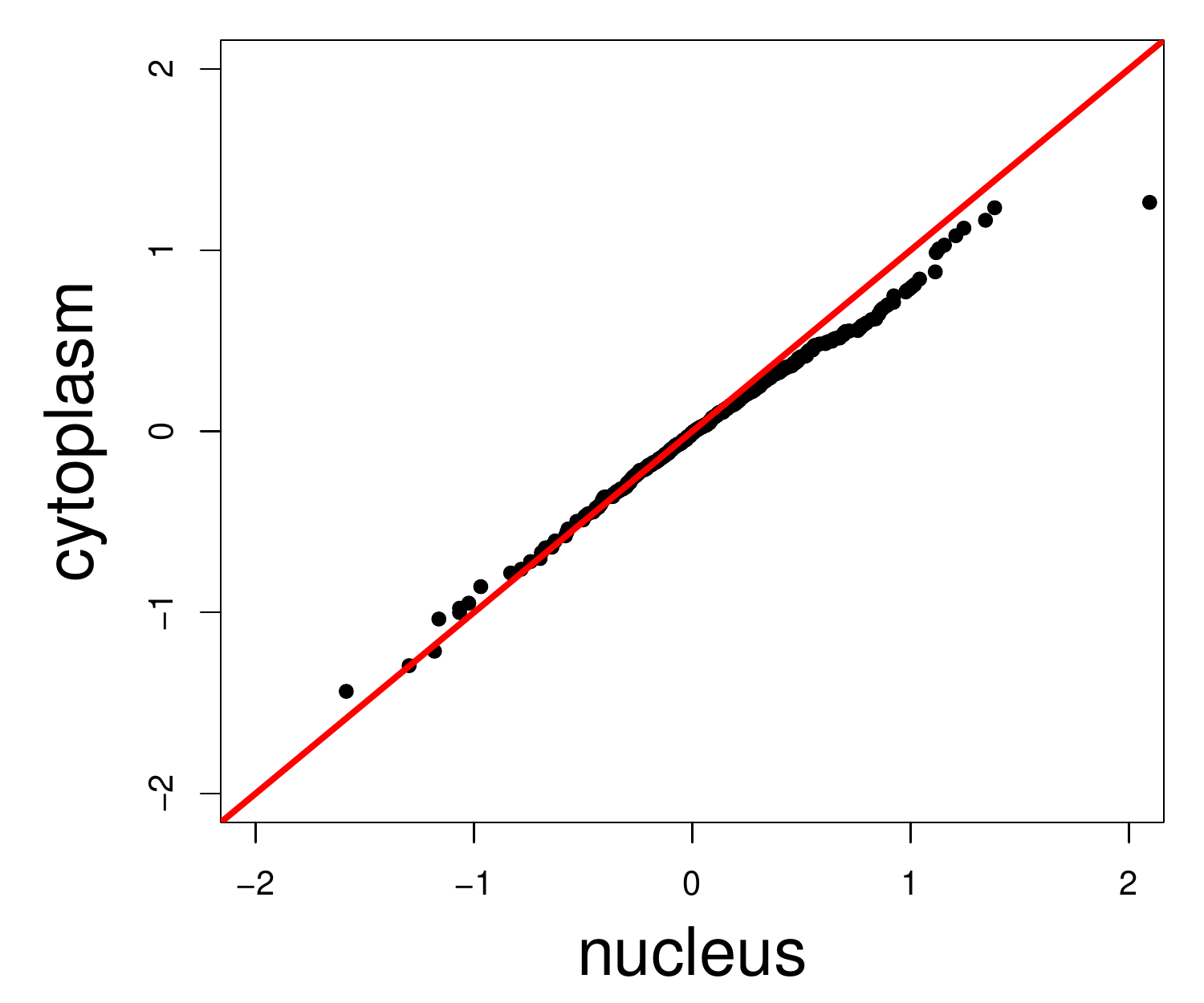}
    \end{minipage}
        \begin{minipage}{5cm}
    \centering
\includegraphics[clip, trim = 0cm 0cm 0cm 0cm, width  = 5cm]{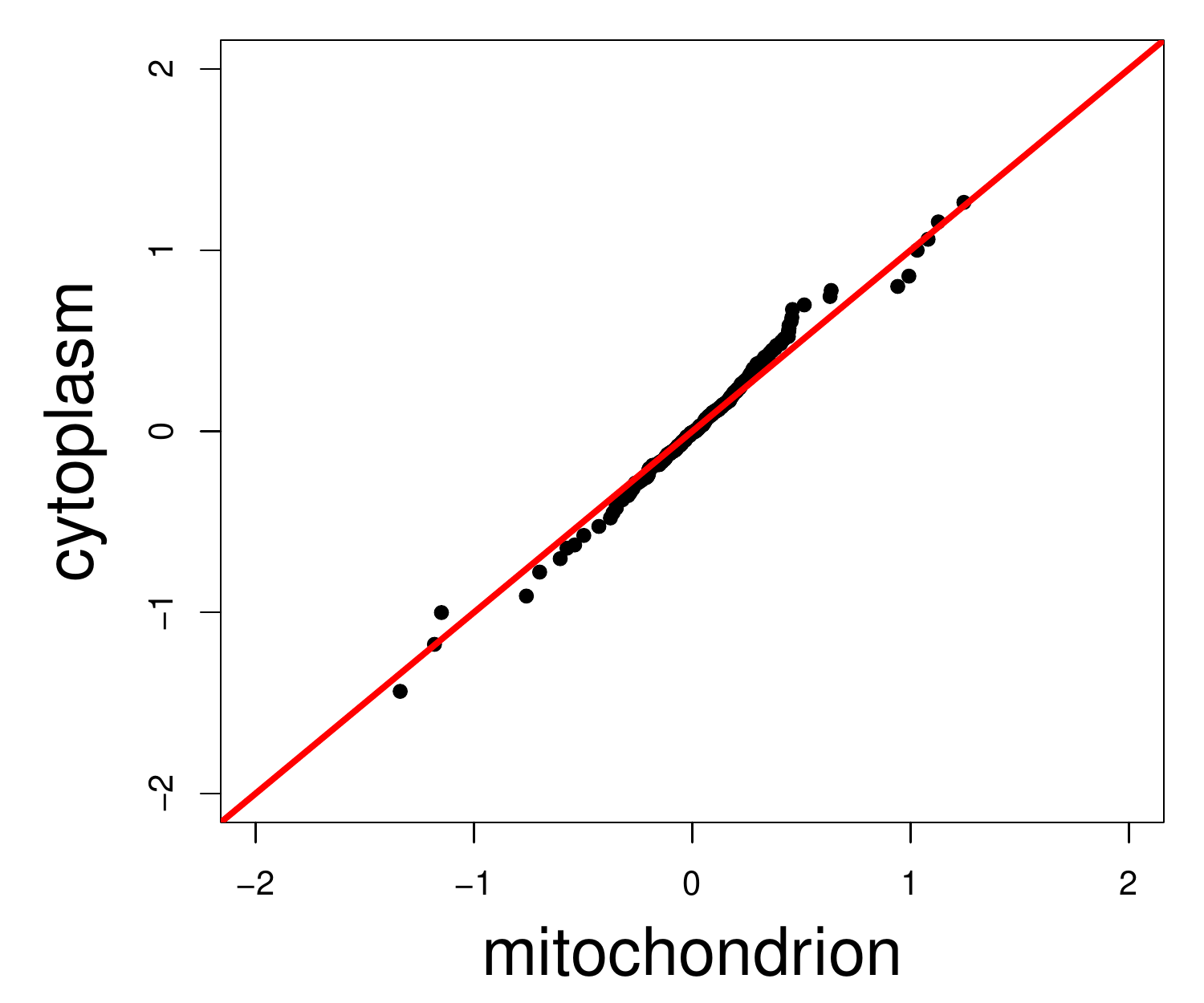}
    \end{minipage}
\caption{{Quantile-quantile plots of the abundance differences across internal negative control subgroups.
We partition the internal negative controls into three subgroups according to the subcellular location keywords ``nuclear'', ``mitochondrial'', and ``cytoplasmic''.
The K-S test p-values for all three pair-wise subgroup comparisons are non-significant at level $0.05$.
}}
    \label{fig:nc.subgroup}
\end{figure}

\subsection{Choice of empirical null}
\label{sec:choice-empir-null}

We investigate different choices of the empirical null distributions
in this proteomic dataset. We apply different normally distributed
null distributions obtained using the method described next. For each
$\calJ \in \{\calI, \calI \cup
\calI_{\text{nc}}, \calI_{\text{nc}}\}$ (corresponding to using,
respectively, all
test statistics under investigation, all test statistics observed, and
only the negative control statistics), we consider three different
estimators of the mean $\mu$ and standard deviation $\sigma$ of the null
distribution:
\begin{enumerate}
\item MAD1: $\mu = 0$, $\sigma =
  \sqrt{\{\texttt{MAD}((\responseTreatment{i})_{i \in \calJ})\}^2 +
    \{\texttt{MAD}((\responseControl{i})_{i \in \calJ})^2\}}$, where
  \texttt{mad} computes the median of the absolute deviations from the
  median multiplied by a factor of 1.4826 that ensures consistency for the
  normal distribution;
\item MAD2: $\mu = 0$, $\sigma =
  \texttt{MAD}((\responseTreatment{i} - \responseControl{i})_{i \in
    \calJ})$;
\item Efron: the method described in \textcite{efron2004large}, in
  which a Poisson regression is first applied to
  $(\responseTreatment{i} - \responseControl{i})_{i \in \calJ}$ to
  estimate the density function and $\mu$ and $\sigma$ are
  then obtained from the mode and the half-width of the center peak of
  the estimated density function.
\end{enumerate}
We then compute the one-sided p-values using each empirical null
distribution. Our proposed \nickname~p-values correspond to using the
empirical cumulative distribution function (ECDF) of
$(\responseTreatment{i} - \responseControl{i})_{i \in
    \calJ}$ for $\calJ = \calI_{\text{nc}}$.
To investigate the performance of each empirical null, we used
the Kolmogorov-Smirnov and Anderson-Darling tests to
assess whether the p-values that are between 0.5 and 0.99 are
approximately uniformly distributed.

\begin{table}[t]
  \centering
\caption{{A comparison of different choices of the empirical null
    distribution. NC: Negative controls; KS test: p-value of the
    Kolmogorov-Smirnov test of uniformity; AD test:
    p-value of the Anderson-Darling test of uniformity;
    \BH~rejections: number of rejections using
    the \BH~procedure with target \FDR~level $q = 0.2$.}}
\label{tab:rejection}
\begin{tabular}{lllccc}
\toprule
Statistics & Method & Empirical null & KS test & AD test & \BH~rejections \\
\midrule
Testing ($\calI$) & MAD1  & $\text{N}(0,0.4)$ & $< 10^{-16}$ & $1.9 \times 10^{-6}$ & $102$ \\
 & MAD2  & $\text{N}(0,0.3)$ & $3.3 \times 10^{-16}$ & $2.0 \times 10^{-6}$ & $144$ \\
 & Efron  & $\text{N}(-0.02,0.22)$ & $1.8 \times 10^{-5}$ & $2.5
                                                             \times 10^{-5}$ & $200$ \\
\midrule
All ($\calI \cup \calI_{\text{nc}}$) & MAD1  & $\text{N}(0,0.39)$ &
                                                              $<10^{-16}$
& $1.9 \times 10^{-6}$ & $103$ \\
 & MAD2  & $\text{N}(0,0.23)$ & $1.5 \times 10^{-5}$ & $5.3 \times 10^{-6}$ & $182$ \\
 & Efron  & $\text{N}(-0.05,0.21)$ & $4.5 \times 10^{-5}$ & $3.5
                                                            \times 10^{-5}$ & $234$ \\
\midrule
NC ($\calI_{\text{nc}}$) & MAD1  & $\text{N}(0,0.4)$ &
                                                                     $<10^{-16}$ & $1.9
                                                                \times
                                                                10^{-6}$ & $102$ \\
 & MAD2  & $\text{N}(0,0.2)$ & $0.004$ & $0.002$ & $211$ \\
           & Efron  & $\text{N}(-0.02,0.18)$ & $0.039$ & $0.023$ & $240$ \\
 & ECDF (RANC)  &  & $0.019$ & $0.063$ & $214$ \\
\bottomrule
\end{tabular}
\end{table}

\Cref{tab:rejection} reports the empirical null obtained by each
method above, the p-values from the two tests of uniformity, and the
number of rejections made by the \BH~procedure ($q = 0.2$). Among
these empirical nulls, only Efron's method applied to
negative controls and the RANC p-values procedure produce p-values
whose bulk are nearly uniformly distributed; see also
\Cref{fig:empirical-null} in the Appendix. Between these methods,
Efron's method is more aggressive and produces more
rejections. However, a normal distribution does not seem to fit the
distribution of negative controls very well; see \Cref{fig:qq-efron}
in the Appendix.

\section{Discussion}\label{sec:discussion}

Motivated by a real proteomic data anslysis, we have suggested three
model-free methods for simultaneous hypothesis testing using
internal negative controls. They can be used to control various
multiple testing error rates and have appealing theoretical
properties. Moreover, these methods offer competitive practical
performane as long as there are a decent number of number of negative
controls. Another attractive property is that the methods proposed
here are all invariant to monotone transformations of the data and
requires no subjective evaluation of the goodness-of-fit of the
empirical null distribution.

As mentioned in the Introduction, several recent articles have proposed
similar ideas about using negative controls in multiple testing
\parencite{bates21_testin_outlier_with_confor_p_values,mary22_semi_super_multip_testin,soloff22_edge_discov}. These
proposals arise from different applied domains and different
terminologies are used. An advantage to use ``negative control'' to
refer to the observations that resemble the null is that negative
control is an integral component of scientific methods and the nature
of the method can be immediately understood by across different
contexts. As a consequence, it is straightforward to understand the
assumptions involved and assess them in practice. In fact, this
prompted us to develop weaker theoretical conditions that allow
dependent test statistics (see \Cref{rem:bates-condition}) or
misspecfified negative controls (see \Cref{rem:one-sided}).

Of course, the statistical power of \nickname~p-values depend closely on
the quantity and quality of the internal negative controls. With a
moderate number of negative controls, the \BH~procedure
applied to the RANC p-values rejects almost as many hypotheses as the
\BH~procedure applied to the ``oracle'' p-values, which are calculated
using the unknown true null
distribution. When the signals are strong, the FDR level is $q = 0.2$,
and
the proportion of non-nulls is $n_1/n = 0.1$, a good rule of thumb
for the number of negatve controls is $\NoNc \ge 2 \No/(\FDRLevel
\NoNonNull) = 100$; see also
\textcite{mary22_semi_super_multip_testin} for related discussion on the
number of negative controls.

Although the RANC p-value is robust to certain
misclassifications of the negative controls, using too many negative
controls of poor quality may lead to low power. Additionally, one can falsify
the crucial exchangeability assumption by considering subgroups of
negative controls defined by domain knowledge; see
\Cref{fig:nc.subgroup} in \Cref{sec:real.data} for an illustration of
this idea using the proteomic dataset.


For convenience, we have assumed throughout the article that there are
no ties among the test statistics. Although we regard this assumption
as inconsequential for most practical applications, one may also
consider randomly breaking ties in the definition of
\nickname~p-values. It
can be easily shown that \Cref{prop:pvalue.validity} still holds.
However, it remains unclear if \Cref{prop:PRDS} holds
as our proof assumes that the
test statistics have continuous distributions.
Also, the empirical process argument in \Cref{sec:empirical.process}
does not directly extend because the process $\ncFalsePositive(t)$
may have large jumps in the presence of ties. See the supplementary
materials of \textcite{bates21_testin_outlier_with_confor_p_values}
for related discussion.

\section*{Acknowledgements}
Qingyuan Zhao and Zijun Gao are partly supported by EPSRC (grant
EP/V049968/1). We thank Jing Ren and Zora
Chan for bringing the proteomic application to us and Jiefu Li for
sharing the dataset analyzed in
\Cref{sec:motivating.data,sec:real.data}. We thank Rajen Shah and
Richard Samworth for helpful comments and Aaditya Ramdas for pointing
us to the conformal inference literature and in particular the paper
by \textcite{bates21_testin_outlier_with_confor_p_values}.

\section*{Data and computer programs}
The data and computer programs to reproduce the analysis and figures are available at \url{https://github.com/ZijunGao}.

\printbibliography

@article{grenander1956theory,
  title={On the theory of mortality measurement: part ii},
  author={Grenander, Ulf},
  journal={Scandinavian Actuarial Journal},
  volume={1956},
  number={2},
  pages={125--153},
  year={1956},
  publisher={Taylor \& Francis}
}

@article{durot2018limit,
  title={Limit theory in monotone function estimation},
  author={Durot, C{\'e}cile and Lopuha{\"a}, Hendrik P},
  journal={Statistical Science},
  volume={33},
  number={4},
  pages={547--567},
  year={2018},
    publisher={JSTOR}
}

@article{sun2007oracle,
  title={Oracle and adaptive compound decision rules for false discovery rate control},
  author={Sun, Wenguang and Cai, T Tony},
  journal={Journal of the American Statistical Association},
  volume={102},
  number={479},
  pages={901--912},
  year={2007},
  publisher={Taylor \& Francis}
}

@article{jin2007estimating,
  title={Estimating the null and the proportion of nonnull effects in large-scale multiple comparisons},
  author={Jin, Jiashun and Cai, T Tony},
  journal={Journal of the American Statistical Association},
  volume={102},
  number={478},
  pages={495--506},
  year={2007},
  publisher={Taylor \& Francis}
}

@article{schweder1982plots,
  title={Plots of p-values to evaluate many tests simultaneously},
  author={Schweder, Tore and Spj{\o}tvoll, Eil},
  journal={Biometrika},
  volume={69},
  number={3},
  pages={493--502},
  year={1982},
  publisher={Oxford University Press}
}

@article{hochberg1990more,
  title={More powerful procedures for multiple significance testing},
  author={Hochberg, Yosef and Benjamini, Yoav},
  journal={Statistics in medicine},
  volume={9},
  number={7},
  pages={811--818},
  year={1990},
  publisher={Wiley Online Library}
}

@article{swanepoel1999limiting,
  title={The limiting behavior of a modified maximal symmetric $2s$-spacing with applications},
  author={Swanepoel, Jan WH},
  journal={The Annals of Statistics},
  volume={27},
  number={1},
  pages={24--35},
  year={1999},
  publisher={Institute of Mathematical Statistics}
}

@article{hengartner1995finite,
  title={Finite-sample confidence envelopes for shape-restricted densities},
  author={Hengartner, Nicolas W and Stark, Philip B},
  journal={The Annals of Statistics},
  volume={23},
  number={2},
  pages={525--550},
  year={1995},
  publisher={JSTOR}
}

@article{dvoretzky1956asymptotic,
  title={Asymptotic minimax character of the sample distribution function and of the classical multinomial estimator},
  author={Dvoretzky, Aryeh and Kiefer, Jack and Wolfowitz, Jacob},
  journal={The Annals of Mathematical Statistics},
  pages={642--669},
  year={1956},
  volume={27},
  number={3},
  publisher={JSTOR}
}

@article{massart1990tight,
  title={The tight constant in the Dvoretzky-Kiefer-Wolfowitz inequality},
  author={Massart, Pascal},
  journal={The Annals of Probability},
  pages={1269--1283},
  year={1990},
  volume={18},
  number={3},
  publisher={JSTOR}
}

@book{massart2007concentration,
  title={Concentration inequalities and model selection: volume 1896 of Lecture Notes in Mathematics},
  author={Massart, Pascal},
  year={2007},
  publisher={Springer}
}

@article{baraud2016bounding,
  title={Bounding the expectation of the supremum of an empirical process over a (weak) VC-major class},
  author={Baraud, Yannick},
  journal={Electronic journal of statistics},
  volume={10},
  number={2},
  pages={1709--1728},
  year={2016},
  publisher={Institute of Mathematical Statistics and Bernoulli Society}
}

@article{hung14_proteom_mappin_human_mitoc_inter,
  author =       {Victoria Hung and Peng Zou and Hyun-Woo Rhee and
                  Namrata D. Udeshi and Valentin Cracan and Tanya
                  Svinkina and Steven A. Carr and Vamsi K. Mootha and
                  Alice Y. Ting},
  title =        {Proteomic Mapping of the Human Mitochondrial
                  Intermembrane Space in Live Cells Via Ratiometric
                  Apex Tagging},
  journal =      {Molecular Cell},
  volume =       55,
  number =       2,
  pages =        {332-341},
  year =         2014,
  doi =          {10.1016/j.molcel.2014.06.003},
  url =          {http://dx.doi.org/10.1016/j.molcel.2014.06.003},
  DATE_ADDED =   {Thu Jan 12 15:24:34 2023},
}

@article{shuster2022situ,
  title={In situ cell-type-specific cell-surface proteomic profiling in mice},
  author={Shuster, S Andrew and Li, Jiefu and Chon, URee and Sinantha-Hu, Miley C and Luginbuhl, David J and Udeshi, Namrata D and Carey, Dominique Kiki and Takeo, Yukari H and Xie, Qijing and Xu, Chuanyun and others},
  journal={Neuron},
  volume={110},
  pages={1--14},
  year={2022},
  publisher={Elsevier}
}

@article{li2020cell,
  title={Cell-surface proteomic profiling in the fly brain uncovers wiring regulators},
  author={Li, Jiefu and Han, Shuo and Li, Hongjie and Udeshi, Namrata D and Svinkina, Tanya and Mani, DR and Xu, Chuanyun and Guajardo, Ricardo and Xie, Qijing and Li, Tongchao and others},
  journal={Cell},
  volume={180},
  number={2},
  pages={373--386},
  year={2020},
  publisher={Elsevier}
}

@article{lippa2010exploring,
  title={Exploring the use of internal and externalcontrols for assessing microarray technical performance},
  author={Lippa, Katrice A and Duewer, David L and Salit, Marc L and Game, Laurence and Causton, Helen C},
  journal={BMC Research Notes},
  volume={3},
  number={349},
  pages={1--14},
  year={2010},
  publisher={BioMed Central}
}

@article{listgarten2013powerful,
  title={A powerful and efficient set test for genetic markers that handles confounders},
  author={Listgarten, Jennifer and Lippert, Christoph and Kang, Eun Yong and Xiang, Jing and Kadie, Carl M and Heckerman, David},
  journal={Bioinformatics},
  volume={29},
  number={12},
  pages={1526--1533},
  year={2013},
  publisher={Oxford University Press}
}

@article{slattery2011cofactor,
  title={Cofactor binding evokes latent differences in DNA binding specificity between Hox proteins},
  author={Slattery, Matthew and Riley, Todd and Liu, Peng and Abe, Namiko and Gomez-Alcala, Pilar and Dror, Iris and Zhou, Tianyin and Rohs, Remo and Honig, Barry and Bussemaker, Harmen J and others},
  journal={Cell},
  volume={147},
  number={6},
  pages={1270--1282},
  year={2011},
  publisher={Elsevier}
}

@article{nix2008empirical,
  title={Empirical methods for controlling false positives and estimating confidence in ChIP-Seq peaks},
  author={Nix, David A and Courdy, Samir J and Boucher, Kenneth M},
  journal={BMC bioinformatics},
  volume={9},
  number={523},
  pages={1--9},
  year={2008},
  publisher={Springer}
}

@article{wang2017confounder,
  title={Confounder adjustment in multiple hypothesis testing},
  author={Wang, Jingshu and Zhao, Qingyuan and Hastie, Trevor and Owen, Art B},
  journal={The Annals of Statistics},
  volume={45},
  number={5},
  pages={1863--1894},
  year={2017},
  publisher={Institute of Mathematical Statistics}
}

@article{parks2018using,
  title={Using controls to limit false discovery in the era of big data},
  author={Parks, Matthew M and Raphael, Benjamin J and Lawrence, Charles E},
  journal={BMC bioinformatics},
  volume={19},
  number={1},
  pages={1--8},
  year={2018},
  publisher={Springer}
}

@article{lipsitch2010negative,
  title={Negative controls: a tool for detecting confounding and bias in observational studies},
  author={Lipsitch, Marc and Tchetgen Tchetgen, Eric J and Cohen, Ted},
  journal={Epidemiology},
  volume={21},
  number={3},
  pages={383-388},
  year={2010},
  publisher={NIH Public Access}
}

@article{gagnon2012using,
  title={Using control genes to correct for unwanted variation in microarray data},
  author={Gagnon-Bartsch, Johann A and Speed, Terence P},
  journal={Biostatistics},
  volume={13},
  number={3},
  pages={539--552},
  year={2012},
  publisher={Oxford University Press}
}

@article{zhang2008model,
  title={Model-based analysis of ChIP-Seq (MACS)},
  author={Zhang, Yong and Liu, Tao and Meyer, Clifford A and Eeckhoute, J{\'e}r{\^o}me and Johnson, David S and Bernstein, Bradley E and Nusbaum, Chad and Myers, Richard M and Brown, Myles and Li, Wei and others},
  journal={Genome biology},
  volume={9},
  number={9},
  pages={1--9},
  year={2008},
  publisher={BioMed Central}
}

@article{song2007model,
  title={Model-based analysis of two-color arrays (MA2C)},
  author={Song, Jun S and Johnson, W Evan and Zhu, Xiaopeng and Zhang, Xinmin and Li, Wei and Manrai, Arjun K and Liu, Jun S and Chen, Runsheng and Liu, X Shirley},
  journal={Genome biology},
  volume={8},
  number={8},
  pages={1--13},
  year={2007},
  publisher={BioMed Central}
}

@article{whitt1980uniform,
  title={Uniform conditional stochastic order},
  author={Whitt, Ward},
  journal={Journal of Applied Probability},
  volume={17},
  number={1},
  pages={112--123},
  year={1980},
  publisher={Cambridge University Press}
}

@article{karlin1980classes,
  title={Classes of orderings of measures and related correlation inequalities. I. Multivariate totally positive distributions},
  author={Karlin, Samuel and Rinott, Yosef},
  journal={Journal of Multivariate Analysis},
  volume={10},
  number={4},
  pages={467--498},
  year={1980},
  publisher={Elsevier}
}

@article{holland1986conditional,
  title={Conditional association and unidimensionality in monotone latent variable models},
  author={Holland, Paul W and Rosenbaum, Paul R},
  journal={The Annals of Statistics},
  volume={14},
  number={4},
  pages={1523--1543},
  year={1986},
  publisher={JSTOR}
}

@article{genovese2004stochastic,
  title={A stochastic process approach to false discovery control},
  author={Genovese, Christopher and Wasserman, Larry},
  journal={The annals of statistics},
  volume={32},
  number={3},
  pages={1035--1061},
  year={2004},
  publisher={Institute of Mathematical Statistics}
}

@article{benjamini2000adaptive,
  title={On the adaptive control of the false discovery rate in multiple testing with independent statistics},
  author={Benjamini, Yoav and Hochberg, Yosef},
  journal={Journal of educational and Behavioral Statistics},
  volume={25},
  number={1},
  pages={60--83},
  year={2000},
  publisher={Sage Publications Sage CA: Los Angeles, CA}
}

@article{efron2004large,
  title={Large-scale simultaneous hypothesis testing: the choice of a null hypothesis},
  author={Efron, Bradley},
  journal={Journal of the American Statistical Association},
  volume={99},
  number={465},
  pages={96--104},
  year={2004},
  publisher={Taylor \& Francis}
}

@article{efron2001empirical,
  title={Empirical Bayes analysis of a microarray experiment},
  author={Efron, Bradley and Tibshirani, Robert and Storey, John D and Tusher, Virginia},
  journal={Journal of the American statistical association},
  volume={96},
  number={456},
  pages={1151--1160},
  year={2001},
  publisher={Taylor \& Francis}
}

@book{efron2012large,
  title={Large-scale inference: empirical Bayes methods for estimation, testing, and prediction},
  author={Efron, Bradley},
  year={2012},
  publisher={Cambridge University Press}
}

@article{goeman2021only,
  title={Only closed testing procedures are admissible for controlling false discovery proportions},
  author={Goeman, Jelle J and Hemerik, Jesse and Solari, Aldo},
  journal={The Annals of Statistics},
  volume={49},
  number={2},
  pages={1218--1238},
  year={2021},
  publisher={Institute of Mathematical Statistics}
}

@article{lehmann2005generalizations,
  title={Generalizations of the familywise error rate},
  author={Lehmann, Erich Leo and Romano, Joseph P},
  journal={The Annals of Statistics},
  volume={33},
  number={3},
  pages={1138--1154},
  year={2005},
  publisher={Institute of Mathematical Statistics}
}

@article{goeman2011multiple,
  title={Multiple testing for exploratory research},
  author={Goeman, Jelle J and Solari, Aldo},
  journal={Statistical Science},
  volume={26},
  number={4},
  pages={584--597},
  year={2011},
  publisher={Institute of Mathematical Statistics}
}

@article{marcus1976closed,
  title={On closed testing procedures with special reference to ordered analysis of variance},
  author={Marcus, Ruth and Eric, Peritz and Gabriel, K Ruben},
  journal={Biometrika},
  volume={63},
  number={3},
  pages={655--660},
  year={1976},
  publisher={Oxford University Press}
}

@article{bretz2009graphical,
  title={A graphical approach to sequentially rejective multiple test procedures},
  author={Bretz, Frank and Maurer, Willi and Brannath, Werner and Posch, Martin},
  journal={Statistics in medicine},
  volume={28},
  number={4},
  pages={586--604},
  year={2009},
  publisher={Wiley Online Library}
}

@book{fisher1925statistical,
  title={Statistical methods for research workers},
  author={Fisher, Ronald Aylmer},
  year={1925},
  publisher={Edinburgh, UK: Oliver and Boyd}
}

@article{simes1986improved,
  title={An improved Bonferroni procedure for multiple tests of significance},
  author={Simes, R John},
  journal={Biometrika},
  volume={73},
  number={3},
  pages={751--754},
  year={1986},
  publisher={Oxford University Press}
}

@article{sarkar1997simes,
  title={The Simes method for multiple hypothesis testing with positively dependent test statistics},
  author={Sarkar, Sanat K and Chang, Chung-Kuei},
  journal={Journal of the American Statistical Association},
  volume={92},
  number={440},
  pages={1601--1608},
  year={1997},
  publisher={Taylor \& Francis}
}

@article{benjamini1995controlling,
  title={Controlling the false discovery rate: a practical and powerful approach to multiple testing},
  author={Benjamini, Yoav and Hochberg, Yosef},
  journal={Journal of the Royal statistical society: series B (Methodological)},
  volume={57},
  number={1},
  pages={289--300},
  year={1995},
  publisher={Wiley Online Library}
}

@article{benjamini2001control,
  title={The control of the false discovery rate in multiple testing under dependency},
  author={Benjamini, Yoav and Yekutieli, Daniel},
  journal={The Annals of Statistics},
  volume={29},
  number={4},
  pages={1165--1188},
  year={2001},
  publisher={Institute of Mathematical Statistics}
}

@article{storey2002direct,
  title={A direct approach to false discovery rates},
  author={Storey, John D},
  journal={Journal of the Royal Statistical Society: Series B (Statistical Methodology)},
  volume={64},
  number={3},
  pages={479--498},
  year={2002},
  publisher={Wiley Online Library}
}

@article{storey2004strong,
	title={Strong control, conservative point estimation and simultaneous conservative consistency of false discovery rates: a unified approach},
	author={Storey, John D and Taylor, Jonathan E and Siegmund, David},
	journal={Journal of the Royal Statistical Society: Series B (Statistical Methodology)},
	volume={66},
	number={1},
	pages={187--205},
	year={2004},
	publisher={Wiley Online Library}
}

@article{zhao2019multiple,
  title={Multiple testing when many p-values are uniformly conservative, with application to testing qualitative interaction in educational interventions},
  author={Zhao, Qingyuan and Small, Dylan S and Su, Weijie},
  journal={Journal of the American Statistical Association},
  volume={114},
  number={527},
  pages={1291--1304},
  year={2019},
  publisher={Taylor \& Francis}
}

@article{leek10_tackl_wides_critic_impac_batch,
  author =       {Jeffrey T. Leek and Robert B. Scharpf and H{\'e}ctor
                  Corrada Bravo and David Simcha and Benjamin Langmead
                  and W. Evan Johnson and Donald Geman and Keith
                  Baggerly and Rafael A. Irizarry},
  title =        {Tackling the Widespread and Critical Impact of Batch
                  Effects in High-Throughput Data},
  journal =      {Nature Reviews Genetics},
  volume =       11,
  number =       10,
  pages =        {733-739},
  year =         2010,
  doi =          {10.1038/nrg2825},
  url =          {http://dx.doi.org/10.1038/nrg2825},
  DATE_ADDED =   {Thu Jan 12 09:39:03 2023},
}

@article{holm1979simple,
  title={A simple sequentially rejective multiple test procedure},
  author={Holm, Sture},
  journal={Scandinavian journal of statistics},
  volume={6},
  number={2},
  pages={65-70},
  year={1979},
  publisher={JSTOR}
}

@article{wiens03_fixed_sequen_bonfer_proced_testin_multip_endpoin,
  author =       {Brian L. Wiens},
  title =        {A Fixed Sequence Bonferroni Procedure for Testing
                  Multiple Endpoints},
  journal =      {Pharmaceutical Statistics},
  volume =       2,
  number =       3,
  pages =        {211-215},
  year =         2003,
  doi =          {10.1002/pst.64},
  url =          {http://dx.doi.org/10.1002/pst.64},
  DATE_ADDED =   {Thu Jan 12 11:17:18 2023},
}

@article{hellwege17_popul_strat_genet_assoc_studies,
  author =       {Jacklyn N. Hellwege and Jacob M. Keaton and Ayush
                  Giri and Xiaoyi Gao and Digna R. Velez Edwards and
                  Todd L. Edwards},
  title =        {Population Stratification in Genetic Association
                  Studies},
  journal =      {Current Protocols in Human Genetics},
  volume =       95,
  number =       1,
  pages =        {},
  year =         2017,
  doi =          {10.1002/cphg.48},
  url =          {http://dx.doi.org/10.1002/cphg.48},
  DATE_ADDED =   {Thu Jan 12 13:53:15 2023},
}

@article{price06_princ_compon_analy_correc_strat,
  author =       {Alkes L Price and Nick J Patterson and Robert M
                  Plenge and Michael E Weinblatt and Nancy A Shadick
                  and David Reich},
  title =        {Principal Components Analysis Corrects for
                  Stratification in Genome-Wide Association Studies},
  journal =      {Nature Genetics},
  volume =       38,
  number =       8,
  pages =        {904-909},
  year =         2006,
  doi =          {10.1038/ng1847},
  url =          {http://dx.doi.org/10.1038/ng1847},
  DATE_ADDED =   {Thu Jan 12 13:53:55 2023},
}

@article{leek07_captur_heter_gene_expres_studies,
  author =       {Jeffrey T Leek and John D Storey},
  title =        {Capturing Heterogeneity in Gene Expression Studies
                  By Surrogate Variable Analysis},
  journal =      {PLoS Genetics},
  volume =       3,
  number =       9,
  pages =        {e161},
  year =         2007,
  doi =          {10.1371/journal.pgen.0030161},
  url =          {http://dx.doi.org/10.1371/journal.pgen.0030161},
  DATE_ADDED =   {Thu Jan 12 14:05:33 2023},
}

@article{wang17_confoun_adjus_multip_hypot_testin,
  author =       {Jingshu Wang and Qingyuan Zhao and Trevor Hastie and
                  Art B. Owen},
  title =        {Confounder Adjustment in Multiple Hypothesis
                  Testing},
  journal =      {The Annals of Statistics},
  volume =       45,
  number =       5,
  pages =        {},
  year =         2017,
  doi =          {10.1214/16-aos1511},
  url =          {http://dx.doi.org/10.1214/16-AOS1511},
  DATE_ADDED =   {Thu Jan 12 14:07:19 2023},
}

@article{miao18_ident_causal_effec_with_proxy,
  author =       {Wang Miao and Zhi Geng and Tchetgen Tchetgen, Eric J},
  title =        {Identifying Causal Effects With Proxy Variables of
                  an Unmeasured Confounder},
  journal =      {Biometrika},
  volume =       105,
  number =       4,
  pages =        {987-993},
  year =         2018,
  doi =          {10.1093/biomet/asy038},
  url =          {http://dx.doi.org/10.1093/biomet/asy038},
  DATE_ADDED =   {Thu Jan 12 14:51:36 2023},
}

@article{tchetgen20_introd_to_proxim_causal_learn,
  author =       {Tchetgen, Eric J Tchetgen and Ying, Andrew and Cui,
                  Yifan and Shi, Xu and Miao, Wang},
  title =        {An Introduction To Proximal Causal Learning},
  journal =      {},
  year =         2020,
  url =          {http://arxiv.org/abs/2009.10982v1},
  archivePrefix ={arXiv},
  eprint =       {2009.10982},
  primaryClass = {stat.ME},
}

@article{hochberg88_sharp_bonfer_proced_multip_tests_signif,
  author =       {Yosef Hochberg},
  title =        {A Sharper Bonferroni Procedure for Multiple Tests of
                  Significance},
  journal =      {Biometrika},
  volume =       75,
  number =       4,
  pages =        {800-802},
  year =         1988,
  doi =          {10.1093/biomet/75.4.800},
  url =          {http://dx.doi.org/10.1093/biomet/75.4.800},
  DATE_ADDED =   {Fri Jan 13 15:59:26 2023},
}

@article{hommel88_stagew_rejec_multip_test_proced,
  author =       {G. Hommel},
  title =        {A Stagewise Rejective Multiple Test Procedure Based
                  on a Modified Bonferroni Test},
  journal =      {Biometrika},
  volume =       75,
  number =       2,
  pages =        {383-386},
  year =         1988,
  doi =          {10.1093/biomet/75.2.383},
  url =          {http://dx.doi.org/10.1093/biomet/75.2.383},
  DATE_ADDED =   {Fri Jan 13 16:00:45 2023},
}

@article{genovese06_exceed_contr_false_discov_propor,
  author =       {Christopher R Genovese and Larry Wasserman},
  title =        {Exceedance Control of the False Discovery
                  Proportion},
  journal =      {Journal of the American Statistical Association},
  volume =       101,
  number =       476,
  pages =        {1408-1417},
  year =         2006,
  doi =          {10.1198/016214506000000339},
  url =          {http://dx.doi.org/10.1198/016214506000000339},
  DATE_ADDED =   {Fri Jan 13 16:41:48 2023},
}

@article{bates21_testin_outlier_with_confor_p_values,
  author =       {Bates, Stephen and Cand{\`e}s, Emmanuel and Lei,
                  Lihua and Romano, Yaniv and Sesia, Matteo},
  title =        {Testing for Outliers With Conformal P-Values},
  journal =      {Annals of Statistics},
  year =         {2021},
  url =          {http://arxiv.org/abs/2104.08279v3},
  archivePrefix ={arXiv},
  eprint =       {2104.08279v3},
  primaryClass = {stat.ME},
  note =         {to appear}
}

@article{marandon22_machin_learn_meets_false_discov_rate,
  author =       {Marandon, Ariane and Lei, Lihua and Mary, David and
                  Roquain, Etienne},
  title =        {Machine Learning Meets False Discovery Rate},
  journal =      {CoRR},
  year =         2022,
  url =          {http://arxiv.org/abs/2208.06685v2},
  archivePrefix ={arXiv},
  eprint =       {2208.06685v2},
  primaryClass = {stat.ME},
}

@article{mary22_semi_super_multip_testin,
  author =       {David Mary and Etienne Roquain},
  title =        {Semi-Supervised Multiple Testing},
  journal =      {Electronic Journal of Statistics},
  volume =       16,
  number =       2,
  pages =        {},
  year =         2022,
  doi =          {10.1214/22-ejs2050},
  url =          {http://dx.doi.org/10.1214/22-EJS2050},
  DATE_ADDED =   {Thu Mar 9 12:06:24 2023},
}

@article{weinstein17_power_predic_analy_knock_with_lasso_statis,
  author =       {Weinstein, Asaf and Barber, Rina and Candes,
                  Emmanuel},
  title =        {A Power and Prediction Analysis for Knockoffs With
                  Lasso Statistics},
  journal =      {CoRR},
  year =         2017,
  url =          {http://arxiv.org/abs/1712.06465v1},
  archivePrefix ={arXiv},
  eprint =       {1712.06465},
  primaryClass = {stat.ME},
}

@article{soloff22_edge_discov,
  author =       {Soloff, Jake A. and Xiang, Daniel and Fithian,
                  William},
  title =        {The Edge of Discovery: Controlling the Local False
                  Discovery Rate At the Margin},
  journal =      {},
  year =         2022,
  url =          {http://arxiv.org/abs/2207.07299v1},
  archivePrefix ={arXiv},
  eprint =       {2207.07299},
  primaryClass = {stat.ME},
}

@book{vovk2005algorithmic,
  title={Algorithmic learning in a random world},
  author={Vovk, Vladimir and Gammerman, Alexander and Shafer, Glenn},
  volume={29},
  year={2005},
  publisher={Springer}
}

@article{angelopoulos2021gentle,
  title={A gentle introduction to conformal prediction and distribution-free uncertainty quantification},
  author={Angelopoulos, Anastasios N and Bates, Stephen},
  journal={arXiv preprint arXiv:2107.07511},
  year={2021}
}

@article{zhang2022randomization,
  title={What is a randomization test?},
  author={Zhang, Yao and Zhao, Qingyuan},
  journal={arXiv preprint arXiv:2203.10980},
  year={2022}
}

@article{liang2022integrative,
  title={Integrative conformal p-values for powerful out-of-distribution testing with labeled outliers},
  author={Liang, Ziyi and Sesia, Matteo and Sun, Wenguang},
  journal={arXiv preprint arXiv:2208.11111},
  year={2022}
}

@article{diaconis1980finite,
  title={Finite exchangeable sequences},
  author={Diaconis, Persi and Freedman, David},
  journal={The Annals of Probability},
  volume={8},
  number={4},
  pages={745-764},
  year={1980},
}

\appendix
\renewcommand\thefigure{\thesection.\arabic{figure}}
\counterwithin{figure}{section}

\section{Technical proofs}\label{sec:appendix}


\subsection{Results in \Cref{sec:pval}}

\begin{proof}[Proof of \Cref{prop:pvalue.validity}]
\textbf{Step 1}. We prove the validity if $\cdfTestStatistics{i}(t) = \cdfTestStatistics{j}(t)$ for all $j \in \hypothesisIndex{\text{nc}}$ and $t \in \RR$.
By condition~\ref{vali:assu:one.exchangeable.general}, $(\testStatistics{j})_{j \in \{i\} \cup \hypothesisIndex{\text{nc}}}$ is exchangeable.
Since there are no ties among $(\testStatistics{j})_{j \in \{i\} \cup \hypothesisIndex{\text{nc}}}$ a.s., then the rank of $\testStatistics{i}$ among $(\testStatistics{j})_{j \in \{i\} \cup \hypothesisIndex{\text{nc}}}$ is uniformly distributed on $\{1, 2, \ldots, 1+\NoNc\}$. By Definition~\eqref{eq:pval},
for $0 \le t \le 1$,
\begin{align*}
    \PP(\pval{i} \le t)
    &= \PP\left(\sum_{j \in \{i\} \cup \hypothesisIndex{\text{nc}}} \1_{\{\testStatistics{j} \le \testStatistics{i}\}}
    \le t(1+\NoNc)\right) \\
    &= \sum_{k=1}^{\lfloor t(1+\NoNc) \rfloor}\PP\left(\sum_{j \in \{i\} \cup\hypothesisIndex{\text{nc}}} \1_{\{\testStatistics{j}
    \le \testStatistics{i}\}}
    = k \right)
    = \frac{\lfloor t(1+\NoNc)\rfloor}{1+\NoNc}
    \le t.
\end{align*}
If we break ties randomly, $(\testStatistics{j})_{j \in \{i\} \cup \hypothesisIndex{\text{nc}}}$ is still exchangeable and the rank of $(\testStatistics{j})_{j \in \{i\} \cup \hypothesisIndex{\text{nc}}}$ is a permutation of $1$, $2$, $\ldots$, $\NoNc+1$.
The marginal distribution of the rank of $\testStatistics{i}$ is still uniform and the proof holds.

\textbf{Step 2}. We prove the validity if $\cdfTestStatistics{i}(t) \le \cdfTestStatistics{j}(t)$ for all $j \in \hypothesisIndex{\text{nc}}$ and $t \in \RR$.
By condition~\ref{vali:assu:one.exchangeable.general} and step 1, the \nickname~p-value $\pval{i}'$ based on $(\cdfTestStatistics{j}(\testStatistics{j}))_{j \in \{i\} \cup \hypothesisIndex{\text{nc}}}$ is valid.
Define $\testStatistics{j}' = \cdfTestStatistics{j}^{-1}(\cdfTestStatistics{j}(\testStatistics{j}))$,
where $\cdfTestStatistics{j}^{-1}(p) := \inf\{t: \cdfTestStatistics{j}(t) \ge p\}$.
Then $\testStatistics{j}' = \cdfTestStatistics{j}^{-1}(\cdfTestStatistics{j}(\testStatistics{j}'))$, $\testStatistics{j}' \le \testStatistics{j}$, and $\PP(\testStatistics{j}' < \testStatistics{j}) = 0$.
On $\{\testStatistics{j}' = \testStatistics{j}\}$, by condition~\ref{vali:assu:null.nc.conservative.general},
\begin{align*}
    \cdfTestStatistics{j}(\testStatistics{j}') \le \cdfTestStatistics{i}(\testStatistics{i})
    \implies
    \cdfTestStatistics{j}(\testStatistics{j}') \le \cdfTestStatistics{j}(\testStatistics{i})
    \implies
    \testStatistics{j}'
    = \cdfTestStatistics{j}^{-1}(\cdfTestStatistics{j}(\testStatistics{j}')) \le \cdfTestStatistics{j}^{-1}\cdfTestStatistics{j}(\testStatistics{i})
    \le \testStatistics{i}.
\end{align*}
Therefore,
\begin{align*}
    \PP(\pval{i} \le t)
    &= \PP\left(\sum_{j \in \{i\} \cup \hypothesisIndex{\text{nc}}} \1_{\{\testStatistics{j}' \le \testStatistics{i}\}}
    \le t(1+\NoNc)\right)
    \\
    &\le \PP\left(\sum_{j \in \{i\} \cup\hypothesisIndex{\text{nc}}} \1_{\{\cdfTestStatistics{j}(\testStatistics{j}') \le \cdfTestStatistics{i}(\testStatistics{i})\}}
    \le t(1+\NoNc)\right)
    =  \PP\left(\pval{i}' \le t\right)
    \le t.
\end{align*}
\end{proof}

\begin{proof}[Proof of \Cref{prop:PRDS}]

Without loss of generality, we assume $\nullHypothesisIndex = \{1,
\dotsc, \NoNull\}$ and $\hypothesisIndex{\text{nc}} = \{\No+1, \dotsc,
\No+\NoNc\}$.
When $\NoNull > 0$, $\testStatistics{1}$ corresponds to a true null.
Under both conditions in \Cref{prop:PRDS}, $\cdfTestStatistics{i} =
\cdfTestStatistics{j}$ for all $i,j \in  \nullHypothesisIndex \cup
\hypothesisIndex{\text{nc}}$. Without loss of generality, we assume
the null distribution is $U[0,1]$ so $\cdfTestStatistics{i}(t) = t$
for all $0 \leq t \leq 1$ and $i \in  \nullHypothesisIndex \cup
\hypothesisIndex{\text{nc}}$; otherwise, we can replace all statistics
$\testStatistics{i}$ by $\cdfTestStatistics{1}(\testStatistics{i})$.
Since the \PRDS~property is invariant by co-monotone transformations,
with a slight abuse of notation, it suffices to prove the
\PRDS~property is satisfied by the unnormalized ranks
\begin{align*}
    \qval{i} := \sum_{j \in \hypothesisIndex{\text{nc}}}
  \1_{\{\testStatistics{j} \le \testStatistics{i}\}},~i \in \hypothesisIndex{}.
\end{align*}
Under both sets of conditions, $(\testStatistics{i})_{i \in \{1\} \cup \hypothesisIndex{\text{nc}}}$ is exchangeable.
Therefore, for any $0 \le k \le \NoNc - 1$, $\PP(\qval{1} = k) =
\PP(\qval{1} = k+1) = 1/(1+\NoNc)$. Then
it suffices to show
\begin{align*}
    \PP\left(\qvalAll \in \calD, ~\qval{1} = k\right) \le \PP\left(\qvalAll \in \calD, ~\qval{1} = k+1\right)
\end{align*}
for all fixed $0 \le k \le \NoNc - 1$ and increasing set $\calD
\subseteq \RR^{\No+\NoNc}$.

Let $\bm C = (C_1,\dotsc,C_{\NoNc}) = (T_{\No+1}, \dotsc,
T_{\No+\NoNc})$ denote all negative control statistics. Let $C_{(1)} <
\dotsb < C_{(\NoNc)}$ be the order statistics. Let
$\boldsymbol{C}_{-(k+1)}$ denote all negative control test
statistics excluding $C_{(k+1)}$. It suffices to show that
\begin{align} \label{eq:intermediate-step-1}
    \PP\left(\qvalAll \in \calD, ~\qval{1} = k \mid \boldsymbol{C}_{-(k+1)}
  = \boldsymbol{c}_{-(k+1)}\right) \le \PP\left(\qvalAll \in \calD, ~\qval{1} =
  k+1 \mid \boldsymbol{C}_{-(k+1)} = \boldsymbol{c}_{-(k+1)}\right)
\end{align}
for all $\bm c_{-(k+1)}$. As a convention, let $c_{(0)} = 0$ and
$c_{(\NoNc+1)} = 1$. Define the bi-variate function
\begin{align*}
    f(t, c) = \PP\left(\pvalAll \in \calD \mid \boldsymbol{C}_{-(k+1)}
  = \boldsymbol{c}_{-(k+1)}, C_{(k+1)} = c, \testStatistics{1} = t
  \right) \quad \text{for}~c_{(k)} < c, t < c_{(k+2)}.
\end{align*}
Let $g(c,t)$ denote the density function of $(C_{(k+1)}, T_1)$ given
$\boldsymbol{C}_{-(k+1)} = \boldsymbol{c}_{-(k+1)}$ and $c_{(k)} < T_1
< c_{(k+2)}$.
To prove \eqref{eq:intermediate-step-1}, it suffices to show that
\[
  \int_{c_{(k)}}^{c_{(k+2)}} \int_{c_{(k)}}^{c_{(k+2)}} f(t,c) g(t,c) \1_{\{t
    < c\}} \,dt\, dc \leq \int_{c_{(k)}}^{c_{(k+2)}}
  \int_{c_{(k)}}^{c_{(k+2)}} f(t,c) g(t,c) \1_{\{t > c\}} \,dt\,dc.
\]
This follows from the next three Lemmas.

\begin{lemma} \label{lemm:symmetric-g}
The function $g$ is symmetric, i.e.\ $g(t,c) = g(c,t)$ for all
$c_{(k)} < c < t < c_{(k+2)}$.
\end{lemma}

\begin{lemma} \label{lemm:symmetric-inequality}
  $f(c,t) \leq f(t,c)$ for all $c_{(k)} < c < t < c_{(k+2)}$.
\end{lemma}

\begin{lemma}\label{lemm:PRDS}
Let $h: [0,1]^2 \to \RR$ be a bi-variate function and $h(x,y) \leq
h(y,x)$ for all $0 < x < y < 1$. Then
\begin{align*}
  \int_0^1 \int_0^1 h(x,y) \1_{\{x < y\}} \, dx \, dy \leq \int_0^1
  \int_0^1 h(x,y) \1_{\{x > y\}} \, dx \, dy
\end{align*}
\end{lemma}

\begin{proof}[Proof of \Cref{lemm:symmetric-g}]
  This is obviously true under the exchangeability condition in
  \ref{PRDS:assu:exchangeable}.
  For the first set of conditions, the
  independence conditions in \ref{PRDS:assu:all.nc.set.independent} and
  \ref{PRDS:assu:nc} imply that
\[
  \left(C_{(k+1)}, T_1\right) \mid \boldsymbol{C}_{-(k+1)}
  = \boldsymbol{c}_{-(k+1)} \sim U\left(\left[c_{(k)}, c_{(k+2)}\right]^2\right),
\]
where $U\left(\left[c_{(k)}, c_{(k+2)}\right]^2\right)$ is the uniform distribution over the
square $\left[c_{(k)}, c_{(k+2)}\right]^2$.
\renewcommand{\qedsymbol}{}
\end{proof}

\begin{proof}[Proof of \Cref{lemm:symmetric-inequality}]
  Under the first set of conditions, we claim that $f(t,c)$ is
  decreasing in $c$ and increasing in $t$, so the desired conclusion
  follows. The first observation follows from the definition of $\bm
  p$ and the assumption that $\mathcal{D}$ is increasing. The second
  claim follows from the PRDS property in
  condition~\ref{PRDS:assu:PRDS}. To see this, given $\bm C = \bm c$,
  the event $\bm p \in \mathcal{D}$ can be rewritten as an event $\bm
  T \in \mathcal{D}'$ where $\mathcal{D}' = \mathcal{D}'(\bm c)$ is an
  increasing set that depends on $\bm c$. Using conditions
  \ref{PRDS:assu:all.nc.set.independent} and \ref{PRDS:assu:PRDS}, we
  have
\begin{align*}
  f(t,c)=&\PP\left(\pvalAll \in \calD \mid \boldsymbol{C}_{-(k+1)}
  = \boldsymbol{c}_{-(k+1)}, C_{(k+1)} = c, \testStatistics{1} = t
  \right) \\
  =& \PP\left(\bm T \in \calD'(\bm c) \mid \boldsymbol{C}_{-(k+1)}
  = \boldsymbol{c}_{-(k+1)}, C_{(k+1)} = c, \testStatistics{1} = t
     \right) \\
  =& \PP\left(\bm T \in \calD'(\bm c) \mid \testStatistics{1} = t
  \right)
\end{align*}
is increasing in $t$.
Now consider the exchangeability condition
\ref{PRDS:assu:exchangeable}, which implies that $(T_2,\dotsc,T_{\No})$ has
the same conditional distribution given
\[
\boldsymbol{C}_{-(k+1)}
  = \boldsymbol{c}_{-(k+1)}, C_{(k+1)} = c, \testStatistics{1} = t,
\]
and given
\[
\boldsymbol{C}_{-(k+1)}
  = \boldsymbol{c}_{-(k+1)}, C_{(k+1)} = t, \testStatistics{1} = c,
\]
for all $c_{(k)} < c, t < c_{(k+2)}$. The conclusion then follows from
the fact that the \nickname~p-values $\bm p$ only become smaller when
we swap $C_{(k+1)} = c$ with $T_1 = t$ if $c < t$ and the assumption
that $\mathcal{D}$ is increasing.
\renewcommand{\qedsymbol}{}
\end{proof}

\begin{proof}[Proof of \Cref{lemm:PRDS}]
  The conclusion follows from rewritting one of the
  integrals as follows,
  \begin{align*}
    &\int_0^1 \int_0^1 h(x,y) \1_{\{x < y\}} \, dx \, dy - \int_0^1
      \int_0^1 h(x,y) \1_{\{x > y\}} \, dx \, dy \\
    =& \int_0^1 \int_0^1 h(x,y) \1_{\{x < y\}} \, dx \, dy - \int_0^1
       \int_0^1 h(y,x) \1_{\{y > x\}} \, dx \, dy \\
    =& \int_0^1 \int_0^1 (h(x,y) -h(y,x)) \1_{\{x < y\}} \, dx \, dy
    \\
    \leq& 0.
  \end{align*}
\renewcommand{\qedsymbol}{}
\end{proof}
\vspace{-2em}
\end{proof}

\begin{proposition}\label{prop:permutation}
    If the permutation test statistic $\permutationFunction{\cdot}$ satisfies
    \begin{enumerate}
        \item \label{prop:assu:monotone.invariance}$\permutationFunction{\testStatisticsAll} = \permutationFunction{g(\testStatisticsAll)}$ for arbitrary strictly increasing function $g(\cdot): \RR \to \RR$,  $g(\testStatisticsAll) = (g(\testStatistics{i}))_{i \in  \hypothesisIndex{} \cup \hypothesisIndex{\text{nc}}}$;
        \item \label{prop:assu:permutation.invariance} $\permutationFunction{\testStatisticsAll} = \permutationFunction{\testStatisticsAll_{P}}$, where $\testStatisticsAll_{P}$ is a permutation of $\testStatisticsAll$ where $(\testStatistics{i})_{i \in  \hypothesisIndex{} }$ are not exchanged with $(\testStatistics{j})_{j \in \hypothesisIndex{\text{nc}}}$,
    \end{enumerate}
    then
    $\permutationFunction{\testStatisticsAll}$ is a function of the set of \nickname~p-values $\{\pval{i}\}_{i \in \hypothesisIndex{}}$.
\end{proposition}

\begin{proof}[Proof of \Cref{prop:permutation}]
The mapping from $\testStatisticsAll$ with no ties to its rank $R(\testStatisticsAll)$ is a strictly increasing function.
Since $\permutationFunction{\testStatisticsAll} = \permutationFunction{g(\testStatisticsAll)}$ for arbitrary strictly increasing function $g(\cdot)$,
then $\permutationFunction{\testStatisticsAll} = \permutationFunction{R(\testStatisticsAll)}$.
For $\testStatisticsAll$ and $\testStatisticsAll'$, we can always permute $(\testStatistics{i})_{i \in  \hypothesisIndex{}}$ so that the ordering within $(\testStatistics{i})_{i \in  \hypothesisIndex{} }$ equals that of $(\testStatistics{i}')_{i \in  \hypothesisIndex{} }$. Similarly for $(\testStatistics{i})_{i \in  \hypothesisIndex{\text{nc}} }$.
Denote the permuted test statistics by $\testStatisticsAll_P$, and then $\{\pval{i}\}_{i \in \hypothesisIndex{}} = \{\pval{P,i}\}_{i \in \hypothesisIndex{}}$.
Since $\permutationFunction{\cdot}$ is invariant to permutations within $(\testStatistics{i})_{i \in  \hypothesisIndex{} }$ and $(\testStatistics{i})_{i \in  \hypothesisIndex{\text{nc}} }$, then $\permutationFunction{\testStatisticsAll} = \permutationFunction{\testStatisticsAll_P}$.
Note that $\{\pval{i}\}_{i \in \hypothesisIndex{}} = \{\pval{P,i}\}_{i \in \hypothesisIndex{}} = \{\pval{i}'\}_{i \in \hypothesisIndex{}}$ implies $R(\testStatisticsAll_P) = R(\testStatisticsAll')$, and further $\permutationFunction{\testStatisticsAll}
= \permutationFunction{\testStatisticsAll_P}
= \permutationFunction{R(\testStatisticsAll_P)}
= \permutationFunction{R(\testStatisticsAll')}
= \permutationFunction{\testStatisticsAll'}$.
Therefore,   $\permutationFunction{\testStatisticsAll}$ is a function of $\{\pval{i}\}_{i \in \hypothesisIndex{}}$.
\end{proof}

\subsection{Results in \Cref{sec:empirical.process}}\label{sec:results-empirical-process}

The next proposition suggests that $\widehat{\FDR}_{\lambda}(t)$ is a
conservative estimate of $\FDR(t)$. The proposition is valid for replacing the $+2$ by $+1$ in the numerator of $\widehat{\FDR}_{\lambda}(t)$.
\begin{proposition}\label{prop:FDR.estimate}
    Assume $(\testStatistics{i})_{i \in \hypothesisIndex{} \cup \hypothesisIndex{\text{nc}}}$ are independent and $\testStatistics{i} \stackrel{d}{=} \ncTestStatistics{j}$ for $i \in \nullHypothesisIndex$, $j \in \hypothesisIndex{\text{nc}}$, then
    $\EE\left[\widehat{\FDR}_{\lambda}(t) \right] \ge \FDR(t)$ when
    $\lambda = 1$.
\end{proposition}
\noindent

\begin{proof}[Proof of \Cref{prop:FDR.estimate}]

For $0 < \lambda < 1$, $\PP(V_{\text{nc}}(\lambda) = \NoNc) > 0$ and $\EE\left[\widehat{\FDR}_\lambda (t)\right] = \infty > \FDR(t)$.
Thus, it is left to show the inequality is true for $\lambda = 1$. For simplicity, we write $\widehat{\FDR}_1 (t)$ as $\widehat{\FDR} (t)$.
We denote the distribution of true null and negative control test statistics by $\cdfTestStatisticsNull(t)$.
Since $V_{\text{nc}}(t) \le \NoNc$,
\begin{align}\label{proof:eq:cdf.estimate}
    \EE\left[\frac{V_{\text{nc}}(t)+1}{\NoNc+1}\right]  \ge \EE\left[\frac{V_{\text{nc}}(t)}{\NoNc}\right] = \cdfTestStatisticsNull(t).
\end{align}
The result in \cite[Theorem 1]{storey2004strong} of p-values following $U[0,1]$ can extend to test statistics of arbitrary distribution with continuous CDF,
\begin{align}\label{proof:eq:FDR.estimate}
     \EE\left[\frac{\NoNull \cdfTestStatisticsNull(t) - V(t)}{R(t) \vee 1} \right] \ge 0.
\end{align}
Since the negative control test statistics are independent of the test statistics under investigation and $\No \ge \NoNull$,
\begin{align*}
    \EE\left[\widehat{\FDR}(t)\right]
    &=\EE\left[\EE\left[\widehat{\FDR}(t) \mid R(t)\right] \right]
    =\EE\left[\frac{\No \EE\left[\frac{V_{\text{nc}}(t)+1}{\NoNc+1} \mid  R(t)\right]}{R(t) \vee 1}  \right] \\
    &\ge \EE\left[\frac{\No \cdfTestStatisticsNull(t)}{R(t) \vee 1} \right]
    \ge \EE\left[\frac{V(t)}{R(t) \vee 1} \right] = \FDR(t).
\end{align*}
\end{proof}

\begin{proof}[Proof of \Cref{prop:FDR.equivalence}]

Let $\testStatistics{q}'$, $\testStatistics{q}$ be the largest test statistic under investigation rejected by applying the \BH~procedure to the \nickname~p-values and the empirical-process-based step-up procedure, respectively.
It is straightforward to show the event that $\testStatistics{q}'$ does not exist is equivalent to  the event that $\testStatistics{q}$ does not exist.
It is left to show the equivalence holds when both $\testStatistics{q}'$ and $\testStatistics{q}$ are well-defined.
On one hand,
\begin{align*}
    \frac{V_{\text{nc}}(\testStatistics{q}')+2}{\NoNc+1}  \wedge 1 \le \frac{R(\testStatistics{q}')}{\No} q
    \implies
    \frac{\No\frac{V_{\text{nc}}(\testStatistics{q}')+2}{\NoNc+1}}{R(\testStatistics{q}') \vee 1} \le q
    \implies \stoppingTime \ge \testStatistics{q}',
\end{align*}
which further implies $\testStatistics{q} \ge \testStatistics{q}'$ by the definition of $\testStatistics{q}$.
On the other hand, by the discussion before \Cref{prop:FDR.equivalence},
there exists $\varepsilon > 0$ such that $\stoppingTime - \varepsilon > \testStatistics{q}$, $V_{\text{nc}}(\stoppingTime - \varepsilon) \ge V_{\text{nc}}(\testStatistics{q})$,
$R(\stoppingTime - \varepsilon) = R(\testStatistics{q})$,
and $\widehat{\FDR}(\stoppingTime - \varepsilon) \le q$.
Then,
\begin{align*}
    \widehat{\FDR}(\testStatistics{q})
    = \frac{\No \cdot \frac{V_{\text{nc}}(\testStatistics{q})+2}{\NoNc+1} }{R(\testStatistics{q}) \vee 1}
    \le \frac{\No \cdot \frac{V_{\text{nc}}(\stoppingTime - \varepsilon)+2}{\NoNc+1} }{R(\stoppingTime - \varepsilon) \vee 1}
    = \widehat{\FDR}(\stoppingTime - \varepsilon)
    \le q.
\end{align*}
Therefore,  $\testStatistics{q} \le \testStatistics{q}'$.

\end{proof}

\begin{proof}[Proof of \Cref{prop:FDR}]
Without loss of generality, we assume the test statistics are supported on $(0,1)$.
Otherwise, we can always apply the transformation $t \to \frac{e^t}{1 + e^t}$.
Define a decreasing family of $\sigma$-algebras $\calF_t = \sigma(\falsePositive(s), \truePositive(s), \ncFalsePositive(s), ~t \le s \le 1)$ for $0 \le t \le 1$.
By definition, $M(t)$ in Eq.~\eqref{eq:super.martingale} is measurable with respect to $\calF_t$, $0 \le M(t) \le (\NoNc+1)\NoNull$, $M(0) = 0$, and $M(1) = \NoNull$.

We show $M(t)$ is a backward super-martingale with respect to $\{\calF_t\}$.
For $0 \le s < t$, by condition~\ref{FDR:assu:independent}, $(\falsePositive(s), \falsePositive(t))$, $(\falsePositive_{\text{nc}}(s), \falsePositive_{\text{nc}}(t))$, and $(\truePositive(s), \truePositive(t))$ are mutually independent.
On $\left\{\falsePositive(t) = ck, \ncFalsePositive(t) + 1 = c\right\}$ with non-zero probability,
\begin{align}\label{proof:eq:mtg2}
\begin{split}
    \EE\left[ \frac{\falsePositive(s)}{\ncFalsePositive(s) + 1} \mid \calF_t \right]
    &= \EE\left[ \frac{\falsePositive(s)}{\ncFalsePositive(s) + 1} \mid \falsePositive(t) = ck, ~\ncFalsePositive(t) + 1 = c \right]\\
    &=  \underbrace{\EE\left[V(s) \mid V(t) = ck\right]}_{:=\text{(I)}} \cdot \underbrace{\EE\left[\frac{1}{\ncFalsePositive(s) + 1} \mid \ncFalsePositive(t) + 1 = c \right]}_{:=\text{(II)}}.
\end{split}
\end{align}
By condition~\ref{FDR:assu:null.nc.uniformly.conservative}, there exists $0 \le p_{t}^s \le 1$ such that
\begin{align}\label{proof:eq:prob}
    \max_{i \in \nullHypothesisIndex }\frac{\cdfTestStatistics{i}(s)}{\cdfTestStatistics{i}(t)}
    \le  p_{t}^s
    \le \min_{j \in \hypothesisIndex{\text{nc}} }\frac{\cdfTestStatistics{j}(s)}{\cdfTestStatistics{j}(t)}.
\end{align}
By condition~\ref{FDR:assu:independent}, and Eq.~\eqref{proof:eq:prob},
\begin{align}\label{proof:eq:stochastic.dominance}
    \falsePositive(s) \mid \falsePositive(t) &\lesssim Z, \quad Z \sim \text{Binomial}(\falsePositive(t), p_t^s),\\ \label{proof:eq:stochastic.dominance.2}
    \frac{1}{1+\ncFalsePositive(s)} \mid \frac{1}{1+\ncFalsePositive(t)} &\lesssim \frac{1}{1+Z^{\text{nc}}}, \quad Z^{\text{nc}} \sim \text{Binomial}(\ncFalsePositive(t), p_t^s).
\end{align}
For term (I) in Eq.~\eqref{proof:eq:mtg2}, by Eq.~\eqref{proof:eq:stochastic.dominance},
\begin{align}\label{proof:eq:mtg(I)}
    \text{(I)}
    \le \falsePositive(t) \cdot p_t^s
    = ck p_t^s.
\end{align}
For term (II) in Eq.~\eqref{proof:eq:mtg2}, by Eq.~\eqref{proof:eq:stochastic.dominance.2},
\begin{align}\label{proof:eq:mtg(II)}
    \begin{split}
        \text{(II)}
        &\le \sum_{i=0}^{c-1} \frac{1}{i+1} \binom{c-1}{i} (p_t^s)^{i} (1-p_t^s)^{c-1-i}\\
        &= \frac{1}{cp_{t}^s}\sum_{i=0}^{c-1} \binom{c}{i+1} (p_t^s)^{i+1} (1-p_t^s)^{c-1-i}
        =\frac{1}{c p_{t}^s} \left(1 - (1-p_t^s)^{c} \right).
    \end{split}
\end{align}
Plug Eq.~\eqref{proof:eq:mtg(I)}, Eq.~\eqref{proof:eq:mtg(II)} into Eq.~\eqref{proof:eq:mtg2},
\begin{align*}
    \EE\left[ \frac{\falsePositive(s)}{\ncFalsePositive(s) + 1} \mid \calF_t \right]
    \le ck p_t^s \cdot \frac{1}{c p_{t}^s} \left(1 - (1-p_t^s)^{c} \right)
    \le k
    = \frac{\falsePositive(t)}{\ncFalsePositive(t) + 1}.
\end{align*}

Next, we show $\stoppingTimeLambda$ is a stopping time with respect to $\{\calF_t\}$.
By definition, if $\lambda = 1$, then for $0 \le t \le 1$,
\begin{align*}
    \{\stoppingTime > t\}
    = \left\{\sup\left\{s \ge 0: \frac{\frac{\No}{\NoNc+1}(\ncFalsePositive(s) + 2)}{(\falsePositive(s) + \truePositive(s)) \vee 1} \le q \right\} > t \right\},
\end{align*}
is $\calF_t$ measurable.
If $0 < \lambda < 1$, for $t \ge \lambda$, $\{\stoppingTimeLambda > t\} = \emptyset$; for $t < \lambda$, $\hat{\pi}(\lambda)$ is $\calF_t$ measurable.
Thus, $\stoppingTimeLambda$ is a stopping time regarding $\{\calF_t\}$.

Since $M(t)$ is bounded and thus uniformly integrable, we can apply the optional stopping time theorem
 \begin{align}\label{proof:eq:super.martingale}
    \EE\left[M\left(\stoppingTimeLambda\right) \mid \calF_1\right]
    \le M(1)
    = \NoNull.
\end{align}

We next show
\begin{align}\label{proof:eq:compare1}
    \frac{\falsePositive(\stoppingTimeLambda)}{(\falsePositive(\stoppingTimeLambda) + \truePositive(\stoppingTimeLambda)) \vee 1}
    \le \frac{q}{\hat{\pi}(\lambda)} \cdot \frac{\NoNc+1}{\No} \cdot \frac{\falsePositive(\stoppingTimeLambda)}{\ncFalsePositive(\stoppingTimeLambda) + 1}.
\end{align}
It suffices to discuss the case where $\stoppingTimeLambda > 0$.
Let $\calS_q = \left\{0 \le t \le \lambda: \frac{\hat{\pi}(\lambda)\No\frac{V_{\text{nc}}(t) + 2}{\NoNc+1}}{(V(t) + S(t)) \vee 1} \le q \right\}$, then $\stoppingTimeLambda = \sup \calS_q$.
If $\stoppingTimeLambda \in \calS_q$, then Eq.~\eqref{proof:eq:compare1} is obviously true.
If $\stoppingTimeLambda \notin \calS_q$, then there exists a sequence $(t_k) \subseteq \calS_q$ such that $t_k \to \stoppingTimeLambda$ as $k \to \infty$, and $\falsePositive(t_k) + \truePositive(t_k) \le \falsePositive(\stoppingTimeLambda) + \truePositive(\stoppingTimeLambda)$.
Since a Poisson process increases at most by one at a time, then $\lim_{k\to \infty}\ncFalsePositive(t_k) + 2 \ge \ncFalsePositive(\stoppingTimeLambda) + 1$. Therefore,
\begin{align*}
 q \cdot \frac{\NoNc+1}{\No} &\ge \lim_{k \to \infty}\frac{\hat{\pi}(\lambda)(\ncFalsePositive(t_k) + 2)}{(\falsePositive(t_k) + \truePositive(t_k)) \vee 1}
 \ge \frac{\hat{\pi}(\lambda)
 \left(\ncFalsePositive\left(\stoppingTimeLambda\right) + 1\right)}{\left(\falsePositive\left(\stoppingTimeLambda\right) + \truePositive\left(\stoppingTimeLambda\right)\right) \vee 1}.
\end{align*}

Finally, we compute the \FDR. By Eq.~\eqref{proof:eq:compare1},
\begin{align*}
    \FDR
    = \EE\left[\frac{\falsePositive(\stoppingTimeLambda)}{(\falsePositive(\stoppingTimeLambda) + \truePositive(\stoppingTimeLambda)) \vee 1} \mid \calF_1 \right]
    \le q \cdot \frac{\NoNc+1}{\No} \cdot \EE\left[\frac{1}{\hat{\pi}(\lambda)} \cdot \frac{\falsePositive(\stoppingTimeLambda)}{\ncFalsePositive(\stoppingTimeLambda) + 1} \mid \calF_1 \right].
\end{align*}
By the tower property and Eq.~\eqref{proof:eq:super.martingale},
\begin{align*}
    \EE\left[\frac{1}{\hat{\pi}(\lambda)} \cdot \frac{\falsePositive(\stoppingTimeLambda)}{\ncFalsePositive(\stoppingTimeLambda) + 1} \mid \calF_1 \right]
    = \EE\left[\frac{1}{\hat{\pi}(\lambda)} \cdot \EE\left[\frac{\falsePositive(\stoppingTimeLambda)}{\ncFalsePositive(\stoppingTimeLambda) + 1} \mid \calF_\lambda\right]\mid \calF_1 \right]
    \le \EE\left[\frac{1}{\hat{\pi}(\lambda)} \cdot \frac{\falsePositive(\lambda)}{\ncFalsePositive(\lambda) + 1} \mid \calF_1 \right].
\end{align*}
If $\lambda = 1$, then
\begin{align*}
    \FDR \le q \cdot \frac{\NoNc + 1}{\No} \cdot \frac{\NoNull}{\NoNc + 1} \le q.
\end{align*}
If $0 < \lambda < 1$, we plug in the definition of $\hat{\pi}(\lambda)$ and by condition~\ref{FDR:assu:independent},
    \begin{align*}
        \FDR
        &\le q \cdot \EE\left[\frac{\NoNc - \ncFalsePositive(\lambda)}{\No - \truePositive(\lambda) - \falsePositive(\lambda) + 1} \cdot \frac{\falsePositive(\lambda)}{\ncFalsePositive(\lambda) + 1} \mid \calF_1 \right]\\
        &= q \cdot \EE\left[\frac{\NoNc - \ncFalsePositive(\lambda)}{\ncFalsePositive(\lambda) + 1} \mid \calF_1 \right] \cdot \EE\left[\frac{\falsePositive(\lambda)}{\No - \truePositive(\lambda) - \falsePositive(\lambda) + 1} \mid \calF_1 \right].
    \end{align*}
    By a similar argument of Eq.~\eqref{proof:eq:mtg(II)},
    \begin{align*}
        \FDR
        \le q \cdot \frac{1-p_1^\lambda}{p_1^\lambda} (1 - (1 - p_1^\lambda)^{\NoNc}) \cdot \frac{p_1^\lambda}{1-p_1^\lambda} (1 - (p_1^\lambda)^{\NoNull})
        \le q.
    \end{align*}

\end{proof}

\begin{remark}
In the proof of \Cref{prop:FDR}, it may be tempting to use the alternative process
\[
  \frac{\falsePositive(t)}{(\ncFalsePositive(t) \vee 1)/\NoNc}
\]
as $M(t)$. However, this is not a super-martingale.
\end{remark}

\subsection{Results in
  \Cref{sec:localFDR}}\label{sec:results-localFDR}

\begin{proposition}\label{prop:convergence.rate.weak}
For any $\cdfTestStatistics{}$, $\cdfTestStatisticsNull$, and $0 < \varepsilon < 1$,
\begin{align*}
    \PP\left(\BayesRiskLambda(\EmpiricalThreshold) - \BayesRiskLambda(\trueThreshold) > 2(1+\lambda)\sqrt{\frac{\log(4/\varepsilon)}{2 (\NoNc \wedge \No)}} \right)
    \le \varepsilon.
\end{align*}
\end{proposition}

\begin{proof}[Proof of \Cref{prop:convergence.rate.weak}]
    By the triangle inequality and the Dvoretzky–Kiefer–Wolfowitz (DKW) inequality \cite{dvoretzky1956asymptotic, massart1990tight}, for any $\varepsilon' > 0$,
    \begin{align}\label{proof:eq:weak.1}
    \begin{split}
        &\quad~\PP\left(\sup_{\threshold \in \RR}\left|
        \EmpiricalBayesRiskLambda(\threshold)
        - \BayesRiskLambda(\threshold)\right| > (1+\lambda) \varepsilon' \right)\\
        &\le \PP\left(\sup_{\threshold \in \RR}\left|
        \cdfTestStatisticsNull(\threshold)
        - \EmpiricalCdfTestStatisticsNull(\threshold)\right| > \varepsilon' \right) +  \PP\left(\sup_{\threshold \in \RR}\left|
        \lambda\cdfTestStatistics{}(\threshold)
        - \lambda\EmpiricalCdfTestStatistics(\threshold)\right| > \lambda \varepsilon' \right)\\
        &\le 2 e^{-2 \NoNc\varepsilon'^2} + 2 e^{-2\No\varepsilon'^2}
        \le  4 e^{-2 (\NoNc \wedge \No)\varepsilon'^2}.
        \end{split}
    \end{align}
    Since $\EmpiricalThreshold$ minimizes $\EmpiricalBayesRiskLambda$,
    \begin{align}\label{proof:eq:weak.2}
    \begin{split}
        &\quad~\BayesRiskLambda(\EmpiricalThreshold) - \BayesRiskLambda(\trueThreshold) \\
        &= \left(\BayesRiskLambda(\EmpiricalThreshold) - \EmpiricalBayesRiskLambda(\EmpiricalThreshold)\right) +
        \left(\EmpiricalBayesRiskLambda(\EmpiricalThreshold) -  \EmpiricalBayesRiskLambda(\trueThreshold)\right) +  \left(\EmpiricalBayesRiskLambda(\trueThreshold) -  \BayesRiskLambda(\trueThreshold)\right)  \\
        &\le 2 \sup_{\threshold \in \RR}\left|
        \EmpiricalBayesRiskLambda(\threshold)
        - \BayesRiskLambda(\threshold)\right|.
    \end{split}
    \end{align}
    Therefore, we combine Eq.~\eqref{proof:eq:weak.1}, Eq.~\eqref{proof:eq:weak.2}, and take $\varepsilon' = \sqrt{\frac{\log(4/\varepsilon)}{2 (\NoNc \wedge \No)}}$,
    \begin{align*}
        &\quad~\PP\left(\BayesRiskLambda(\EmpiricalThreshold) - \BayesRiskLambda(\trueThreshold) > 2 (1+\lambda)\sqrt{\frac{\log(4/\varepsilon)}{2 \NoNc \wedge \No}} \right) \\
        &\le \PP\left(\sup_{\threshold \in \RR}\left|
        \EmpiricalBayesRiskLambda(\threshold)
        - \BayesRiskLambda(\threshold)\right| > (1+\lambda) \sqrt{\frac{\log(4/\varepsilon)}{2 (\NoNc \wedge \No)}}  \right)
        = \varepsilon.
    \end{align*}
\end{proof}

\begin{proposition}\label{prop:convergence.rate}
   Consider the setting in \Cref{sec:localFDR} and assume the following assumptions hold:
    \begin{enumerate}
        \item \label{prop:assu:density.derivative} $\pdfTestStatistics{}(\trueThreshold)$, $\pdfTestStatisticsNull(\trueThreshold) > 0$ and $\pdfTestStatistics{}'(\trueThreshold)$, $\pdfTestStatisticsNull'(\trueThreshold)$ exist;
        \item \label{prop:assu:density.upper.bound}
         $\pdfTestStatistics{}(\threshold) \le \pdfTestStatistics{\max}$, $\pdfTestStatisticsNull(\threshold) \le \pdfTestStatistics{0,\max}$, in $|\threshold - \trueThreshold| \le D$ for some $\pdfTestStatistics{\max}$, $\pdfTestStatistics{0,\max}$, $D > 0$;
        \item \label{prop:assu:MLR} the derivative $(\pdfTestStatisticsNull/\pdfTestStatistics{})'(\trueThreshold)$ is positive;
        \item \label{prop:assu:Bayes.risk} for any $\delta > 0$, there exists $\varepsilon_\delta > 0$ such that $\cdfTestStatisticsNull(\threshold) - \lambda \cdfTestStatistics{}(\threshold) > \cdfTestStatisticsNull(\trueThreshold) - \lambda \cdfTestStatistics{}(\trueThreshold) + \varepsilon_\delta$ for $|\threshold - \trueThreshold| > \delta$.
    \end{enumerate}
    Then there exists a constant $C^*>0$ such that with
    probability at least $1 - C^*/\log(\No \wedge \NoNc)$,
    \begin{align}\label{eq:prop:convergence.rate}
        \left|\EmpiricalThreshold - \trueThreshold\right|
        \le \frac{\log(\No \wedge \NoNc)}{(\No \wedge
      \NoNc)^{1/3}}\quad \text{for all}~\No, \NoNc \ge 5.
    \end{align}
\end{proposition}

\begin{lemma}\label{lemm:risk.valley}
Under Assumptions~\ref{prop:assu:density.derivative}, \ref{prop:assu:MLR}, there exists $C' > 0$, $D' > 0$ such that for any $\threshold$ satisfying $\left|\threshold - \trueThreshold \right| \le D'$,
\begin{align}\label{eq:lemm:risk.valley}
    \BayesRiskLambda(\trueThreshold) \le \BayesRiskLambda(\threshold) - C' (\threshold - \trueThreshold)^2.
\end{align}
\end{lemma}

\begin{proof}[Proof of \Cref{lemm:risk.valley}]

We perform Taylor expansion of $\BayesRiskLambda(\threshold)$ at $\trueThreshold$ with the Peano's form of remainder,    \begin{align}\label{proof:lemm:eq:taylor.expansion}
        \BayesRiskLambda(\threshold)
        &= \BayesRiskLambda(\trueThreshold) + \BayesRiskLambda'(\trueThreshold) (\threshold - \trueThreshold) + \frac{\BayesRiskLambda''(\trueThreshold)}{2}(\threshold - \trueThreshold)^2
        + o((\threshold - \trueThreshold)^2).
    \end{align}
    Since $\trueThreshold$ minimizes $\BayesRiskLambda(\threshold)$, by the KKT condition,    \begin{align}\label{proof:lemm:eq:first.order.derivative}
        \BayesRiskLambda'(\trueThreshold)
        &= \cdfTestStatisticsNull'(\trueThreshold) - \lambda \cdfTestStatistics{}'(\trueThreshold)
        = \pdfTestStatisticsNull(\trueThreshold) - \lambda \pdfTestStatistics{}(\trueThreshold)
        = 0.
    \end{align}
    Note that
    \begin{align}\label{proof:lemm:eq:MLR}
        0 < \left(\frac{\pdfTestStatisticsNull}{\pdfTestStatistics{}}\right)'(\trueThreshold)
        = \frac{\pdfTestStatisticsNull'{}(\trueThreshold) \pdfTestStatistics{}(\trueThreshold) - \pdfTestStatisticsNull{}(\trueThreshold) \pdfTestStatistics{}'(\trueThreshold)}{\pdfTestStatistics{}^2(\trueThreshold)}.
    \end{align}
    By Eq.~\eqref{proof:lemm:eq:first.order.derivative} and \eqref{proof:lemm:eq:MLR},
    \begin{align}\label{proof:lemm:eq:second.order.derivative}
    \begin{split}
    \BayesRiskLambda''(\trueThreshold)
    &= \pdfTestStatisticsNull'(\trueThreshold) - \lambda \pdfTestStatistics{}'(\trueThreshold)
    = \pdfTestStatisticsNull'(\trueThreshold)
    - \frac{\pdfTestStatisticsNull(\trueThreshold)}{\pdfTestStatistics{}(\trueThreshold)} \pdfTestStatistics{}'(\trueThreshold)
    = \frac{\pdfTestStatisticsNull'(\trueThreshold)
     \pdfTestStatistics{}(\trueThreshold) - \pdfTestStatisticsNull(\trueThreshold) \pdfTestStatistics{}'(\trueThreshold)}{\pdfTestStatistics{}(\trueThreshold)}
     >0.
     \end{split}
    \end{align}
    Plug Eq.~\eqref{proof:lemm:eq:first.order.derivative} and \eqref{proof:lemm:eq:second.order.derivative} into Eq.~\eqref{proof:lemm:eq:taylor.expansion} and we have proved that there exists $D' > 0$, $C' = \BayesRiskLambda''(\trueThreshold)/4 > 0$, such that Eq.~\eqref{eq:lemm:risk.valley} is valid.
\end{proof}

\begin{lemma}\label{lemm:consistency}
Under the assumptions in \Cref{prop:convergence.rate},
$\EmpiricalThreshold \to \trueThreshold$ in probability as $\No$, $\NoNc \to \infty$.
\end{lemma}

\begin{proof}[Proof of \Cref{lemm:consistency}]
    Let $\EmpiricalBayesRiskLambda(\threshold):= \EmpiricalCdfTestStatisticsNull(\threshold) - \lambda \EmpiricalCdfTestStatistics(\threshold)$.
    For any $\delta > 0$, there exists $\varepsilon > 0$ such that $\BayesRiskLambda(\trueThreshold) < \BayesRiskLambda(\threshold) - \varepsilon$ for any $|\threshold - \trueThreshold| > \delta$.
    By Eq.~\eqref{proof:eq:weak.1} with $\varepsilon' = \varepsilon/2(1+\lambda)$, we have with probability at least $1 - 8 e^{-2 (\NoNc \wedge \No)\varepsilon'^2}$, for any $|\threshold - \trueThreshold| > \delta$,
    \begin{align*}
    &\quad~\EmpiricalBayesRiskLambda(\trueThreshold)  - \EmpiricalBayesRiskLambda(\threshold)\\
          &= (\EmpiricalBayesRiskLambda(\trueThreshold)  - \BayesRiskLambda(\trueThreshold))
          + (\BayesRiskLambda(\trueThreshold) - \BayesRiskLambda(\threshold))
          + (\BayesRiskLambda(\threshold) - \EmpiricalBayesRiskLambda(\threshold))\\
          &< -\varepsilon + 2 (1+\lambda)\varepsilon'
          = 0.
    \end{align*}
   Therefore, $\PP(|\EmpiricalThreshold - \trueThreshold| \le \delta) \to 1$.
\end{proof}

\begin{lemma}\label{lemm:maximal.inequality}
Under the assumptions in \Cref{prop:convergence.rate},
there exists $C > 0$ such that for any $0 < \delta \le D$, $\No \wedge \NoNc \ge 5$,
    \begin{align}\label{lemm:eq:M.estimator}
        &\quad~\EE\left[\sup_{|\threshold - \trueThreshold| \le \delta} \left|
        \left(\EmpiricalBayesRiskLambda(\threshold) - \BayesRiskLambda{}(\threshold)\right) -
        \left(\EmpiricalBayesRiskLambda(\trueThreshold) - \BayesRiskLambda{}(\trueThreshold)\right)
        \right|\right] \\
        &\le C (1+\lambda) \left(\sqrt{\frac{\delta \left(\pdfTestStatistics{\max} \vee \pdfTestStatistics{0,\max}\right) \log(\No \wedge \NoNc)}{\No \wedge \NoNc}} + \frac{\log(\No \wedge \NoNc)}{\No \wedge \NoNc}\right).
    \end{align}
\end{lemma}

\begin{proof}[Proof of \Cref{lemm:maximal.inequality}]
By the triangle inequality,
\begin{align*}
    &\quad~\sup_{|\threshold - \trueThreshold| \le \delta}\left|
    \left(\EmpiricalBayesRiskLambda(\threshold) - \BayesRiskLambda{}(\threshold)\right) -
    \left(\EmpiricalBayesRiskLambda(\trueThreshold) - \BayesRiskLambda{}(\trueThreshold)\right)
    \right|\\
    &\le \sup_{|\threshold - \trueThreshold| \le \delta}\left|
    \left(\EmpiricalCdfTestStatisticsNull(\threshold) - \cdfTestStatisticsNull(\threshold)\right) -
    \left(\EmpiricalCdfTestStatisticsNull(\trueThreshold) - \cdfTestStatisticsNull(\trueThreshold)\right)
    \right| \\
    &\quad~+ \sup_{|\threshold - \trueThreshold| \le \delta}\lambda\left|
    \left(\EmpiricalCdfTestStatistics{}(\threshold) - \cdfTestStatistics{}(\threshold)\right) -
    \left(\EmpiricalCdfTestStatistics(\trueThreshold) - \cdfTestStatistics{}{}(\trueThreshold)\right)
    \right|.
\end{align*}
We show there exists $C > 0$, for any $0 < \delta \le D$, $\No \ge 5$,
\begin{align}\label{proof:lemm:eq:M.estimator}
    \EE\left[\sup_{|\threshold - \trueThreshold| \le \delta} \left|
        \left(\EmpiricalCdfTestStatistics(\threshold) - \cdfTestStatistics{}(\threshold)\right) -
        \left(\EmpiricalCdfTestStatistics(\trueThreshold) - \cdfTestStatistics{}(\trueThreshold)\right)
        \right|\right] \le C \left(\sqrt{\frac{\delta \pdfTestStatistics{\max} \log(\No)}{\No}} + \frac{\log(\No)}{\No}\right).
\end{align}
The analysis also applies to the null CDF part and combining the two parts yields the desired result.

For $\threshold < \trueThreshold$,
    \begin{align*}
        \left(\EmpiricalCdfTestStatistics(\threshold) - \cdfTestStatistics{}(\threshold)\right) -
        \left(\EmpiricalCdfTestStatistics(\trueThreshold) - \cdfTestStatistics{}(\trueThreshold)\right)
        = \frac{1}{n} \sum_{i=1}^\No \1_{\{\testStatistics{i} \in (\threshold, \trueThreshold]\}} - \PP\left(\testStatistics{i} \in (\threshold, \trueThreshold]\right),
    \end{align*}
    and similarly for $\threshold > \trueThreshold$. To bound the left hand side of Eq.~\eqref{proof:lemm:eq:M.estimator}, we define   \begin{align*}
        &\sigma_\delta
        := \sup_{|\threshold - \trueThreshold| \le \delta} \sqrt{\var\left(\1_{\{\testStatistics{i} \in [\threshold \wedge \trueThreshold, \threshold \vee \trueThreshold]\}}\right)},\\
        N_\delta :=& \left\{\left(i, \1_{\{\testStatistics{i} \in [\threshold \wedge \trueThreshold, \threshold \vee \trueThreshold]\}}\right), ~|\threshold - \trueThreshold| < \delta \right\}, \quad
        \pi_\delta
        := \EE \left[\log(2 |N_\delta|)\right].
    \end{align*}
    We bound $\sigma_\delta$ and $ \pi_\delta $ separately. For $\sigma_\delta$, by Assumption~\ref{prop:assu:density.upper.bound},
    \begin{align}\label{proof:lemm:eq:bound.sigma}
        \sigma_\delta
        \le \sqrt{\PP\left(\testStatistics{i} \in [\trueThreshold - \delta, \trueThreshold + \delta]\right)}
        \le \sqrt{2 \delta \pdfTestStatistics{\max}}.
    \end{align}
    For $\pi_\delta$, notice that $\1_{\{\testStatistics{i} \in [\trueThreshold, \threshold]\}}$ is binary and increases with regard to $\threshold$, then
    \begin{align*}
         \left|\left\{\left(i, \1_{\{\testStatistics{i} \in [\trueThreshold, \threshold]\}}\right), ~0 \le \threshold - \trueThreshold \le \delta \right\}\right| \le \No+1,
    \end{align*}
    and similarly for $-\delta \le \threshold - \trueThreshold \le 0$. Combine the two parts, $ \left|N_\delta\right| \le 2 \No+2$,
    and $\pi_\delta \le \log(4\No+4)$.
    By \cite[Lemma 6.4]{massart2007concentration} or \cite[Theorem 3.1]{baraud2016bounding},
    \begin{align*}
        &\quad~\EE\left[\sup_{|\threshold - \trueThreshold| \le \delta} \left|
        \frac{1}{n} \sum_{i=1}^\No \1_{\{\testStatistics{i} \in [\threshold \wedge \trueThreshold, \threshold \vee \trueThreshold]\}} - \PP\left(\testStatistics{i} \in [\threshold \wedge \trueThreshold, \threshold \vee \trueThreshold]\right)
        \right|\right]\\
        &\le 2 \sigma_\delta \sqrt{\frac{2 \pi_\delta}{\No}} + \frac{8\pi_\delta}{\No}
        \le 4 \sqrt{\frac{\delta \pdfTestStatistics{\max} \log(4\No+4)}{\No}} + \frac{8\log(4\No+4)}{\No}.
    \end{align*}
    Take $C = 16$, Eq.~\eqref{proof:lemm:eq:M.estimator} is valid for $n \ge 5$.
\end{proof}

\begin{proof}[Proof of \Cref{prop:convergence.rate}]
By the proof of \Cref{lemm:consistency}, for $\No \wedge \NoNc$ large enough,  $\PP(|\EmpiricalThreshold - \trueThreshold| > D'/2) \le 1/\log(\No \wedge \NoNc)$ for the $D'$ in \Cref{lemm:risk.valley}.
Define the rate
$r_{\No \wedge\NoNc} = ({\No \wedge \NoNc}/\log({\No \wedge \NoNc}))^{-1/3}$, and shells
\begin{align*}
 S_{{\No \wedge\NoNc},j} = \left\{\tau: \frac{2^{j-1}}{r_{\No \wedge\NoNc}} \le  |\threshold - \trueThreshold| \le \frac{2^j}{r_{\No \wedge\NoNc}} \right\}, \quad j \ge 1.
\end{align*}
If the event $\{|\EmpiricalThreshold- \trueThreshold| \ge 2^t/r_{\No \wedge\NoNc}\}$ happens for some $t \ge 0$, then there exists $j \ge t$ such that $\EmpiricalThreshold \in S_{{\No \wedge\NoNc}, j}$. Therefore, for $\No \wedge \NoNc$ large enough,
\begin{align*}
    &\quad~\PP\left(
    |\EmpiricalThreshold- \trueThreshold| \ge \frac{2^t}{r_{\No \wedge\NoNc}}
    \right)\\
    &\le \PP\left(
    |\EmpiricalThreshold - \trueThreshold| \ge \frac{2^t}{r_{\No \wedge\NoNc}},  ~|\EmpiricalThreshold - \trueThreshold|\le D'/2
    \right) + \PP\left(
    |\EmpiricalThreshold- \trueThreshold| > D'/2
    \right) \\
    &\le \sum_{j \ge t, ~\frac{2^{j-1}}{r_{\No \wedge\NoNc}} \le D'/2} \PP\left(
    \EmpiricalThreshold \in S_{\No \wedge\NoNc,j}
    \right)  + \frac{1}{\log(\No \wedge\NoNc)} \\
    &\le \sum_{j \ge t, ~\frac{2^{j-1}}{r_{\No \wedge\NoNc}} \le D'/2} \PP\left(
    \exists~\threshold \in  S_{\No \wedge\NoNc,j},~ \EmpiricalBayesRisk(\threshold) < \EmpiricalBayesRisk(\trueThreshold)
    \right)  + \frac{1}{\log(\No \wedge\NoNc)}.
\end{align*}
Notice that if $\EmpiricalBayesRisk(\threshold) < \EmpiricalBayesRisk(\trueThreshold)$
and $|\threshold - \trueThreshold| \le D'$, by \Cref{lemm:risk.valley},
\begin{align*}
    \EmpiricalBayesRisk(\threshold) - \BayesRisk(\threshold) - (\EmpiricalBayesRisk(\trueThreshold) - \BayesRisk(\trueThreshold))
    &\le \BayesRisk(\trueThreshold) - \BayesRisk(\threshold)
    \le -C'(\threshold - \trueThreshold)^2
    \le  -C'\frac{4^{j-1}}{r_{\No \wedge\NoNc}^2}.
\end{align*}
Then by Markov inequality and \Cref{lemm:maximal.inequality},
\begin{align*}
    &~\quad\PP\left(
    \exists~\threshold \in  S_{\No \wedge\NoNc,j},~ \EmpiricalBayesRisk(\threshold) < \EmpiricalBayesRisk(\trueThreshold)
    \right)\\
    &\le \PP\left(\sup_{\threshold \in S_{\No \wedge\NoNc,j}}\left|\EmpiricalBayesRisk(\threshold) - \BayesRisk(\threshold) - (\EmpiricalBayesRisk(\trueThreshold) - \BayesRisk(\trueThreshold))\right| \ge C'\frac{4^{j-1}}{r_{\No \wedge\NoNc}^2} \right)\\
    &\le \frac{\EE\left[\sup_{\threshold \in S_{\No \wedge\NoNc,j}}\left|\EmpiricalBayesRisk(\threshold) - \BayesRisk(\threshold) - (\EmpiricalBayesRisk(\trueThreshold) - \BayesRisk(\trueThreshold))\right| \right]}{C'\frac{4^{j-1}}{r_{\No \wedge\NoNc}^2}}\\
    &\le \frac{C (1+\lambda) \left(\sqrt{\frac{(2^j/r_{\No \wedge \NoNc}) (\pdfTestStatistics{\max} \vee \pdfTestStatistics{0,\max})\log(\No \wedge \NoNc)}{\No \wedge \NoNc}} + \frac{\log(\No \wedge \NoNc)}{\No \wedge \NoNc}\right)}{C'\frac{4^{j-1}}{r_{\No \wedge\NoNc}^2}}\\
    &\le \frac{8C(1+\lambda) \sqrt{\pdfTestStatistics{\max} \vee \pdfTestStatistics{0,\max} \vee 1})}{C'} \cdot \sqrt{\frac{\log(\No \wedge \NoNc)}{\No \wedge \NoNc}} \cdot r_{\No \wedge\NoNc}^{3/2} \cdot  2^{-3j/2} \\
    &= \frac{8C(1+\lambda) \sqrt{\pdfTestStatistics{\max} \vee \pdfTestStatistics{0,\max} \vee 1})}{C'} \cdot 2^{-3j/2},
\end{align*}
where we use $\sqrt{\frac{(2^j/r_{\No \wedge\NoNc})\log(\No \wedge \NoNc)}{\No \wedge \NoNc}} \ge \frac{\log(\No \wedge \NoNc)}{\No \wedge \NoNc}$.
Finally, we sum over shells $S_{\No \wedge \NoNc, j}$,
\begin{align*}
    \PP\left(|\EmpiricalThreshold- \trueThreshold| \ge \frac{2^t}{r_{\No \wedge\NoNc}}\right)
    &\le \frac{8C(1+\lambda) \sqrt{\pdfTestStatistics{\max} \vee \pdfTestStatistics{0,\max} \vee 1})}{C'}  \sum_{j \ge t}^\infty  2^{-3j/2} + \frac{1}{\log(\No \wedge \NoNc)}\\
    &\le \frac{8C(1+\lambda) \sqrt{\pdfTestStatistics{\max} \vee \pdfTestStatistics{0,\max} \vee 1})}{C'} \cdot \frac{2^{-3t/2}}{1-2^{-3/2}} + \frac{1}{\log(\No \wedge \NoNc)}.
\end{align*}
Take $2^{3t/2} = \log(\No \wedge \NoNc)$,
\begin{align*}
    \PP\left(|\EmpiricalThreshold- \trueThreshold| \ge \frac{2^t}{r_{\No \wedge\NoNc}}\right)
    = \PP\left(|\EmpiricalThreshold- \trueThreshold| \ge \frac{\log(\No \wedge\NoNc)}{(\No \wedge\NoNc)^{1/3}}\right)
    \le \frac{C^*}{\log(\No \wedge\NoNc)},
\end{align*}
where we let $C^* =  \frac{8C(1+\lambda) \sqrt{\pdfTestStatistics{\max} \vee \pdfTestStatistics{0,\max} \vee 1})}{(1-2^{-3/2})C'} + 1$.
\end{proof}

\begin{proposition}\label{prop:convergence.rate.uniform}
   Let $\Lambda = (a, b) \subseteq (0,\probNull)$ and $\calT =
   \{\trueThreshold: \lambda \in \Lambda\}$. Suppose there exists
   a neighbourhood $\calT_\Delta = (\inf {\calT} - \Delta, \sup
   {\calT} + \Delta)$ for some $\Delta > 0$ and some $0 <
   \pdfTestStatistics{}_{\min} < \pdfTestStatistics{}_{\max} < \infty$
   such that
   \begin{enumerate}
        \item \label{prop:assu:Bayes.risk.uniform}  $\inf_{\lambda \in
            \Lambda}\inf_{\threshold \notin
            \calT_\Delta}(\cdfTestStatisticsNull(\threshold) - \lambda
          \cdfTestStatistics{}(\threshold)) -
          (\cdfTestStatisticsNull(\trueThreshold) - \lambda
          \cdfTestStatistics{}(\trueThreshold)) > 0$;
        \item \label{prop:assu:density.upper.bound.uniform}
          $0 < \pdfTestStatistics{\min} \leq \pdfTestStatistics{}(t),\pdfTestStatistics{0}(t) \leq  \pdfTestStatistics{\max} < \infty$ for
          all $\threshold \in \calT_{\Delta}$;
        \item \label{prop:assu:MLR.uniform}
           $\pdfTestStatisticsNull'$, $f'$, $(\pdfTestStatisticsNull/\pdfTestStatistics{})'$ exist and $0 < f_{\min}'
     \le \pdfTestStatisticsNull'(\threshold), \pdfTestStatistics{}'(\threshold), (\pdfTestStatisticsNull/\pdfTestStatistics{})'(\threshold) \le f_{\max}' < \infty$  for
          all $\threshold \in \calT_{\Delta}$.
    \end{enumerate}
    Then there exists a constant $C^* > 0$ such that with probability at least $1 - C^*/\log(\No \wedge \NoNc)$,
    \begin{align}\label{eq:prop:convergence.rate}
        \sup_{\lambda \in \Lambda}\left|\EmpiricalThreshold - \trueThreshold\right|
        \le \frac{\log(\No \wedge \NoNc)}{(\No \wedge \NoNc)^{1/3}}, \quad \text{for all} \No \wedge \NoNc \ge 5.
    \end{align}
\end{proposition}

\begin{proof}[Proof of \Cref{prop:convergence.rate.uniform}]
We can follow the proof of \Cref{prop:convergence.rate} and substitute \Cref{lemm:risk.valley}, \Cref{lemm:consistency}, and \Cref{lemm:maximal.inequality} with their corresponding uniform results provided below.

Analogous to \Cref{lemm:risk.valley}, we have under Assumption~\ref{prop:assu:density.upper.bound.uniform} and Assumption~\ref{prop:assu:MLR.uniform}, there exists $C' > 0$, $D' > 0$ such that for any $\lambda \in \Lambda$, any $\threshold$ satisfying $\left|\threshold - \trueThreshold \right| \le D'$,
\begin{align}\label{eq:lemm:risk.valley.uniform}
    \BayesRiskLambda(\trueThreshold) \le \BayesRiskLambda(\threshold) - C' (\threshold - \trueThreshold)^2.
\end{align}
In fact, by Assumption~\ref{prop:assu:MLR.uniform}, we perform Taylor expansion of $\BayesRiskLambda(\threshold)$ at $\trueThreshold$. For $|\threshold - \trueThreshold| <  \Delta \wedge (f_{\min}' f_{\min}/2(f_{\max}')^2)$, \begin{align}\label{proof:lemm:eq:taylor.expansion.uniform}
        \BayesRiskLambda(\threshold)
        &= \BayesRiskLambda(\trueThreshold) + \BayesRiskLambda'(\trueThreshold) (\threshold - \trueThreshold) + \frac{\BayesRiskLambda''(\threshold')}{2}(\threshold - \trueThreshold)^2,
    \end{align}
    for some $\threshold' \in [t \wedge \trueThreshold, t \vee \trueThreshold]$.
    Since $\trueThreshold$ minimizes $\BayesRiskLambda(\threshold)$, by the KKT condition,    \begin{align}\label{proof:lemm:eq:first.order.derivative.uniform}
        \BayesRiskLambda'(\trueThreshold)
        &= \cdfTestStatisticsNull'(\trueThreshold) - \lambda \cdfTestStatistics{}'(\trueThreshold)
        = \pdfTestStatisticsNull(\trueThreshold) - \lambda \pdfTestStatistics{}(\trueThreshold)
        = 0.
    \end{align}
    By Assumption~\ref{prop:assu:MLR.uniform},
    \begin{align}\label{proof:lemm:eq:MLR.uniform}
        f_{\min}' < \left(\frac{\pdfTestStatisticsNull}{\pdfTestStatistics{}}\right)'(\threshold')
        = \frac{\pdfTestStatisticsNull'{}(\threshold') \pdfTestStatistics{}(\threshold') - \pdfTestStatisticsNull{}(\threshold') \pdfTestStatistics{}'(\threshold')}{\pdfTestStatistics{}^2(\threshold')}.
    \end{align}
    By Eq.~\eqref{proof:lemm:eq:first.order.derivative.uniform}, \eqref{proof:lemm:eq:MLR.uniform}, and Assumption~\ref{prop:assu:density.upper.bound.uniform}, Assumption~\ref{prop:assu:MLR.uniform},
    \begin{align}\label{proof:lemm:eq:second.order.derivative.uniform}
    \begin{split}
    \BayesRiskLambda''(\threshold')
    &= \pdfTestStatisticsNull'(\threshold') - \lambda \pdfTestStatistics{}'(\threshold')
    = \pdfTestStatisticsNull'(\threshold')
    - \frac{\pdfTestStatisticsNull(\trueThreshold)}{\pdfTestStatistics{}(\trueThreshold)} \pdfTestStatistics{}'(\threshold') \\
    &= \pdfTestStatisticsNull'(\threshold') - \frac{\pdfTestStatisticsNull(\threshold')}{\pdfTestStatistics{}(\threshold')} \pdfTestStatistics{}'(\threshold') + \frac{\pdfTestStatisticsNull(\threshold')}{\pdfTestStatistics{}(\threshold')} \pdfTestStatistics{}'(\threshold')
    - \frac{\pdfTestStatisticsNull(\trueThreshold)}{\pdfTestStatistics{}(\trueThreshold)} \pdfTestStatistics{}'(\threshold') \\
    &> \frac{\pdfTestStatisticsNull'(\threshold')
     \pdfTestStatistics{}(\threshold') - \pdfTestStatisticsNull(\threshold') \pdfTestStatistics{}'(\threshold')}{\pdfTestStatistics{}(\threshold')} - |\threshold' - \trueThreshold|\sup_{\threshold \in \calT_\Delta} \left(\frac{\pdfTestStatisticsNull}{\pdfTestStatistics{}}\right)'(\threshold) \sup_{\threshold \in \calT_\Delta} \pdfTestStatistics{}'(\threshold)\\
     &> f_{\min}' f_{\min} - |\threshold' - \trueThreshold| (f_{\max}')^2.
     \end{split}
    \end{align}
    Plug Eq.~\eqref{proof:lemm:eq:first.order.derivative.uniform} and \eqref{proof:lemm:eq:second.order.derivative.uniform} into Eq.~\eqref{proof:lemm:eq:taylor.expansion.uniform} and we have proved that Eq.~\eqref{eq:lemm:risk.valley.uniform} is valid for $C' =  f_{\min}' f_{\min}/4 > 0$ and $D' < \Delta \wedge (f_{\min}' f_{\min}/2(f_{\max}')^2)$.

Analogous to \Cref{lemm:consistency}, under the assumptions in \Cref{prop:convergence.rate.uniform}, $\sup_{\lambda \in \Lambda} |\EmpiricalThreshold - \trueThreshold| \to 0$ in probability as $\No$, $\NoNc \to \infty$.
    In fact, by Assumption~\ref{prop:assu:MLR.uniform}, $\BayesRiskLambda(\threshold)$ decreases on $(\inf \calT - \Delta, \trueThreshold]$ and increases on $[\trueThreshold, \sup \calT + \Delta)$.
    By Assumption~\ref{prop:assu:Bayes.risk.uniform} and \eqref{eq:lemm:risk.valley.uniform}, for any $0 < \delta < D'$, $\lambda \in \Lambda$, $\threshold - \trueThreshold > \delta$,
    \begin{align*}
        \BayesRiskLambda(\trueThreshold) - \BayesRiskLambda(\threshold)
        \le \BayesRiskLambda(\trueThreshold) - \BayesRiskLambda(\trueThreshold + \delta)
        \le -C'\delta^2.
    \end{align*}
    Similarly for $\threshold - \trueThreshold < -\delta$. The rest of the proof is similar to that of \Cref{lemm:consistency}.

Analogous to \Cref{lemm:maximal.inequality}, we have under Assumption~\ref{prop:assu:density.upper.bound.uniform},
there exists $C > 0$, $\delta > 0$, such that for $\No \wedge \NoNc \ge 5$,
    \begin{align*}
        \begin{split}
        &\quad~\EE\left[\sup_{\lambda \in \Lambda, |\threshold - \trueThreshold| \le \delta} \left|
        \left(\EmpiricalBayesRiskLambda(\threshold) - \BayesRiskLambda{}(\threshold)\right) -
        \left(\EmpiricalBayesRiskLambda(\trueThreshold) - \BayesRiskLambda{}(\trueThreshold)\right)
        \right|\right] \\
        &\le C (1+\lambda) \left(\sqrt{\frac{\delta \pdfTestStatistics{\max} \log(\No \wedge \NoNc)}{\No \wedge \NoNc}} + \frac{\log(\No \wedge \NoNc)}{\No \wedge \NoNc}\right).
        \end{split}
    \end{align*}
In fact, the outline of the proof is similar to that of \Cref{lemm:maximal.inequality}.
    It is left to characterize the complexity of the function class $\left\{\1_{\{\testStatistics{i} \in [\threshold \wedge \trueThreshold, \threshold \vee \trueThreshold]\}}, \lambda \in \Lambda, |t - \trueThreshold| < \delta\right\}$. We define
        \begin{align*}
        &\sigma_\delta
        := \sup_{\lambda \in \Lambda, |\threshold - \trueThreshold| \le \delta} \sqrt{\var\left(\1_{\{\testStatistics{i} \in [\threshold \wedge \trueThreshold, \threshold \vee \trueThreshold]\}}\right)},\\
        N_\delta :=& \left\{\left\{i: \1_{\{\testStatistics{i} \in [\threshold \wedge \trueThreshold, \threshold \vee \trueThreshold]\}}\right\}: ~\lambda \in \Lambda, ~|\threshold - \trueThreshold| < \delta \right\}, \quad
        \pi_\delta
        := \EE \left[\log(2 |N_\delta|)\right].
    \end{align*}
    We bound $\sigma_\delta$ and $ \pi_\delta $ separately. For $\sigma_\delta$, by Assumption~\ref{prop:assu:density.upper.bound.uniform},
    \begin{align*}
        \sigma_\delta
        \le \sqrt{\PP\left(\testStatistics{i} \in [\trueThreshold - \delta, \trueThreshold + \delta]\right)}
        \le \sqrt{2 \delta \pdfTestStatistics{\max}}.
    \end{align*}
    For $\pi_\delta$, notice that $\{[\threshold \wedge \trueThreshold, \threshold \vee \trueThreshold]: \lambda \in \Lambda, |\threshold - \trueThreshold| < \delta\}$ is a subset of all intervals of $\RR$, then
    \begin{align*}
         \left|\left\{\left\{i: \1_{\{\testStatistics{i} \in [\threshold \wedge \trueThreshold, \threshold \vee \trueThreshold]\}}\right\}: ~\lambda \in \Lambda, ~|\threshold - \trueThreshold| < \delta \right\}\right|
         \le 1 + \binom{\No+1}{2}
         \le \No^2,
    \end{align*}
    for $\No \ge 5$, and immediately $\pi_\delta \le \log(2\No^2)$. 
\end{proof}

\begin{remark}\label{rmk:non.monotone}
    In the cases with non-monotone likelihood ratios, we identify the test
    statistic values $t$ that minimize the objective function
    $\EmpiricalCdfTestStatisticsNull(\threshold) - \lambda
\EmpiricalCdfTestStatistics(\threshold)$ within a neighborhood $[t -h_n, t+h_n]$, where $h_n$ denotes a vanishing bandwidth. If there are a finite number of test statistic values $\tau_{\lambda,l}^*$ such that $\localFDR(\tau_{\lambda,l}^*) = q$, an analogous convergence result to \Cref{prop:convergence.rate} can be obtained.
\end{remark}

\section{Counter-examples}\label{sec:example}

\subsection{Stochastic dominance and \PRDS}\label{sec:example.PRDS}

We provide an example to show that condition~\ref{PRDS:assu:null.nc.identical} in
\Cref{prop:PRDS} cannot be relaxed to the
stochastic dominance
condition~\ref{vali:assu:null.nc.conservative.general} in
\Cref{prop:pvalue.validity}.
Consider two true null test statistics $\testStatistics{1}$, $\testStatistics{2} \sim U[0,1]$, two internal negative control test statistics $\testStatistics{3}$, $\testStatistics{4} \sim \text{Beta}(1,2)$, and assume all variables are independent.
Then $\testStatistics{1}$, $\testStatistics{2}$ uniformly stochastically dominate $\testStatistics{3}$, $\testStatistics{4}$, but
$(\pval{i})_{1 \le i \le 2}$ is not PRD because
$\PP\left(\pval{2} = 1 \mid \pval{1} = \frac{1}{3}\right)
= \frac{4}{9}
> \frac{5}{12}
= \PP\left(\pval{2} = 1 \mid \pval{1} = \frac{2}{3}\right)$.

\subsection{Fisher's method and permutation test}\label{sec:example.fisher}

We show directly combining Fisher's method with \nickname~p-values fails to control the type I error.
In \Cref{fig:fisher.test} panel (a), the theoretical null distribution $\chi_{2 \No}^2$ assuming independent p-values is inappropriate.
In contrast, the simulated null based on $1000$ permutations is
reasonably accurate (\Cref{fig:fisher.test} panel (b)).

\begin{figure}[tbp]
    \centering
    \begin{minipage}{8cm}
    \centering
\includegraphics[clip, trim = 0cm 0cm 0cm 0cm, width  = 7cm]{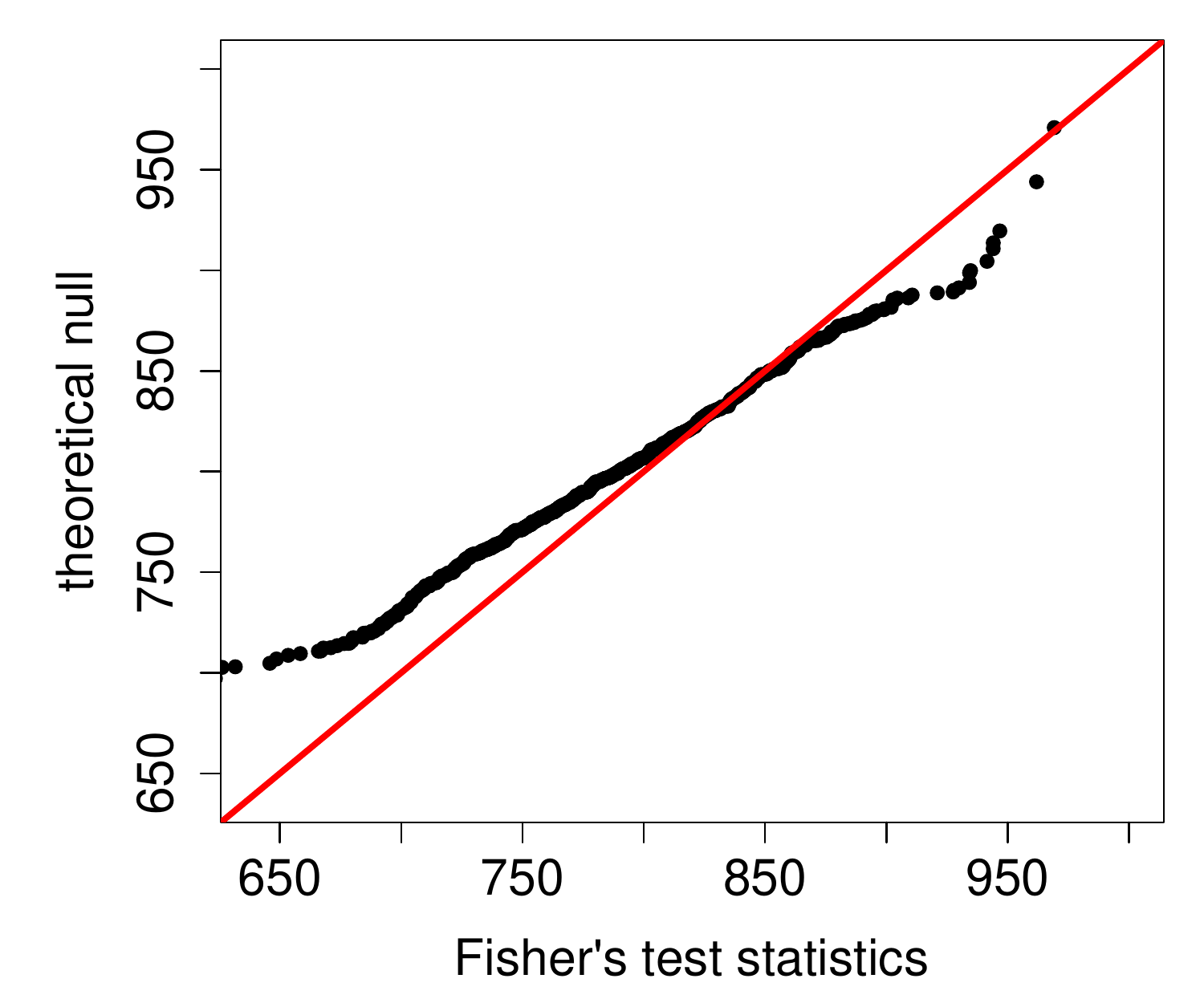}
    \subcaption{theoretical null}
    \end{minipage}
    \begin{minipage}{8cm}
    \centering
\includegraphics[clip, trim = 0cm 0cm 0cm 0cm, width  = 7cm]{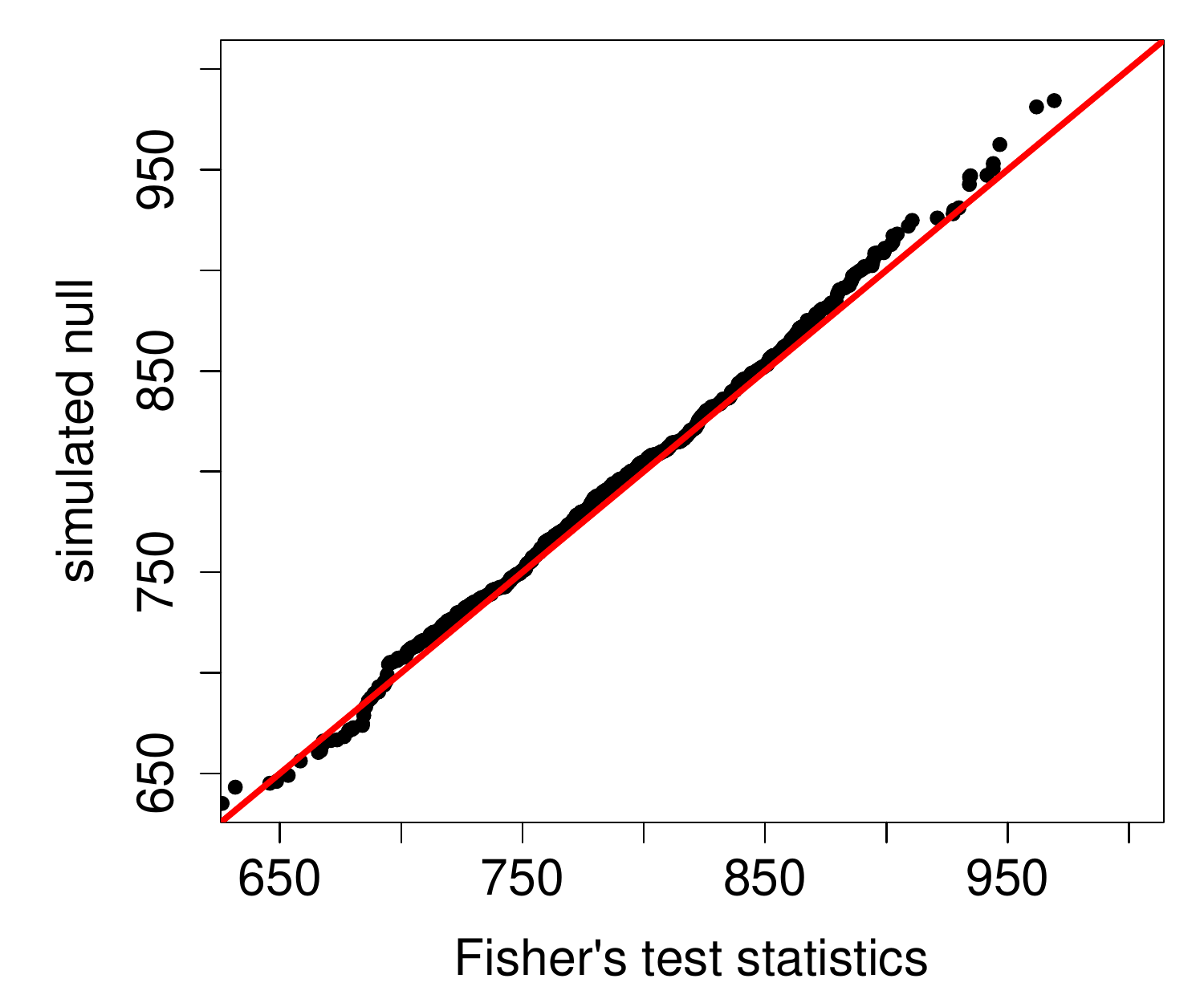}
    \subcaption{simulated null}
    \end{minipage}
    \caption{{Combining Fisher's method with \nickname~p-values.
    We generate $400$ test statistics for investigation and $400$ internal negative control test statistics independently from $\mathcal{N}(0,1)$, compute the \nickname~p-values and further the Fisher's test statistic.
    Panel (a) displays the quantile-quantile plot based on the theoretical null $\chi^2_{2\No}$ and panel (b) depicts the simulated null ($10^3$ permutations).
}}
    \label{fig:fisher.test}
\end{figure}

\subsection{\textcite[thm.\
  2]{bates21_testin_outlier_with_confor_p_values}}
\label{sec:bates-typo}

Theorem 2 in \cite{bates21_testin_outlier_with_confor_p_values} states that if $(\testStatistics{i})_{i \in \nullHypothesisIndex}$ is jointly independent and $(\testStatistics{i})_{i \in \nullHypothesisIndex} \independent (\testStatistics{j})_{j \in \hypothesisIndex{\text{nc}}}$, then $(\pval{i})_{i \in \hypothesisIndex{}}$ is \PRDS~on $(\pval{i})_{i \in \nullHypothesisIndex}$.
Surprisingly, this theorem does not impose any assumptions on
$(\testStatistics{i})_{i \in \hypothesisIndex{1}}$ (what
\textcite{bates21_testin_outlier_with_confor_p_values} call outliers). The
next simple counter-example shows that some assumptions on
$(\testStatistics{i})_{i \in \hypothesisIndex{1}}$ are indeed
necessary.   
In this example, $\nullHypothesisIndex = \{1\}$, $ \hypothesisIndex{1} =
\{2\}$, $\hypothesisIndex{\text{nc}} = \{3\}$,
$\testStatistics{1}$, $\testStatistics{3}$ are i.i.d.\ $U(0,1)$,
and $\testStatistics{2} = \1_{\{\testStatistics{1} < 0.5\}}$. For the
increasing set $\calD = \{\pvalAll: \pval{2} > 0.5\}$, we have
\begin{align*}
    \PP(\pvalAll \in \calD \mid \pval{1} = 0.5)
    &= \PP(\testStatistics{1} < 0.5 \mid \testStatistics{1} <
      \testStatistics{3}) = 2 \PP(\testStatistics{1} < 0.5, \testStatistics{1} < \testStatistics{3})
    = \frac{3}{4}, \\
    \PP(\pvalAll \in \calD \mid \pval{1} = 1)
    &= \PP(\testStatistics{1} < 0.5 \mid \testStatistics{1} \ge
      \testStatistics{3}) = 2 \PP(\testStatistics{1} < 0.5, \testStatistics{1} \ge \testStatistics{3})
    = \frac{1}{4}.
\end{align*}
Thus, the desired \PRDS~property is not true in this case.

A closer examination of Theorem 2 in
\cite{bates21_testin_outlier_with_confor_p_values} shows that a key
step in their proof (reproduced below using our notation) directly drops the
conditioning event $p_i = p$ without any justification:
\begin{align*}
    \PP\left(\pvalAll \in \calD \mid \pval{i} = p \right)
    = \EE_{\boldsymbol{Z}| \pval{i} = p}\left[ \PP\left(\pvalAll \in \calD \mid \boldsymbol{Z} \right)\right],
\end{align*}
where $\boldsymbol{Z}$ is the set of order statistics of
$(\testStatistics{j})_{j \in\hypothesisIndex{ \text{nc}}}$. This step
can be rectified by assuming that $(\testStatistics{i})_{i \in
  \hypothesisIndex{1}}$ is mutually independent and is independent of
the remaining test statistics.

\section{Additional numerical results}\label{sec:additional.simulation}

We report in \Cref{fig:nc.sample.size} the power of
\BH~\nickname~for $\No \in \{100, 200\}$, $\NoNonNull
\in \{10, 20\}$, $q \in \{0.1, 0.2\}$, and $\NoNc \in \{50, 100, 150,
\dotsc, 500\}$. We adopt the setting of independent p-values with
exact baseline p-values in \Cref{sec:simulation}. We do not report
the power of \BH, as it is numerically identical to \BH~oracle in this
case.

\begin{figure}[tbp]
    \centering
    \begin{minipage}{6.7cm}
    \centering
\includegraphics[clip, trim = 0cm 0cm 0cm 1cm, width = 5.6cm]{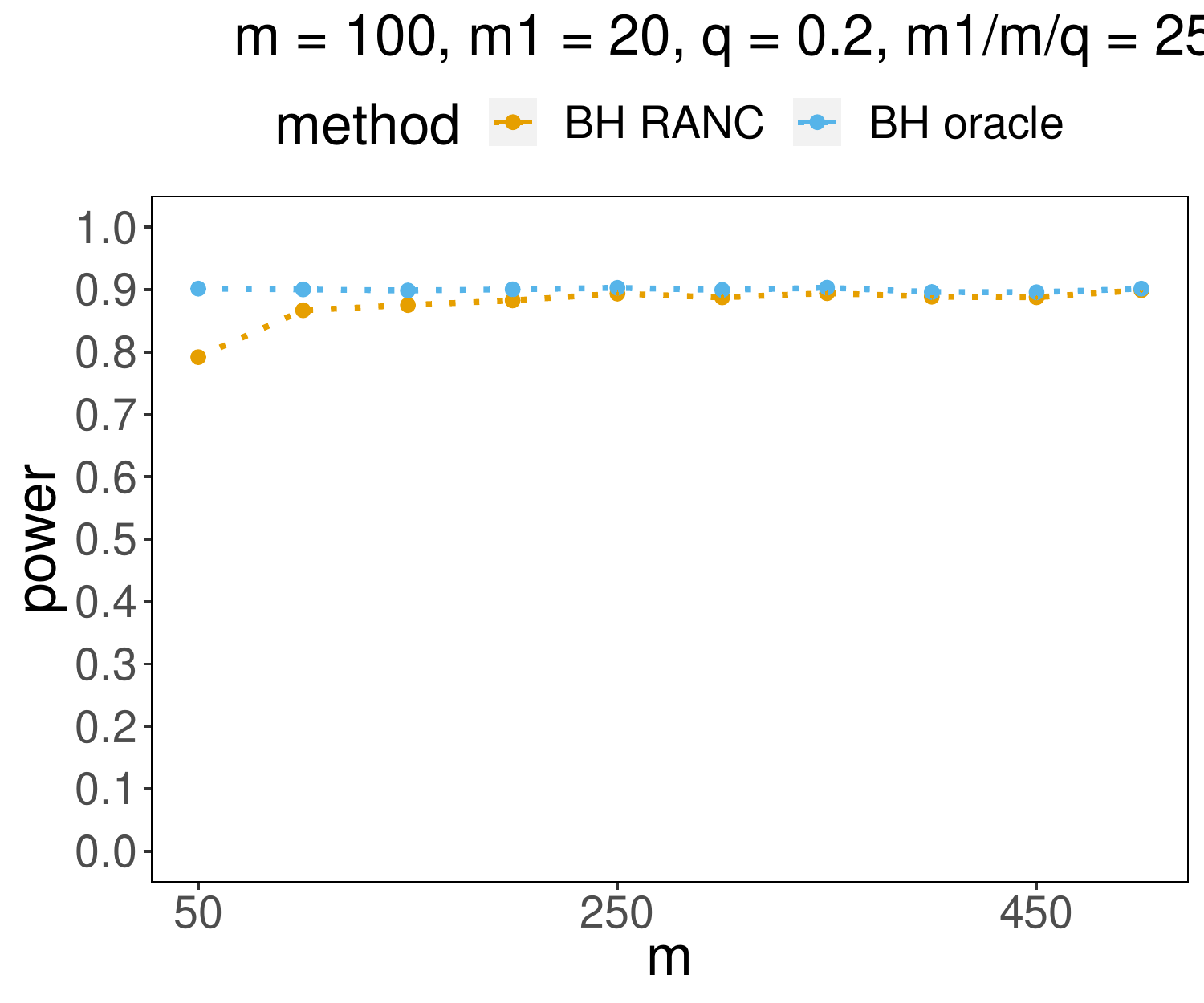}
    \subcaption{$\No = 100, \NoNonNull = 20, q = 0.2, \No/(\NoNonNull q) = 25$}
    \end{minipage}
        \begin{minipage}{6.7cm}
    \centering
\includegraphics[clip, trim = 0cm 0cm 0cm 1cm, width = 5.6cm]{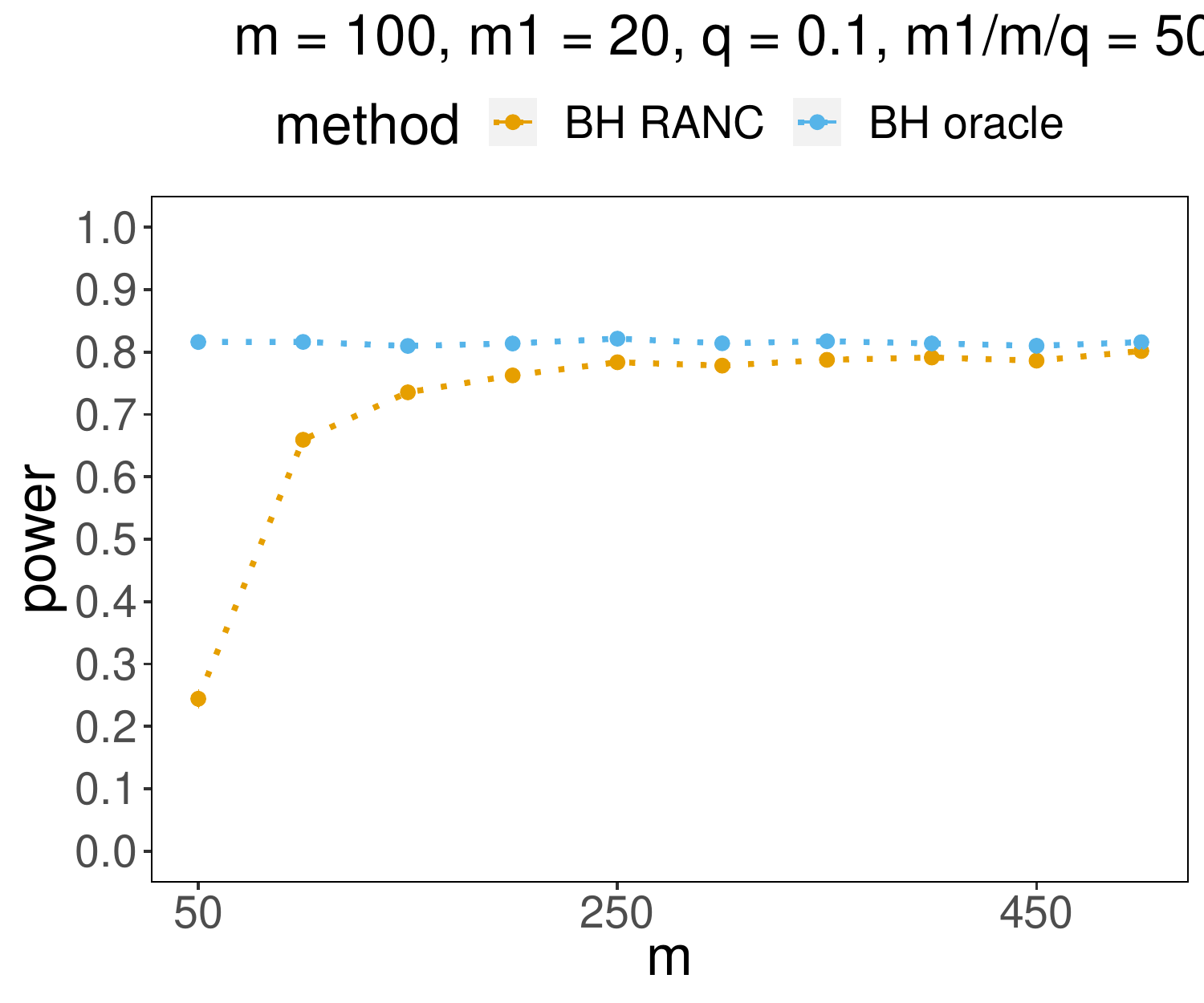}
    \subcaption{$\No = 100, \NoNonNull = 20, q = 0.1, \No/(\NoNonNull q) = 50$}
    \end{minipage}
    \begin{minipage}{6.7cm}
    \centering
\includegraphics[clip, trim = 0cm 0cm 0cm 1cm, width = 5.6cm]{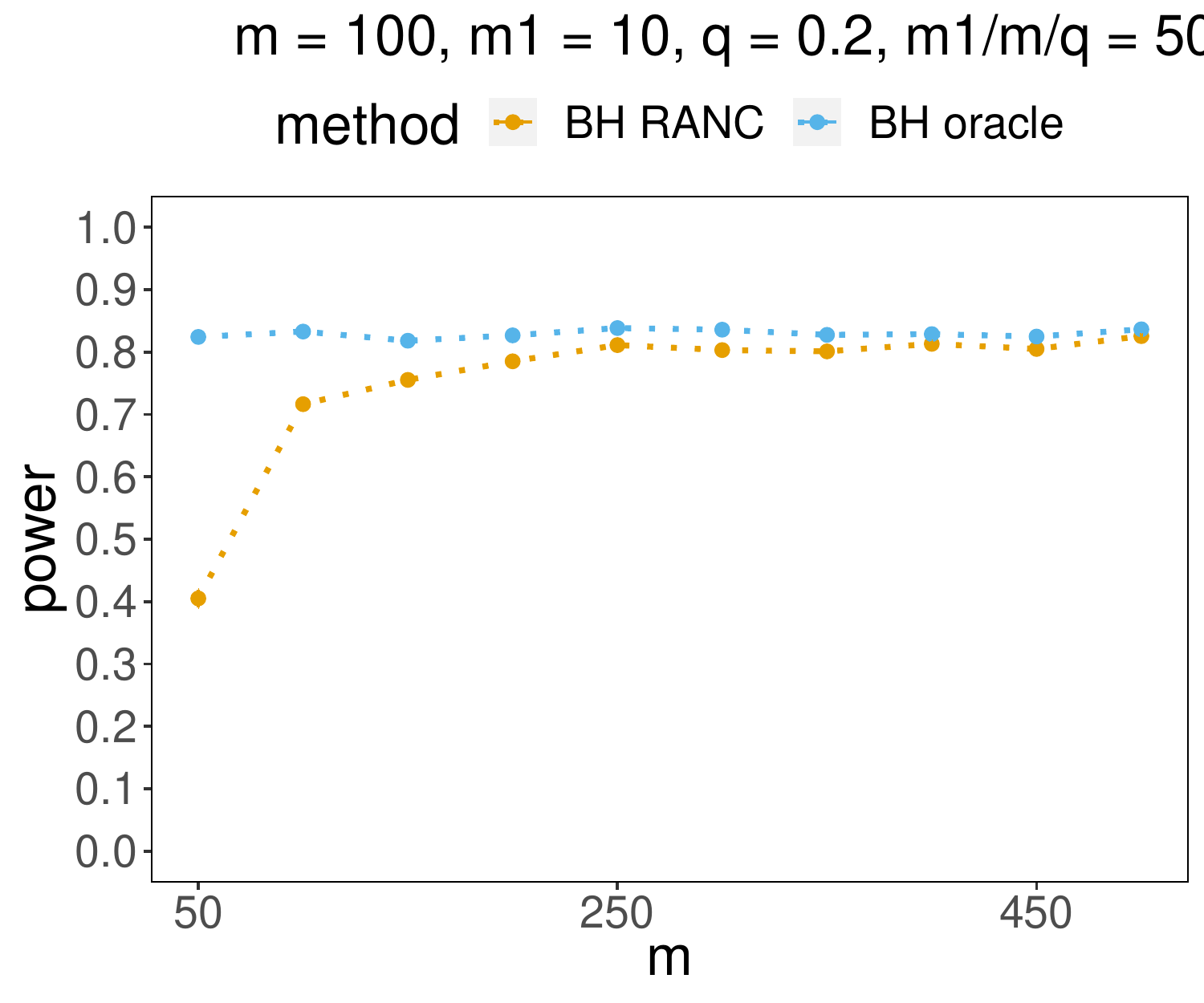}
    \subcaption{$\No = 100, \NoNonNull = 10, q = 0.2, \No/(\NoNonNull q) = 50$}
    \end{minipage}
            \begin{minipage}{6.7cm}
    \centering
\includegraphics[clip, trim = 0cm 0cm 0cm 1cm, width = 5.6cm]{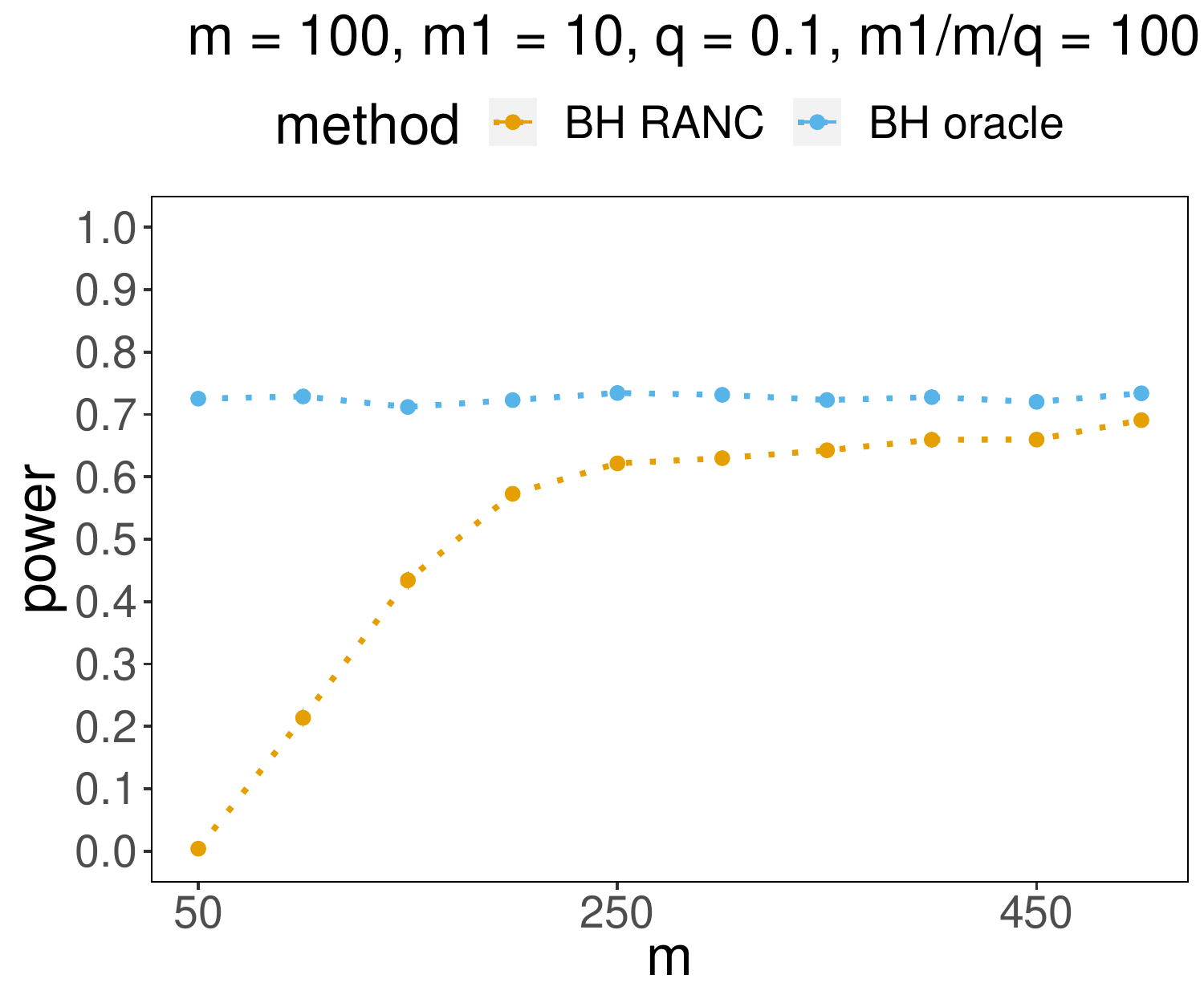}
    \subcaption{$\No = 100, \NoNonNull = 10, q = 0.1, \No/(\NoNonNull q) = 100$}
    \end{minipage}
    \begin{minipage}{6.7cm}
    \centering
\includegraphics[clip, trim = 0cm 0cm 0cm 1cm, width = 5.6cm]{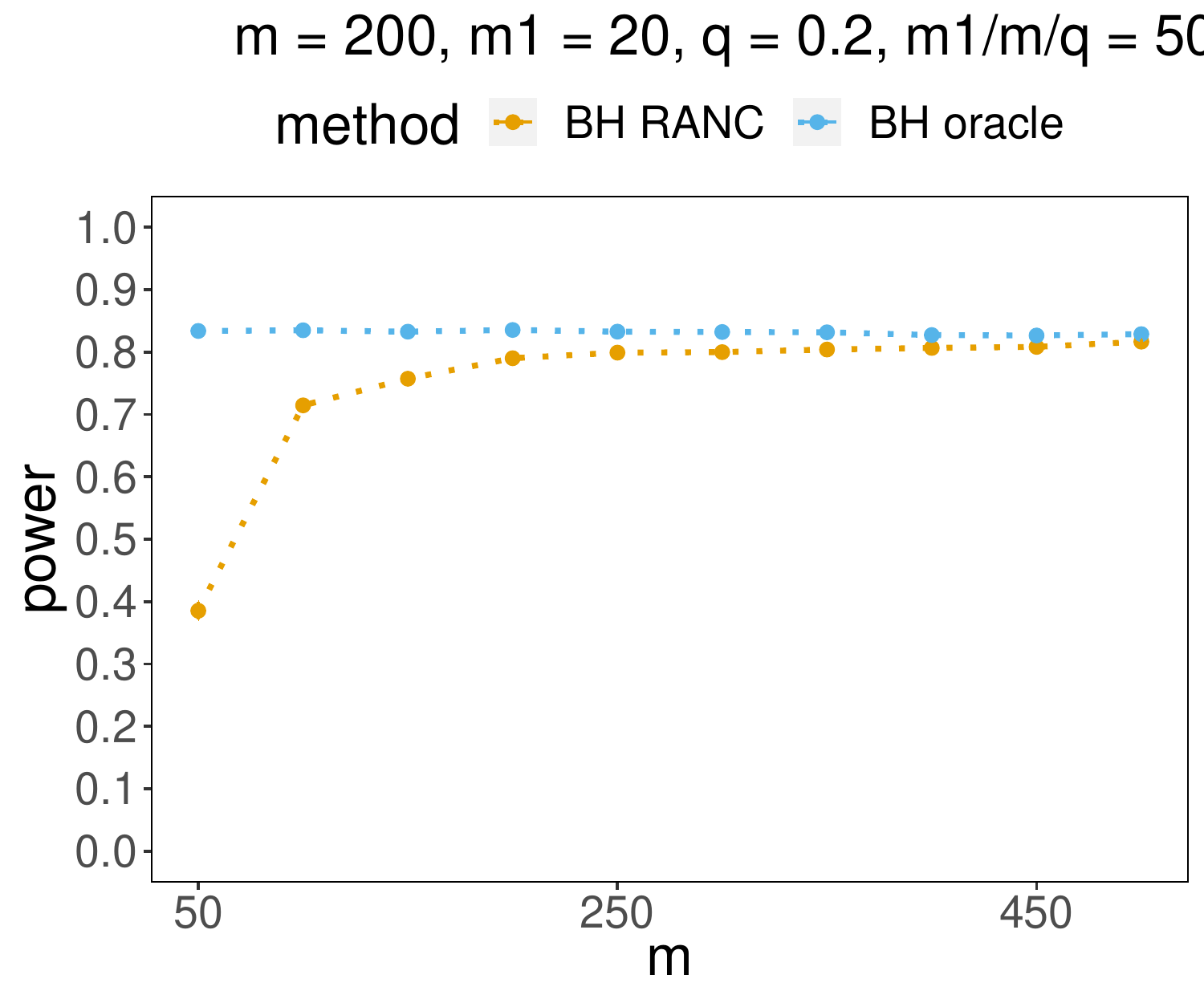}
    \subcaption{$\No = 200, \NoNonNull = 20, q = 0.2, \No/(\NoNonNull q) = 50$}
    \end{minipage}
            \begin{minipage}{6.7cm}
    \centering
\includegraphics[clip, trim = 0cm 0cm 0cm 1cm, width = 5.6cm]{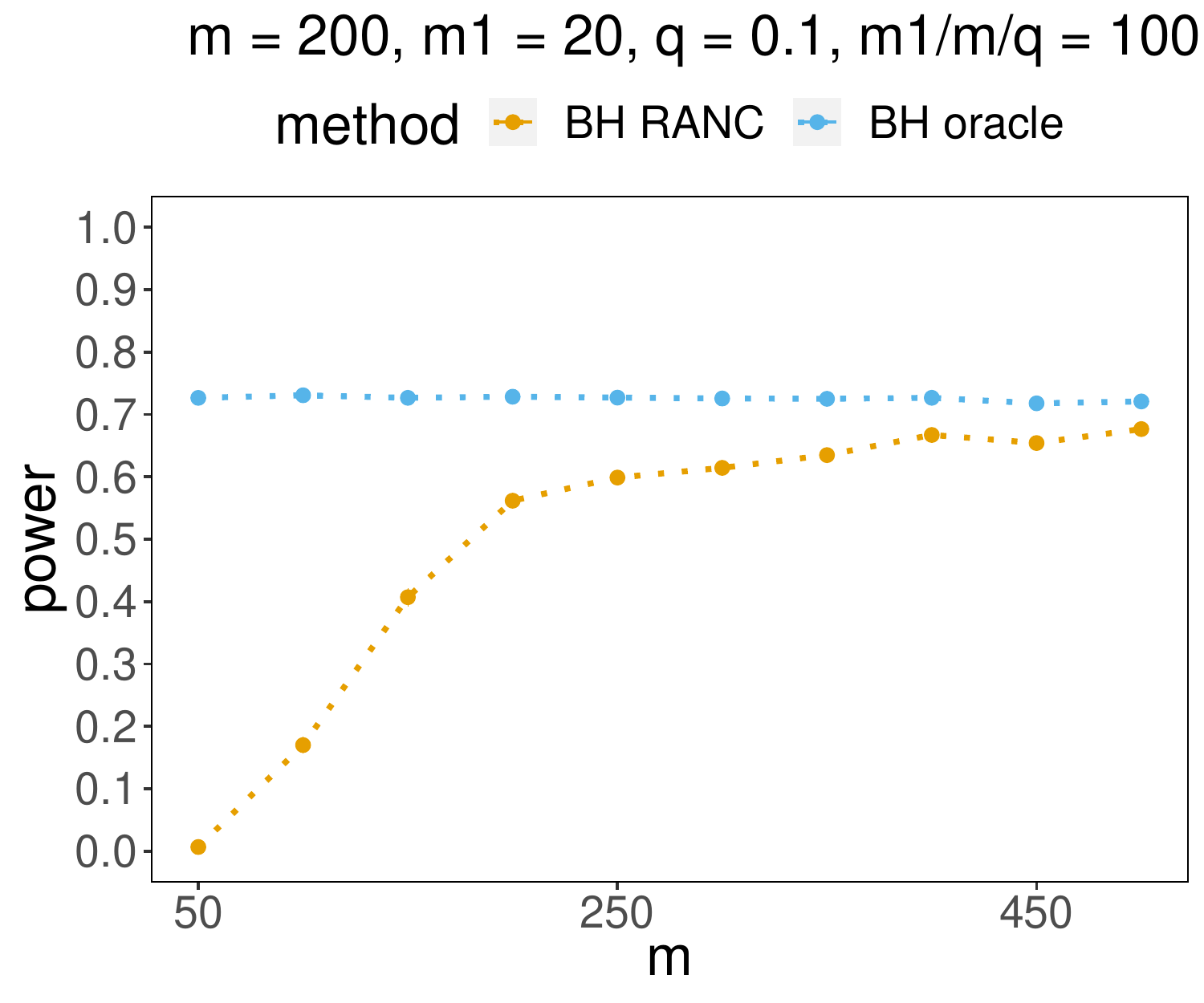}
    \subcaption{$\No = 200, \NoNonNull = 20, q = 0.1, \No/(\NoNonNull q) = 100$}
    \end{minipage}
    \begin{minipage}{6.7cm}
    \centering
\includegraphics[clip, trim = 0cm 0cm 0cm 1cm, width = 5.6cm]{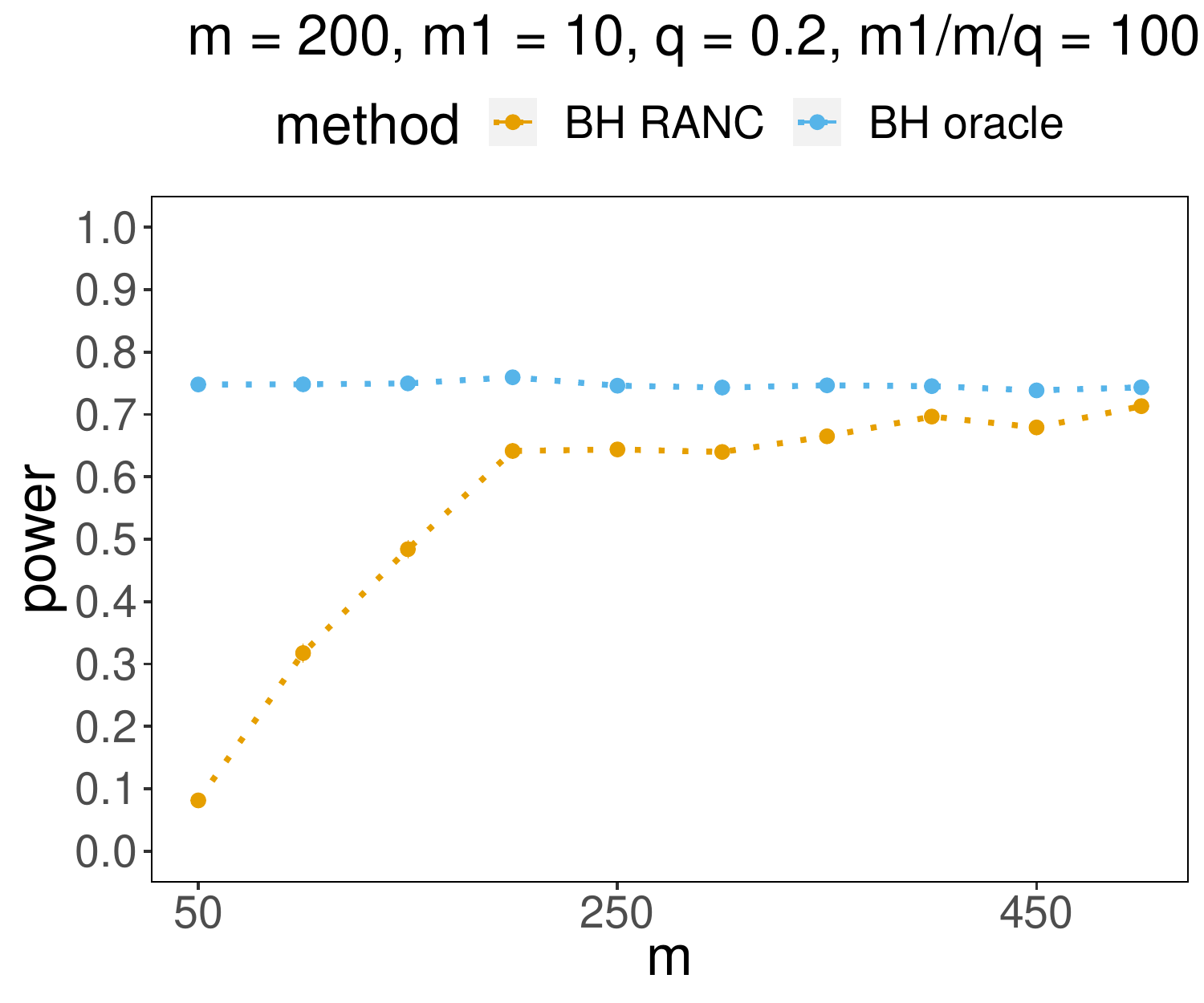}
    \subcaption{$\No = 200, \NoNonNull = 10, q = 0.2, \No/(\NoNonNull q) = 100$}
\end{minipage}
    \begin{minipage}{6.7cm}
    \centering
\includegraphics[clip, trim = 0cm 0cm 0cm 1cm, width = 5.6cm]{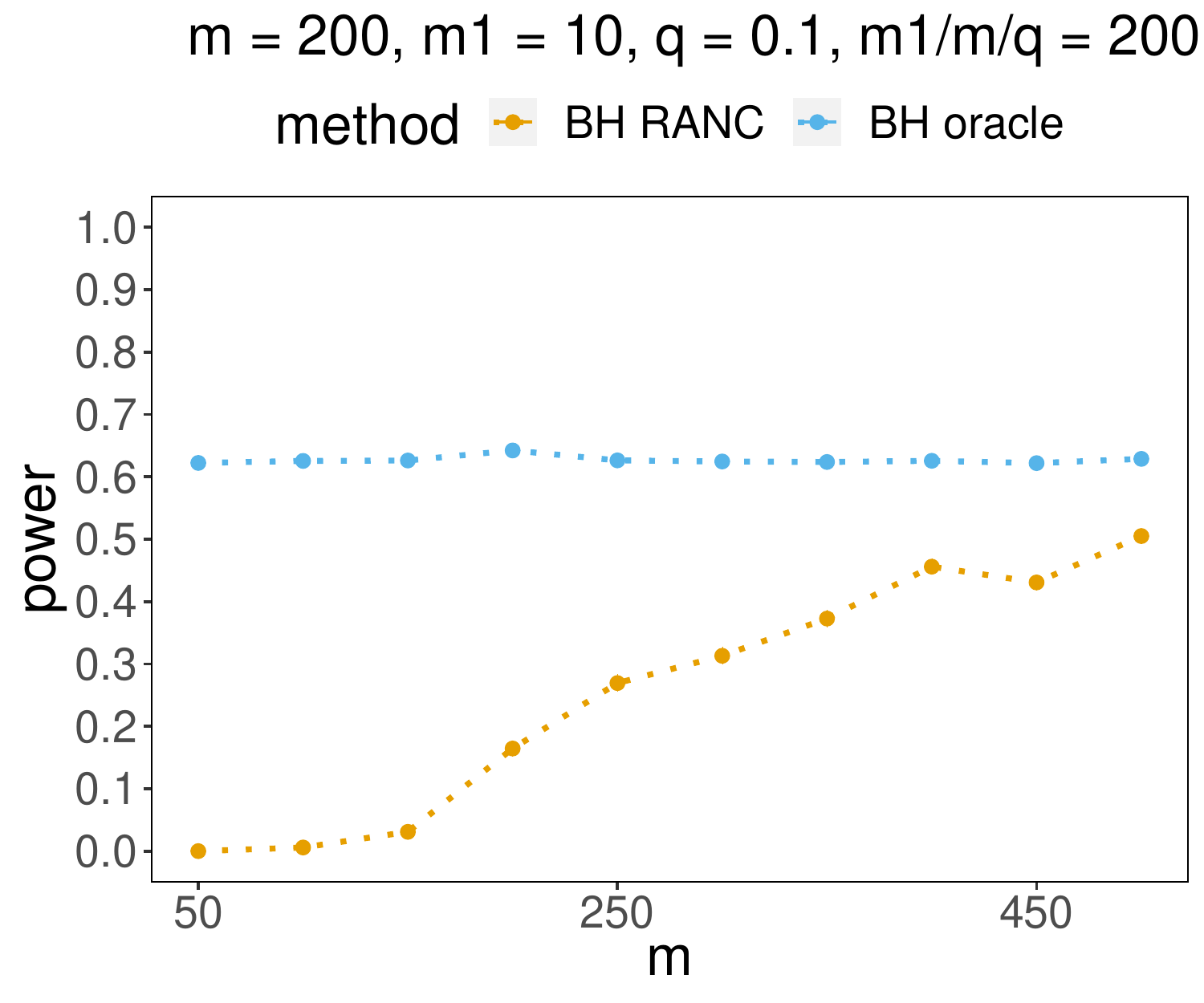}
    \subcaption{$\No = 200, \NoNonNull = 10, q = 0.1, \No/(\NoNonNull q) = 200$}
\end{minipage}
    \caption{{Power analysis of the \BH~\nickname~with different numbers of internal negative control units. We vary $\No \in \{100, 200\}$, $\NoNonNull  \in \{10, 20\}$, and $q \in \{0.1, 0.2\}$. The number of internal negative controls varies from $50$ to $500$. The results are aggregated over $10^3$ trials. }}
    \label{fig:nc.sample.size}
  \end{figure}

Next, we carry out this power analysis with weaker non-nulls.
In particular, we adopt the setting of independent p-values with exact
baseline p-values in \Cref{sec:simulation} but set the marginal
distribution of the non-nulls to $\Phi(Z - 2)$. We vary $\No \in
\{100, 200\}$, $\NoNonNull  \in \{10, 20\}$, and $q \in \{0.1,
0.2\}$. We increase the internal negative control sample size from
$50$ to $1000$. The results are reported in
\Cref{fig:nc.sample.size.threshold.weak}. In this setting, we find
that the \BH~\nickname~requires $\NoNc \ge
5 \No/(\FDRLevel \NoNonNull)$ internal negative controls to achieve
comparable power as \BH~oracle.

\begin{figure}[tbp]
    \centering
    \begin{minipage}{6.7cm}
    \centering
\includegraphics[clip, trim = 0cm 0cm 0cm 1cm, width = 5.6cm]{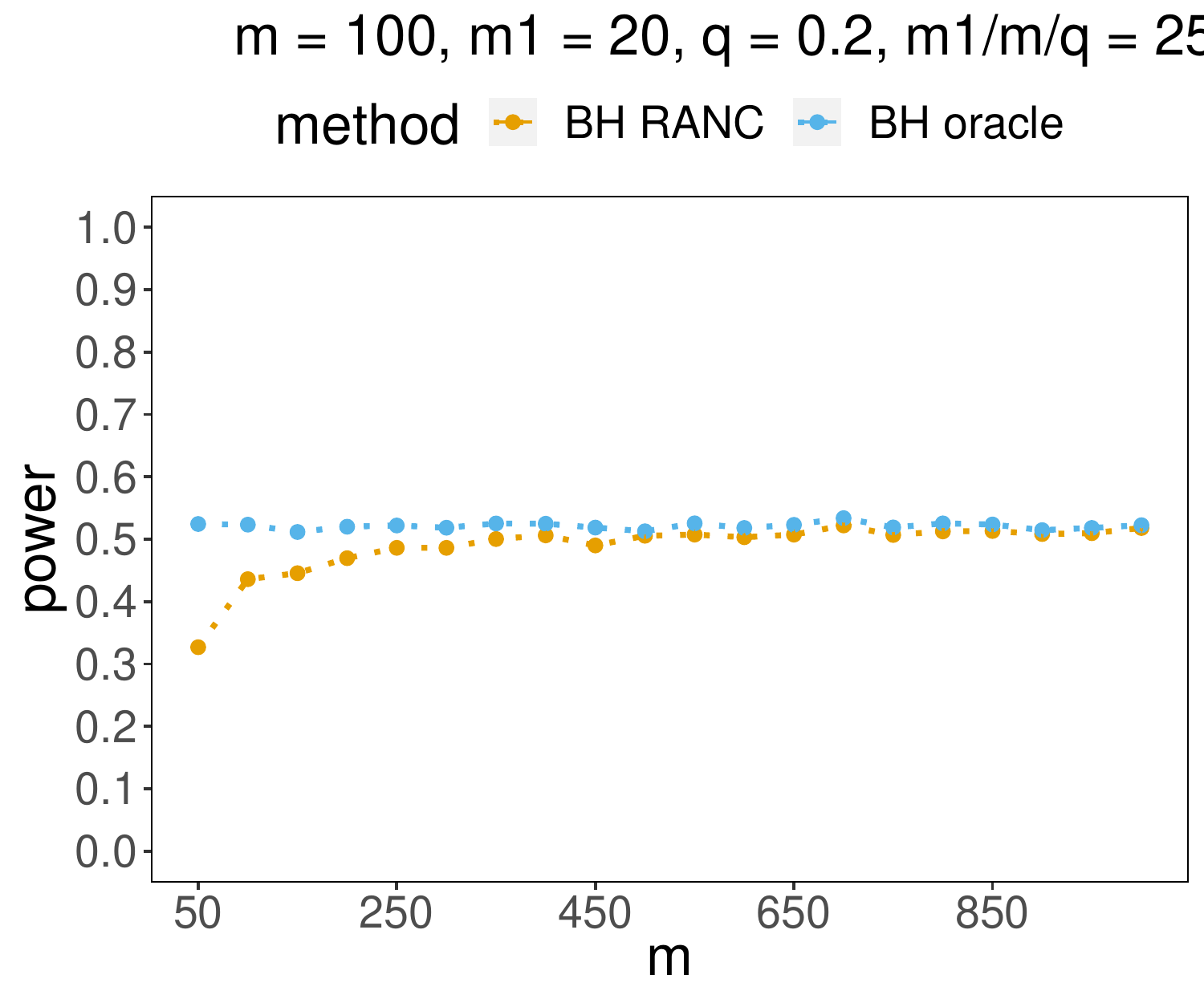}
    \subcaption{$\No = 100, \NoNonNull = 20, q = 0.2, \frac{\No}{\NoNonNull q} = 25$}
    \end{minipage}
    \begin{minipage}{6.7cm}
    \centering
\includegraphics[clip, trim = 0cm 0cm 0cm 1cm, width = 5.6cm]{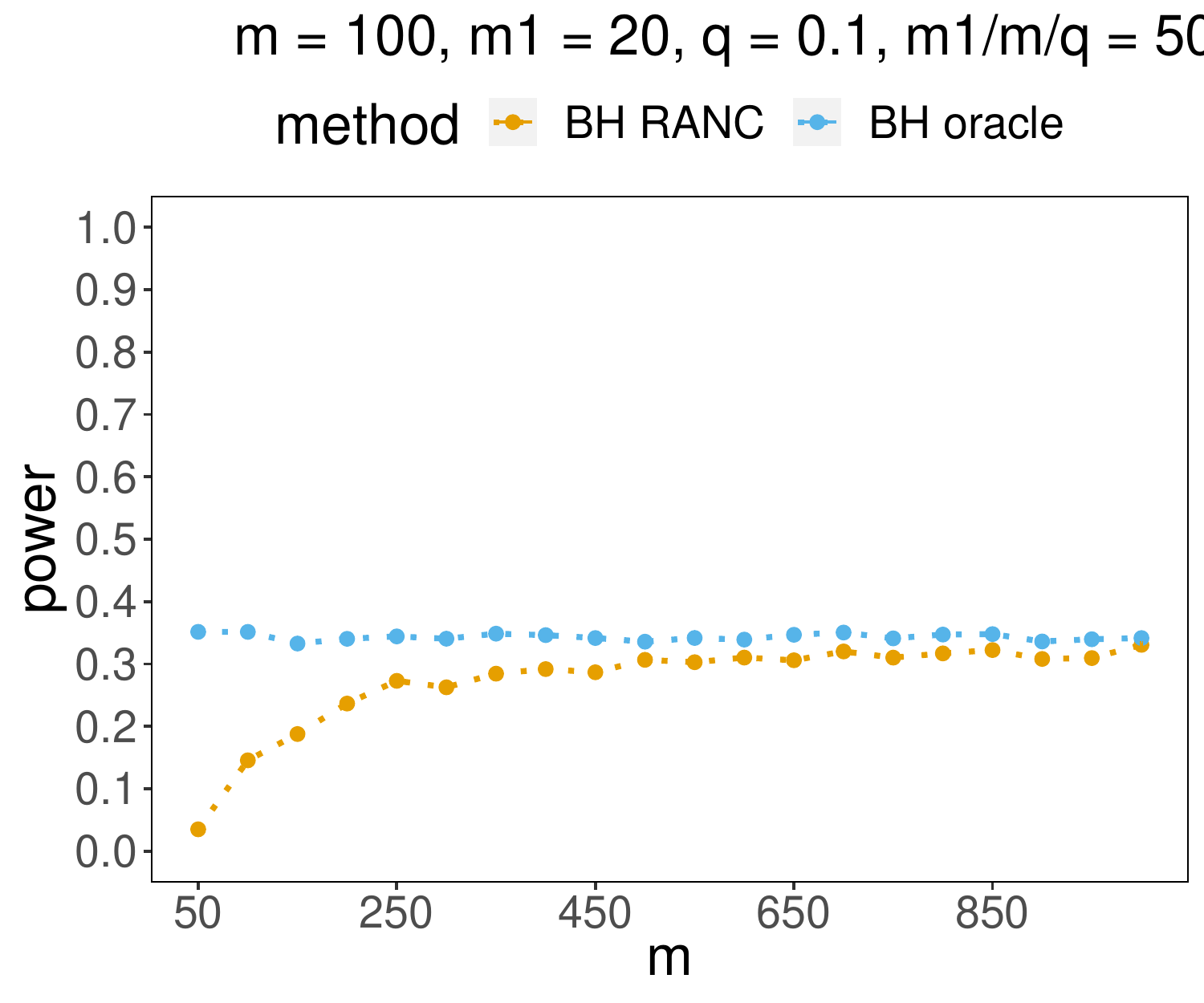}
    \subcaption{$\No = 100, \NoNonNull = 20, q = 0.1, \frac{\No}{\NoNonNull q} = 50$}
    \end{minipage}
    \begin{minipage}{6.7cm}
    \centering
\includegraphics[clip, trim = 0cm 0cm 0cm 1cm, width = 5.6cm]{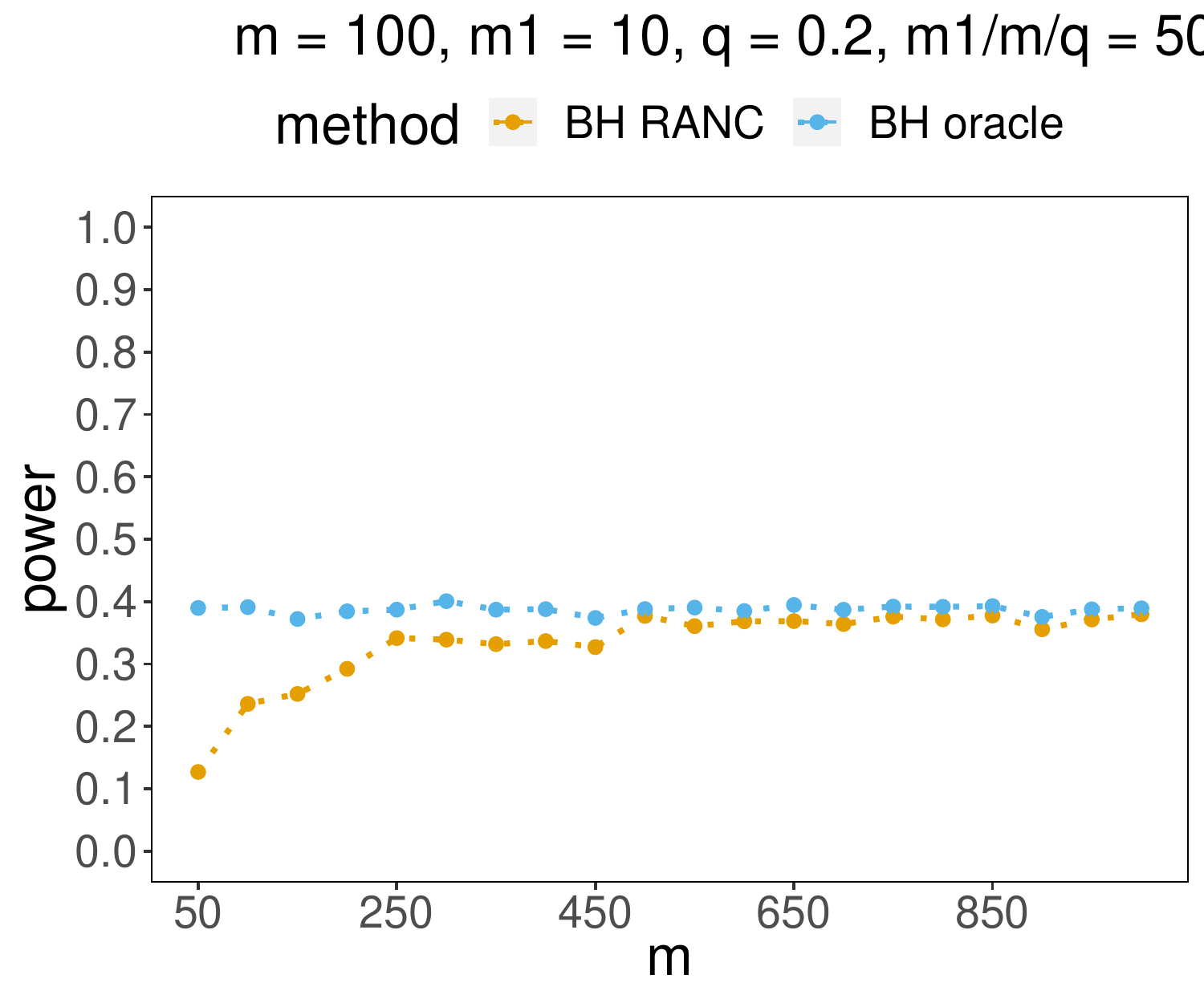}
    \subcaption{$\No = 100, \NoNonNull = 10, q = 0.2, \frac{\No}{\NoNonNull q} = 50$}
    \end{minipage}
    \begin{minipage}{6.7cm}
    \centering
\includegraphics[clip, trim = 0cm 0cm 0cm 1cm, width = 5.6cm]{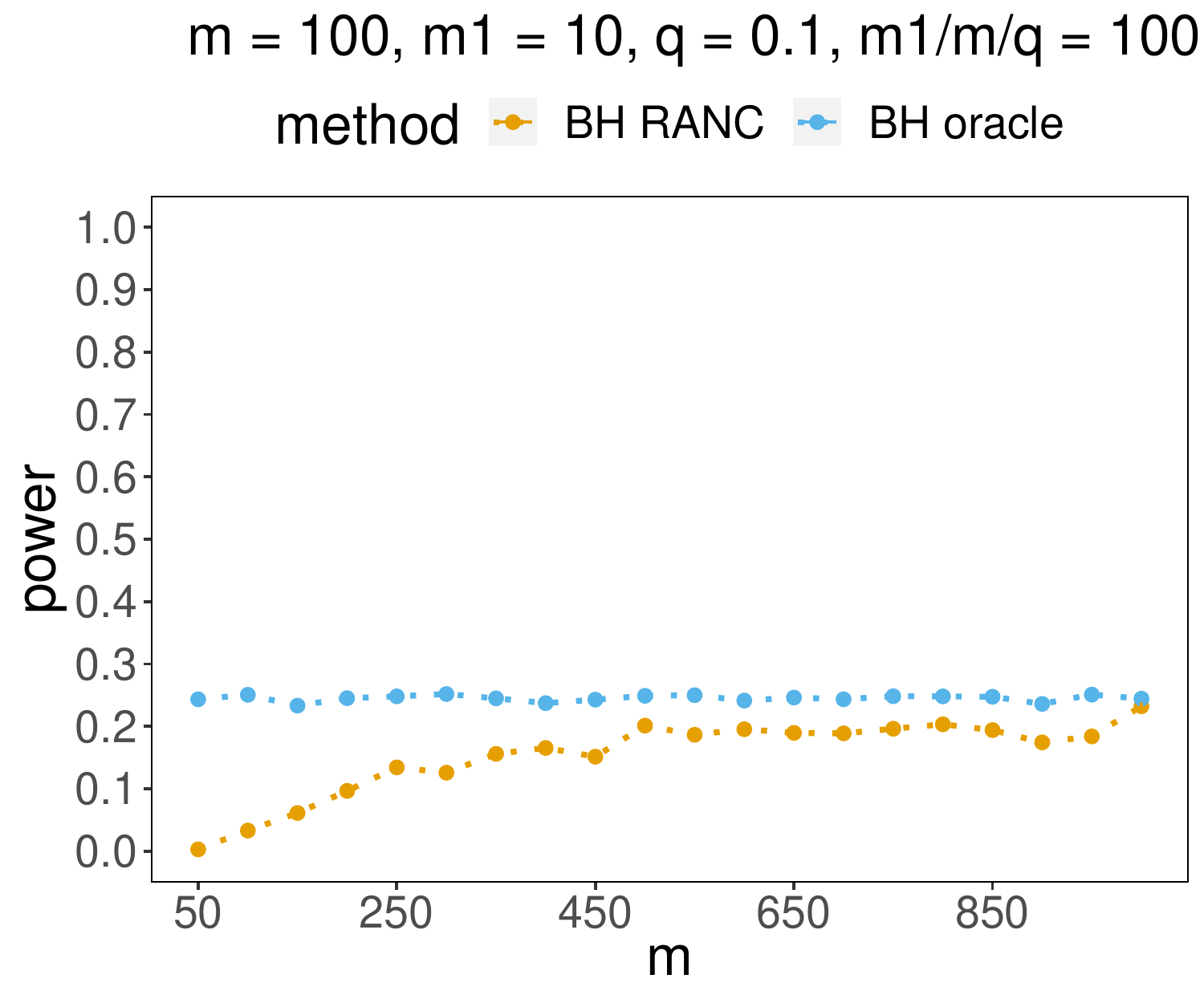}
    \subcaption{$\No = 100, \NoNonNull = 10, q = 0.1, \frac{\No}{\NoNonNull q} = 100$}
    \end{minipage}
    \begin{minipage}{6.7cm}
    \centering
\includegraphics[clip, trim = 0cm 0cm 0cm 1cm, width = 5.6cm]{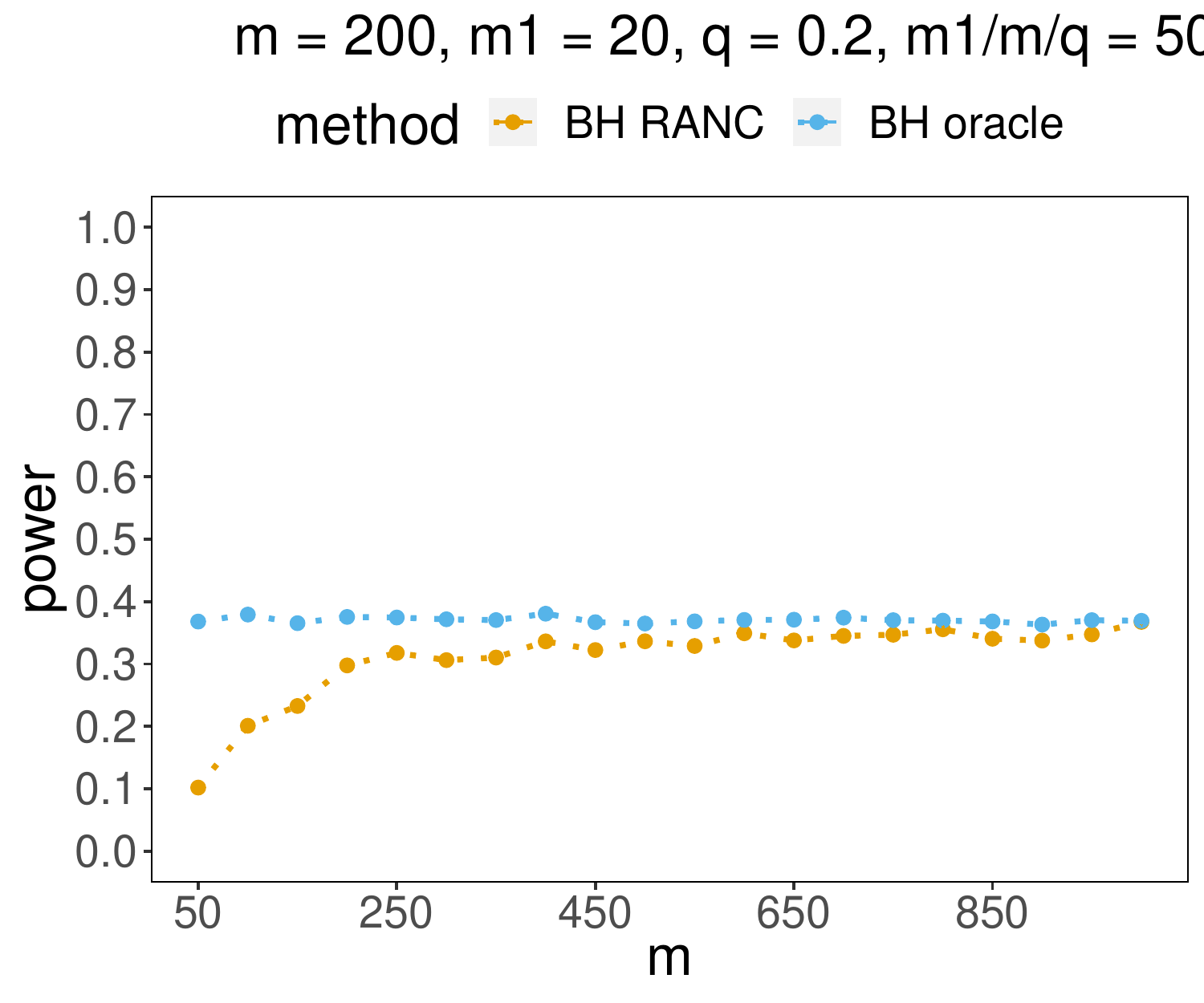}
    \subcaption{$\No = 200, \NoNonNull = 20, q = 0.2, \frac{\No}{\NoNonNull q} = 50$}
    \end{minipage}
            \begin{minipage}{6.7cm}
    \centering
\includegraphics[clip, trim = 0cm 0cm 0cm 1cm, width = 5.6cm]{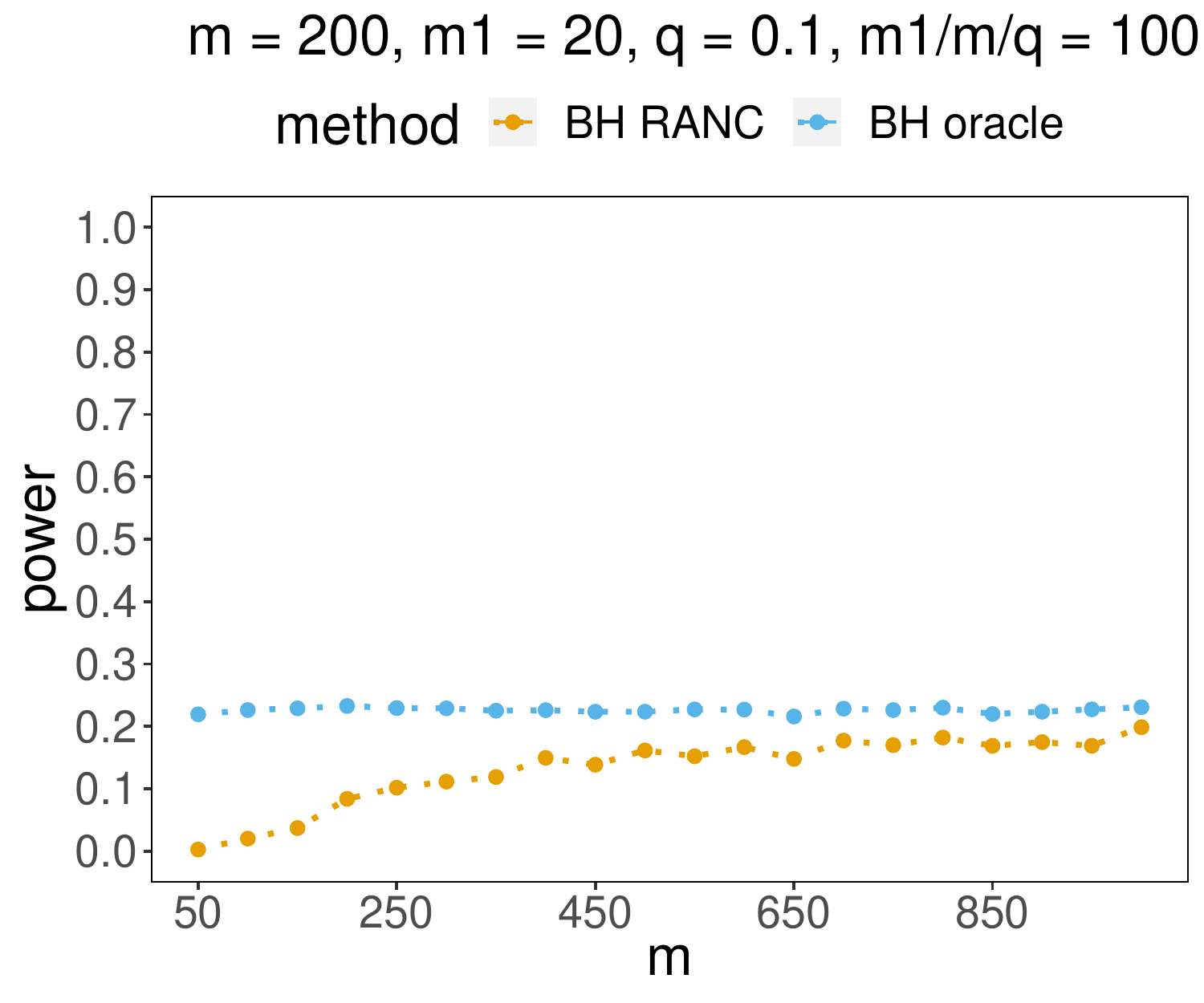}
    \subcaption{$\No = 200, \NoNonNull = 20, q = 0.1, \frac{\No}{\NoNonNull q} = 100$}
    \end{minipage}
    \begin{minipage}{6.7cm}
    \centering
\includegraphics[clip, trim = 0cm 0cm 0cm 1cm, width  = 5cm]{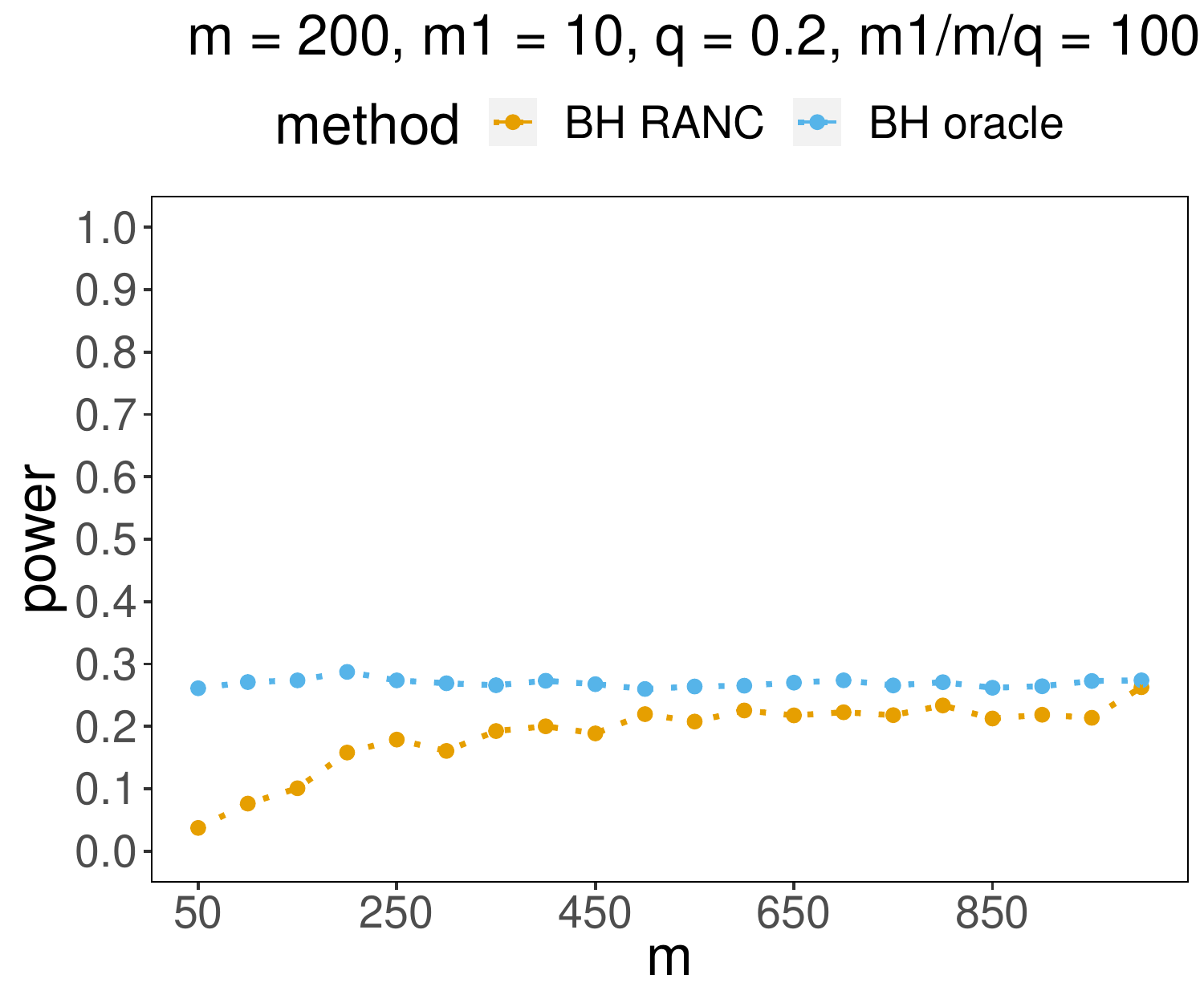}
    \subcaption{$\No = 200, \NoNonNull = 10, q = 0.2, \frac{\No}{\NoNonNull q} = 100$}
\end{minipage}
\begin{minipage}{6.7cm}
    \centering
\includegraphics[clip, trim = 0cm 0cm 0cm 1cm, width = 5.6cm]{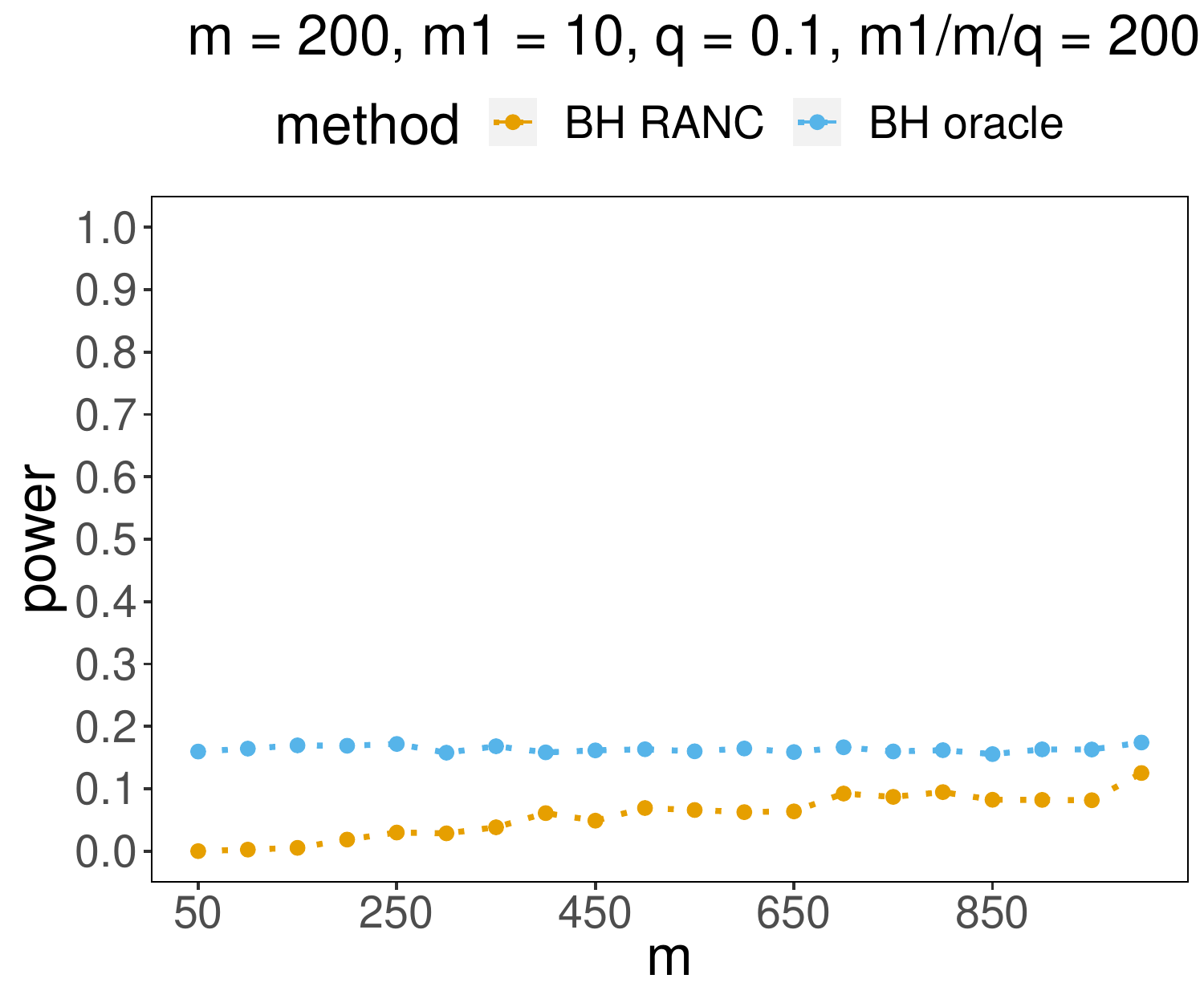}
    \subcaption{$\No = 200, \NoNonNull = 10, q = 0.1, \frac{\No}{\NoNonNull q} = 200$}
\end{minipage}
    \caption{{Power analysis of the \BH~\nickname~with weaker non-nulls and different numbers of internal negative control units. We vary $\No \in \{100, 200\}$, $\NoNonNull \in \{10, 20\}$, and $q \in \{0.1, 0.2\}$. The internal negative control sample size varies from $50$ to $1000$. The results are aggregated over $10^3$ trials.}}
    \label{fig:nc.sample.size.threshold.weak}
\end{figure}

\begin{figure}[h]
  \centering
  \includegraphics[width = \textwidth]{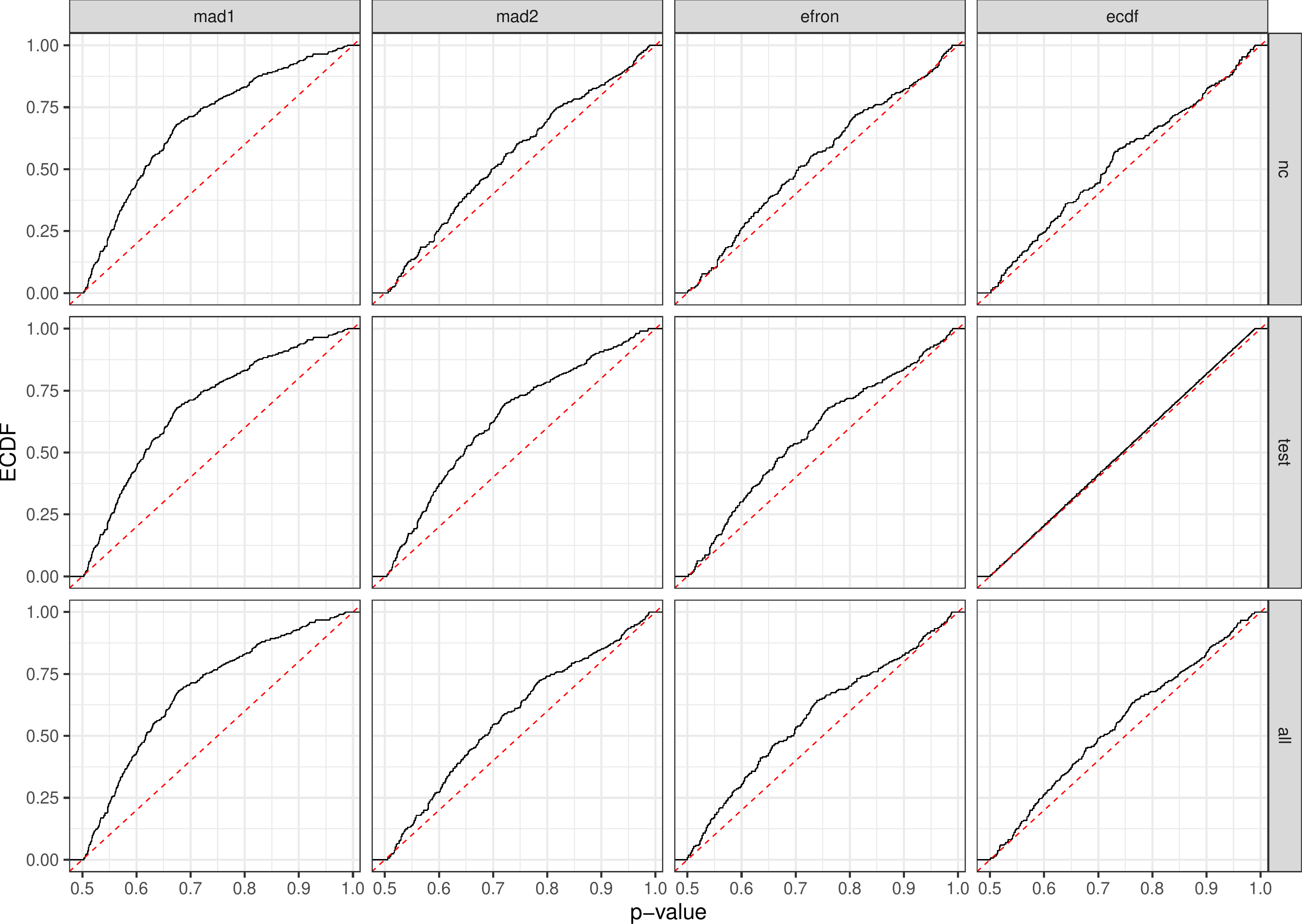}
  \caption{A comparison of the distribution of p-values for the
    proteomic application obtained from the empirical null
    distributions described in \Cref{sec:choice-empir-null}. It is
    expected that, when using the correct null distribution, p-values
    between 0.5 and 1 should be nearly uniformly distributed.}
  \label{fig:empirical-null}
\end{figure}

\begin{figure}[h]
  \centering
  \includegraphics[width = 0.5\textwidth]{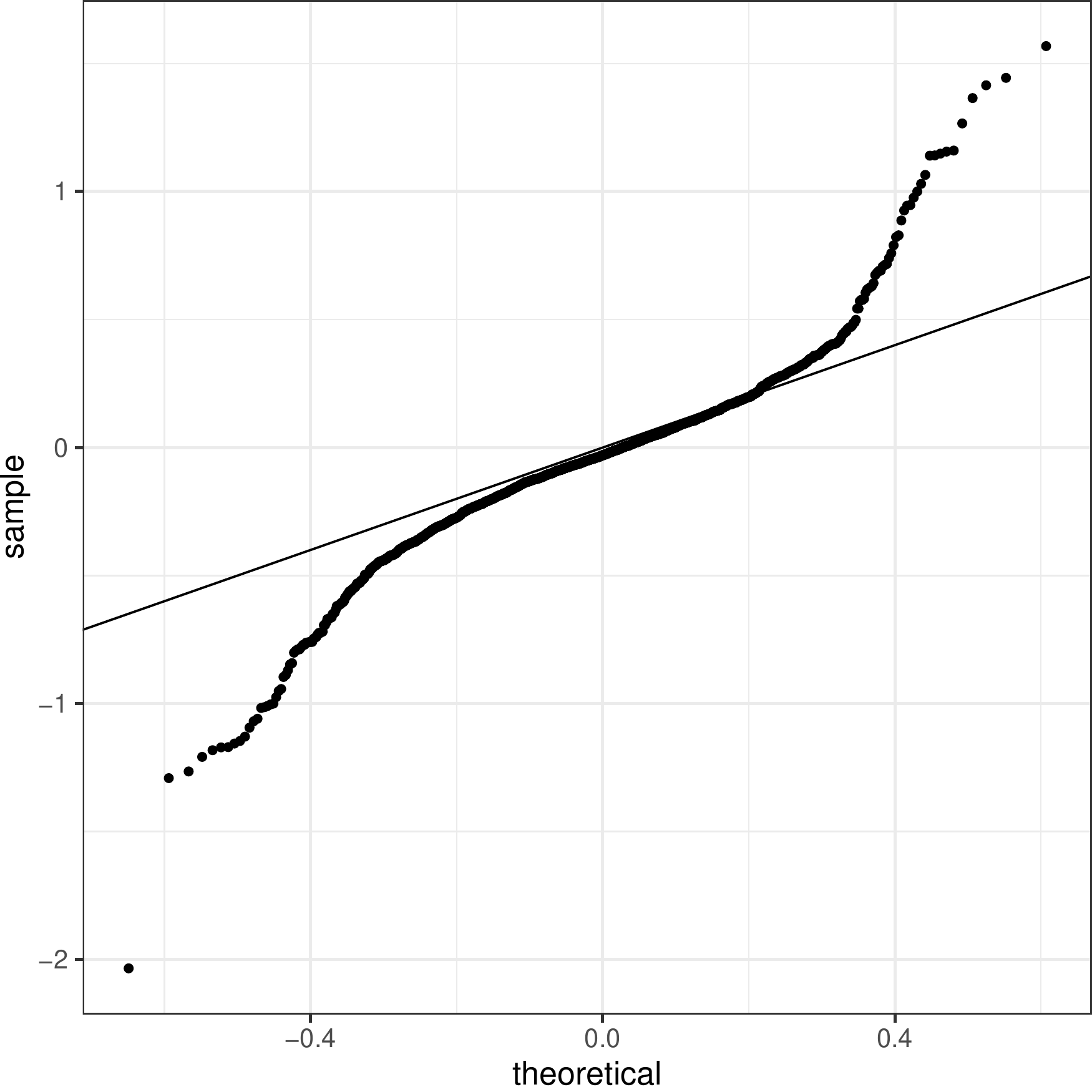}
  \caption{A quantile-quantile plot of the negative control statistics
    against the empirical null distribution $\text{N}(-0.02, 0.18)$
    obtained by Efron's method. The solid line corresponds to the diagonal.}
  \label{fig:qq-efron}
\end{figure}

\end{document}